%% file: main.tex
\documentclass[11pt]{article}

\input{head}

\title{
    Exactly 
    simulating stochastic chemical reaction networks\\in sub-constant time per reaction
}
\author[2]{Joshua Petrack\footnote{
    Supported by NSF awards 2211793 and 1844976.
}}
\author[1]{David Doty\footnote{
    Supported by NSF awards 2329909 and 2211793 and  DoE award DE-SC0024467.
}}
\affil[1]{University of California, Davis, \texttt{doty}@ucdavis.edu}
\affil[2]{University of California, Davis and Maynooth University, \texttt{Josh.Petrack@mu.ie}}
\date{}

\begin{document}

\maketitle
\thispagestyle{empty}

\input{sections/abstract}

\clearpage
\setcounter{page}{1}

\input{sections/intro}

\input{sections/definitions}
\input{sections/transformation}
\input{sections/simulation}

\input{sections/removing_fuel}
\input{sections/making_time_work}

\input{sections/data}

\input{sections/conclusion}

\bibliography{main}


\end{document}

%% file: head.tex
\usepackage[full]{optional}

\opt{full}{\usepackage[appendix=inline, repeqn=same]{apxproof}}
\opt{sub}{\usepackage[appendix=append]{apxproof}}
\opt{final}{\usepackage[appendix=strip]{apxproof}}

\usepackage[
disable,
textsize=scriptsize]{todonotes} 
\setuptodonotes{fancyline}

\newcommand{\eq}[1]{
    \opt{full}{\[#1\]}
    \opt{sub,final}{\ensuremath{#1}}
}

\usepackage{authblk}

\usepackage[utf8]{inputenc}
\usepackage{fullpage}
\usepackage{amsthm}
\usepackage{amsfonts}
\usepackage{amsmath}
\usepackage{graphicx}
\usepackage{svg}

\usepackage[colorlinks=true,citecolor=blue,linkcolor=blue]{hyperref}
\usepackage[capitalise]{cleveref}

\usepackage{enumitem}
\setlist{nosep}

\usepackage{graphicx} 

\usepackage{cite}

\usepackage{algorithmic}
\usepackage{algorithm}

\newcommand{\E}[1]{\mathrm{E}\left[ #1 \right]}

\usepackage{etoolbox} 

\usepackage{subcaption}

\newtheoremrep{thm}{Theorem}[section]
\newtheoremrep{lem}[thm]{Lemma}
\newtheoremrep{cor}[thm]{Corollary}
\newtheorem{defn}[thm]{Definition}

\newtheorem{obs}[thm]{Observation}

\Crefname{them}{Theorem}{Theorems}
\Crefname{lem}{Lemma}{Lemmas}
\Crefname{cor}{Corollary}{Corollaries}
\Crefname{defn}{Definition}{Definitions}
\Crefname{prop}{Proposition}{Propositions}
\Crefname{alg}{Algorithm}{Algorithms}

\renewcommand{\vec}[1]{\mathbf{#1}}
\newcommand{\vr}{\vec{r}}
\newcommand{\vp}{\vec{p}}
\newcommand{\vc}{\vec{c}}

\newcommand{\vd}{\vec{d}}
\newcommand{\cC}{\mathcal C}
\newcommand{\cE}{\mathcal E}

\usepackage{comment}

\def\longrightharpoonup{\relbar\joinrel\rightharpoonup}
\def\longleftharpoondown{\leftharpoondown\joinrel\relbar}
\def\longrightleftharpoons{\mathop{\vcenter{\hbox{\ooalign{\raise1pt\hbox{$\longrightharpoonup\joinrel$}\crcr\lower1pt\hbox{$\longleftharpoondown\joinrel$}}}}}}
\def\rxn{\mathop{\rightarrow}\limits}  

\def\revrxn{\mathop{\rightleftharpoons}\limits}

\newcommand{\N}{\mathbb{N}}
\newcommand{\R}{\mathbb{R}}

\DeclareMathOperator{\gen}{gen}
\DeclareMathOperator{\ord}{ord}
\newcommand{\ellstep}[4]{\mathcal F^{\mathrm {dis}}_{#1,#4}(#2, #3)}
\newcommand{\tstep}[4]{\mathcal F^{\mathrm {con}}_{#1,#4}(#2, #3)}
\newcommand{\propensity}[3]{p_{#1,#2}(#3)}
\newcommand{\totalpropensity}[2]{p^{\mathrm{tot}}_{#1,#2}}
\newcommand{\adjpropensity}[3]{\bar{p}_{#1,#2}(#3)}
\newcommand{\adjrateconstant}[4]{\bar{k}_{#1}(#2,#3,#4)}
\newcommand{\totalrateconstant}[2]{k^{\mathrm{tot}}_{#1}(#2)}
\newcommand{\schedulerslowdown}[3]{S_{#1,#2}(#3)}

\let\Pr\relax
\DeclareMathOperator{\Pr}{Pr}

\newcommand{\Var}[1]{\mathrm{Var}\left[ #1 \right]}

\usepackage{csquotes}

%% file: sections/abstract.tex
\begin{abstract}
The model of \emph{chemical reaction networks} is among the oldest and most widely studied and used in natural science~\cite{guldberg1864studies}.
The model describes reactions among abstract chemical species,
for instance
$A + B \to C$, 
which indicates that if a molecule of type $A$ interacts with a molecule of type $B$
(the \emph{reactants}), 
they may stick together to form a molecule of type $C$
(the \emph{product}).
The standard algorithm for simulating (discrete, stochastic) chemical reaction networks is the \emph{Gillespie algorithm}~\cite{gillespie1977exact},
which stochastically simulates one reaction at a time,
so to simulate $\ell$ consecutive reactions,
it requires total running time $\Omega(\ell)$.

We give the first chemical reaction network stochastic simulation algorithm that can simulate $\ell$ reactions,
provably preserving the exact stochastic dynamics
(sampling from precisely the same distribution as the Gillespie algorithm),
yet using time provably sublinear in $\ell$.
Under reasonable assumptions,
our algorithm can simulate $\ell$ reactions among $n$ total molecules in expected time $\tilde{O}(\ell/\sqrt n)$ when $\ell \ge n^{5/4}$, and in expected time $\tilde{O}(\ell/n^{2/5})$ when $n \le \ell \le n^{5/4}$.
Our work adapts an algorithm of Berenbrink, Hammer, Kaaser, Meyer, Penschuck, and Tran~\cite{berenbrink2020simulating} for simulating the distributed computing model known as \emph{population protocols}, extending it (in a very nontrivial way) to the more general chemical reaction network setting.

We provide an implementation of our algorithm as a Python package,
with the core logic implemented in Rust,
with remarkably fast performance in practice.
\end{abstract}

%% file: sections/intro.tex
\section{Introduction}

The model of \emph{chemical reaction networks} (\emph{CRNs}) is one of the oldest and most widely used in natural science~\cite{guldberg1864studies,waage1986studies}.
It describes reactions between abstract species, 
for example $A + B \rxn^k C + 2D$, 
which indicates that if a molecule of type $A$ 
(a.k.a., \emph{species} $A$) 
collides with a molecule of type $B$
($A$ and $B$ are the \emph{reactants}), 
they may split into three molecules, one of type $C$ and two of type $D$
(the \emph{products}).
The number $k \in \R_{>0}$ is called a \emph{rate constant}.
The standard algorithm for stochastically simulating CRNs is the \emph{Gillespie algorithm}~\cite{gillespie1977exact},
which we now describe.

The Gillespie algorithm iteratively executes one reaction at a time,
taking time $\Theta(\ell)$ for $\ell$ reactions.
A \emph{configuration} is a multiset of species;
equivalently a nonnegative integer vector $\vc$,
where $\vc(S) \in \N$ denotes the \emph{count} of $S$;
when $\vc$ is clear from context,
we write $\#S$ to denote $\vc(S)$.
A reaction is more likely to occur the more of its reactants there are;
more precisely,
given a fixed \emph{volume} $v \in \R_{>0}$,
the \emph{propensity} (a.k.a., \emph{rate})
of a reaction is 
proportional to 
the product of reactant counts and rate constant $k$, divided by $v^{r-1}$, where $r$ is the number of reactants.
For example the propensity of the reaction $A+B \rxn^k C+2D$ is $\frac{k \cdot \# A \cdot \# B}{v}.$
(See \Cref{sec:prelim} for a detailed definition.)
The Gillespie algorithm repeatedly samples a reaction to execute
(for example executing the above reaction would decrease $\#A,\#B$ by $1$, increase $\#C$ by $1$,
and increase $\#D$ by $2$),
with probability proportional to its propensity.
The time until this occurs is an exponential random variable with rate equal to the sum of all reaction propensities.
In other words
the model is a continuous-time Markov chain whose states are CRN configurations,
with transitions determined by reactions.

A distributed computing model known as \emph{population protocols}~\cite{AngluinADFP2006},
originally defined to model mobile finite-state sensor networks,
turns out to be equivalent to the special case of CRNs in which every reaction has exactly two reactants, two products, and rate constant $k=1$.
In the population protocol model, a scheduler repeatedly picks a pair of distinct molecules $A,B$ uniformly at random from a population of $n$ molecules and has them \emph{interact},
which means replacing them with $C$ and $D$ if there is a reaction $A+B \rxn C+D$,
or doing nothing if there is no such reaction 
(i.e., simulating the ``null'' reaction $A+B \rxn A+B$).
In the natural timescale in which we expect each molecule to have $O(1)$ interactions per unit of time,
the (discrete) \emph{parallel time} in a population protocol is defined as the number of interactions to occur, divided by the population size $n$.
A continuous time variant~\cite{fanti2020communication} gives each molecule a rate-1 Poisson clock, upon which it reacts
with a randomly chosen other molecule. 
The time until the next interaction is an exponential random variable with expected value $1/n$, so these two models are
equivalent up to a
re-scaling of time, which by straightforward Chernoff bounds is negligible.
The continuous-time variant turns out to coincide \emph{precisely} with the Gillespie model
restricted to population protocols, when the volume $v=n$.

The na\"{i}ve simulation algorithm for population protocols implements the model straightforwardly,
one (potentially null) reaction at a time,
taking time $\Theta(\ell)$ to simulate $\ell$ reactions.

\subsection{Related work on faster stochastic simulation of CRNs}
One ``speedup heuristic'' changes the model entirely:
the \emph{deterministic} or \emph{continuous mass-action} model of CRNs represents the amount of each species as a nonnegative real-valued \emph{concentration} (average count per unit volume),
defining polynomial ordinary differential equations (ODEs) with a unique trajectory~\cite{feinberg2019foundations}.
Integrating the ODEs using standard numerical methods~\cite{press2007section17} is typically much faster than stochastic simulation.
There is a technical sense in which this model is the large-count, large-volume limit of the discrete stochastic model~\cite{kurtz1972relationship}, but the ODE model removes any stochastic effects and can have vastly different behavior than the stochastic model~\cite[Fig.~3]{ppsim}.

The obvious way to pick the next reaction to occur in the Gillespie algorithm, if there are $R$ types of reactions, uses time $O(R)$.
There are variants of  Gillespie that reduce the time to apply a single
reaction from $O(R)$ to $O(\log R)$~\cite{gibson2000efficient} or even $O(1)$~\cite{slepoy2008constant,thanh2017efficient}.
However, the time to apply $\ell$ total reactions still scales linearly with $\ell$.
Other exact methods have been developed that, for some CRNs, empirically appear to simulate $\ell$ reactions in time $o(\ell)$~\cite{cai2009efficient,mjolsness2009exact}, but none have been proven rigorously to give an asymptotic speedup on all CRNs while maintaining exactness.
Linear noise approximation~\cite{cardelli2016stochastic} is another workaround, 
adding stochastic noise to an ODE trajectory.

Other methods give a provable and practical speedup,
but they also provably sacrifice exactness, i.e., the distribution of trajectories sampled is not the same as the Gillespie algorithm.
A common speedup heuristic for simulating $\omega(1)$ reactions in $O(1)$ time is \emph{$\tau$-leaping}~\cite{gillespie2001approximate, rathinam2007reversible, cao2006efficient, gillespie2007stochastic, soloveichik2009robust},
which ``leaps'' ahead by time $\tau$, by assuming
reaction propensities will not change
and updating counts in a single batch step by sampling according these propensities.
Such methods necessarily approximate the kinetics inexactly, though it is sometimes possible to prove bounds on the approximation accuracy~\cite{soloveichik2009robust}.
Nevertheless, there are CRNs with stochastic effects not observed in ODE simulation,
and where $\tau$-leaping introduces systematic inaccuracies that disrupt the fundamental qualitative behavior of the system,
demonstrating the need for exact stochastic simulation even with large population sizes;
see~\cite{ppsim,lathrop2020population} for examples.

A speedup idea for population protocol simulation is to sample the number of each reaction that would result from a random matching of size $\ell$
(i.e., $\ell$ pairs of molecules, all disjoint), and update species counts in a single step.
This heuristic, too, is inexact:
unlike the true process, it prevents any molecule from participating in more than one of the next $\ell$ reactions.

However,
based on this last heuristic,
Berenbrink, Hammer, Kaaser, Meyer, Penschuck, and Tran~\cite{berenbrink2020simulating} created a remarkable ``batching'' algorithm that is quadratically faster than the na\"{i}ve population protocol simulation algorithm, 
yet samples from the same exact distribution.
It can simulate,
on population size $n$,
batches of $\Theta(\sqrt{n})$ reactions in time $O(\log n)$ each.
Conditioned on the event that no molecule is picked twice during the next $\ell$ interactions, these interacting pairs are a random disjoint matching of the molecules.
Define the random variable $\mathbf{L}$ as the number of interactions 
until the same molecule is picked twice.
The algorithm samples this ``collision'' length $\mathbf{L}$, 
then updates counts assuming all pairs of interacting molecules are disjoint until this collision,
and finally simulates the interaction involving the collision.
By the Birthday Paradox, $\E{\mathbf{L}} \approx \sqrt{n}$ in a population of $n$ molecules,
giving a quadratic factor speedup over the na\"{i}ve algorithm.
(The time to update a batch scales quadratically with the total number of species.)

The batching algorithm extends straightforwardly to a generalized model of population protocols in which,
for some $o \in \N^+$,
all reactions have exactly $o$ reactants and $o$ products
(we call these \emph{uniformly conservative} CRNs; see \Cref{sec:crn-new-defns}).
However, the algorithm's correctness relies crucially on the fact that the population size $n$ never changes.
General CRNs may increase or decrease $n$ in different reactions.

There are several challenges to generalizing \cite{berenbrink2020simulating} to arbitrary CRNs.
Many of these arise because the complexity of CRNs makes it difficult to reason about a batch of reactions as a whole,
without knowing which individual reactions comprise it.
It is not clear how to even define $\mathbf{L}$ usefully for general CRNs with heterogeneous reaction sizes.
Another surprisingly difficult challenge is generalizing to continuous time,
i.e., drawing from the distribution of configurations after exactly time $t$ has elapsed.
In population protocols, 
the time until each interaction has the same exact (exponential) distribution.
This makes it easy to sample,
in constant time,
the elapsed time until $\ell$ interactions occur.
In CRNs, individual reaction times may have different distributions.
The time distribution to run a fixed batch of $\ell$ reactions may even depend on the order they occur in,
which batching cannot efficiently specify.

\subsection{Our contribution}
We extend the batching algorithm of Berenbrink, Hammer, Kaaser, Meyer, Penschuck, and Tran~\cite{berenbrink2020simulating}
to the fully general model of CRNs.
This is, to our knowledge,
the first simulation algorithm for the CRN model
provably preserving the exact stochastic dynamics,
yet provably simulating $\ell$ reactions in $o(\ell)$ expected time.
We implement our algorithm as a Python package,
with remarkably fast performance in practice.
See \Cref{sec:data} for simulation data generated using this tool.

Our algorithm works on arbitrary CRNs and gives similar asymptotic speedup with respect to population size to the original batching algorithm for population protocols~\cite{berenbrink2020simulating}.
Our main theorem, \cref{thm:maintheorem}, describes this speedup in detail.

There is an important efficiency difference between the Gillespie algorithm and the batching algorithms (both ours and~\cite{berenbrink2020simulating}): 
the Gillespie algorithm
automatically skips null reactions.
For example, a reaction such as $L+L \rxn L + F$ in volume $n$
(simulating leader election), 
when $\#L = 2$
and $\#F = n - 2$,
is much more efficient in the Gillespie algorithm, 
which runs in time $O(1)$.
A na\"{i}ve
population protocol simulation,
on the other hand,
iterates through $\Theta(n)$ expected null interactions 
($L + F \rxn L + F$
and $F + F \rxn F + F$) until the two $L$’s react.
Even the batching algorithm, since it explicitly simulates null reactions,
takes time $\Theta(\sqrt{n})$,
in this regime still asymptotically (and practically) much slower than the Gillespie algorithm.
For this reason,
any practical implementation should switch to the Gillespie algorithm when most sampled reactions in a batch are null.
We formalize this slowdown in \cref{def:efficiency-condition}.

\subsection{High-level overview of algorithm}
\todo{DD: I think toehold occlusion is a good thing to model with this algorithm}
The key insight to our extension of the batching algorithm is that we modify the simulated CRN so that
every reaction in the modified CRN increases the population size by the same amount (the \emph{generativity} of the reaction)
and has the same number of reactants (the \emph{order} of the reaction).
The full transformation is described in \Cref{sec:uniform-transformation}.
A simple example is the reversible reaction $A \revrxn^1_1 B+C$;
applying the transformation of \Cref{sec:uniform-transformation} in volume $v$ results in the CRN,
with two new species $K$ and $W$, with reactions $A+K \rxn^{v/\#K} B+C+K$\footnote{Note $K$ is a \emph{catalyst}, both a reactant and product, so $\#K$ never changes; we set $\#K=n$,
the molecular count, updating $\#K$ (and adjusting rate constants) between batches if needed to maintain $\#K=\Theta(n)$.} 
and $B+C \rxn^1 A + 2W$.
More generally,
add $K$ as catalysts to appropriate reactions,
until every reaction has the same order,
then add $W$ as products until every reaction has the same generativity.
Note $W$ (``waste'') is never a reactant so does not affect reaction propensities;
after each batch, $W$ can be removed to prevent the population size from exploding.


Since every reaction has the same order,
we can sample reactions with batching.
When we modify a reaction to have a $K$ catalyst,
we adjust its rate constant to compensate;
the reaction retains the same probability of being sampled as in the original CRN.
Since every reaction has the same generativity,
the distribution of collision-free run lengths does not depend on which reactions are part of the run,
allowing our batching algorithm to sample from this exact distribution before sampling which reactions occur.

In \cref{sec:discrete-time-simulation}, we describe sampling in \emph{discrete time}, i.e., drawing from the distribution of configurations after exactly $\ell$ reactions. 
This includes our core batching algorithm.
In \cref{sec:exact-time}, we adapt this to the much more challenging continuous-time regime,
modifying the discrete-time algorithm and using adaptive rejection sampling~\cite{gilks1992adaptive} to efficiently sample inter-reaction times.





%% file: sections/definitions.tex
\section{Definitions}

\subsection{Preliminaries}
\label{sec:prelim}

We use $\Lambda$ to denote a generic finite set of labels for chemical species. 
We use $\mathbb N^\Lambda$ to denote the set of functions $f: \Lambda \to \mathbb N$. 
We can also interpret such a function as a vector $\vec{r}$ indexed by elements of $\Lambda$, 
calling $\vec{r}(S)$ the \emph{count} of $S \in \Lambda$ in $\vec{r}$.
Equivalently, it can be interpreted as a multiset, containing $\vec{r}(S)$ copies of each $S \in \Lambda$.
We use
$\|\vec{r}\|_1 = \sum_{A \in \Lambda} \vec{r}(A)$ or simply $\|\vec{r}\|$ to denote its 1-norm.
We use $\mathbb N^\Lambda_i$ to denote the set of $\vec{r}$ such that $\|\vec{r}\| = i$.
For $A \in \Lambda$ and $n \in \mathbb N$, we use $n\cdot A$ to denote the vector $\mathbf{a}$ with $\mathbf{a}(A) = n$ and $\mathbf{a}(B) = 0$ for all $B \neq A$.
For $a,b\in\N$ with $a \geq b - 1$,
we write $a^{\underline{b}}$ to denote the \emph{falling power} $a! / (a-b)! = a (a-1) (a-2) \dots (a-b+1)$.

Given a set $\Lambda$, a \emph{reaction} over $\Lambda$ is a triple $\alpha = (\vr, \vp, k) \in \N^\Lambda \times \N^\Lambda \times \R_{>0}$, with $\vec{r}$ and $\vec{p}$ specifying the stoichiometric coefficients of the reactants and products, and $k$ specifying the \emph{rate constant}. 
A \emph{chemical reaction network (CRN)} is a pair $\cC = (\Lambda, R)$, where $\Lambda$ is a finite set of chemical species and $R$ is a set of reactions over $\Lambda$.
A \emph{configuration} of a CRN is a vector $\vec{c} \in \mathbb N^\Lambda$.
For each species $A \in \Lambda$, $\vec{c}(A)$ gives the count of $A$ in $\vec{c}$.

Given a configuration $\vec{c}$ of a CRN $\mathcal C = (\Lambda, R)$, we say a reaction $\alpha = (\vec{r}, \vec{p}, k)$ is \emph{applicable} if $\vec{r} \leq \vec{c}$, that is if there is enough of each reactant of $\alpha$ to carry out the reaction. 
If no reaction is applicable to $\vec{c}$, then we say $\vec{c}$ is \emph{terminal}. 
If $\alpha$ is applicable to $\vec{c}$, we write $\alpha(\vec c)$ to denote the configuration $\vec c - \vec r + \vec p$ obtained by subtracting the reactants and adding the products of $\alpha$ to $\vec c$. If $\vec{d} = \alpha(\vec{c})$ for some reaction $\alpha \in R$, then we write $\vc \rxn^{\alpha} \vd$ or simply $\vec{c} \rxn \vec{d}$. 
An \emph{execution} $\mathcal E$ of $\mathcal C$ is a finite or infinite sequence of configurations $(\vec{c}_0, \vec{c}_1, \vec{c}_2, \ldots)$ such that for each $i \geq 0$, $\vec{c}_i \rxn \vec{c_{i+1}}$.

Given a CRN $\mathcal C$ and a configuration $\vec{c}$ of $\mathcal C$, for this paper, the evolution of $\vec{c}$ over time is governed by the laws of \emph{Gillespie kinetics}\cite{gillespie1977exact}.
Each reaction $(\vec{r}, \vec{p}, k)$ is treated as a Poisson process whose rate is given by its \emph{propensity},
defined for each reaction $\alpha = (\vec{r}, \vec{p}, k)$ in a configuration $\vc$, with volume $v \in \R_{>0}$, as
\begin{equation}
\label{eq:propensity-defn}
    \propensity{\vc}{v}{\alpha} = \frac{k}{v^{\|\vec{r}\| - 1}}
    \cdot
    \prod_{A \in \vec{r}} \frac {\vec{c}(A)^{\underline{\vec{r}(A)}}}{\vr(A)!}.
\end{equation}
For example, if $\alpha$ is $A + 3B \rxn^{4.5} C$,
then $\propensity{\vc}{v}{\alpha}  = \frac{4.5}{6v^3} \vc(A) \vc(B) (\vc(B)-1) (\vc(B)-2).$
The final product term of \eqref{eq:propensity-defn} ($\vc(A) \vc(B) (\vc(B)-1) (\vc(B)-2)$ in this example) is  the number of ways to choose individual molecules from $\vc$ with the given reactant identities.
Note that reactions lacking sufficient reactants (e.g., if $\vc(A)=0$ or $\vc(B) < 3$)
have propensity 0. 
We also note that propensity is sometimes defined without this $\vr(A)!$ factor; 
these two conventions are equivalent, and conversion between them is done by multiplying rate constants by a simple correction factor. 
Define the \emph{total propensity} to be the sum of all propensities 
$p_{\vc,v}^\mathrm{tot} = \sum_{\alpha \in R} \propensity{\vc}{v}{\alpha}.$
The next reaction to be executed is $\alpha$ with probability $\frac{\propensity{\vc}{v}{\alpha}}{p_{\vc,v}^\mathrm{tot}}$,
and the elapsed time until this reaction occurs is distributed as an exponential random variable with rate $p_{\vc,v}^\mathrm{tot}.$


\begin{toappendix}    
The evolution of a CRN under Gillespie kinetics in some volume $v$ can be viewed as a \emph{continuous-time Markov chain} $\mathcal M_{\mathcal C,v} = (S, \mathbf{R})$.
The set of states $S$ is the set of configurations of $\mathcal C$.
The rate matrix $\mathbf{R}: S \times S \to \mathbb R_{\geq 0}$ is defined by taking, for every $(\vec{c}, \vec{d})$ such that $\vec{c} \rxn \vec{d}$,
\[\mathbf{R}(\vec{c}, \vec{d}) = \sum_{\{\alpha \in R| \vc \rxn^{\alpha} \vd\}}\propensity{\vc}{v}{\alpha} ,\]
with $\mathbf{R}(\vec{c}, \vec{d}) = 0$ if no such reaction exists.
The \emph{transition matrix} $\overline{\mathbf{R}}$ is obtained by normalizing each row of $\mathbf{R}$ to have sum 1, so that (viewed as a matrix) the row $\overline{\mathbf{R}}_{\vec{c}, \cdot}$ gives the probability of transition from $\vec{c}$ to each other configuration.
If $\vec{c}$ is terminal, then $\overline{\mathbf{R}}(\vec c, \vec c) = 1$ and all other values in the row are 0. \todo{Generally, make sure I'm being consistent about different kinds of letters, bolding, etc among notation choices.}

\todo{move the comment below this todo into a defn environment or something? Or define it in the section with finite time blowup? not sure.}

\end{toappendix}


\subsection{Uniform CRNs}
\label{sec:crn-new-defns}

\begin{defn}
    The \emph{order} $\ord(\alpha)$ of a reaction $\alpha=(\vr,\vp,k)$ is $\|\vec{r}\|$.
    The \emph{order} $\ord(\mathcal{C})$ of a CRN $\mathcal{C}$ is the maximum order of its reactions.
    If every reaction has the same order,
    $\mathcal{C}$ has \emph{uniform order}.
\end{defn}

\begin{defn}
    The \emph{generativity} $\gen(\alpha)$ of a reaction $\alpha=(\vr,\vp,k)$ is $\|\vp\| - \|\vr\|$. 
    The \emph{generativity} $\gen(\mathcal C)$ of a CRN $\mathcal C$ is the maximum generativity of its reactions. 
    If every reaction has the same generativity,
    $\mathcal{C}$ has \emph{uniform generativity}.
\end{defn}


\begin{defn}
    A CRN is \emph{uniform} if it has uniform order and uniform generativity;
    that is, there are $r, p \in \mathbb N$ such that every reaction has exactly $r$ reactants and $p$ products.
\end{defn}

Uniform CRNs are general enough to allow any CRN to be transformed into a uniform CRN with the same dynamics, while being constrained enough to allow batching.
Previous work\cite{mayr2014framework} defines and discusses \emph{conservative} CRNs, which require $\|\vr\| = \|\vp\|$ for each \emph{individual} reaction $(\vr, \vp, k)$, but allow different values of $\|\vr\|$ between different reactions. 
We call a CRN that is both uniform and conservative \emph{uniformly conservative}.
(Population protocols are uniformly conservative with order 2.)
Because uniform CRNs have uniform generativity, their molecular count changes in a predictable way, allowing similar analytical utility to conservative CRNs.
Uniform order is also important for the batching algorithm,
which is designed to simulate uniform CRNs.
Thus the first step of our simulation algorithm is to translate any CRN into a uniform CRN;
the next few definitions describe some quantities that need to be tracked while we do this transformation in order to preserve the dynamics of the original CRN.
We also note that although $\gen(\mathcal C)$ may be negative,
this case is straightforward to handle with the original batching algorithm~\cite{berenbrink2020simulating} by adding inert waste products to each reaction $\alpha$ with $\gen(\alpha) < \gen(\mathcal C)$ until  the CRN is uniformly conservative.
Thus throughout the rest of the paper we will assume $\gen(\mathcal C) \ge 0$ for any CRN $\mathcal C$.

\subsection{Definitions for CRN transformation}
\begin{defn}
    \label{def:total-rate-constant}
    Given a CRN $\mathcal C = (\Lambda, R)$, its \emph{total rate constant} for a multiset of reactants $\vec{r} \in \mathbb N^\Lambda$, denoted $\totalrateconstant{\cC}{\vr}$, is given by
    \eq{\totalrateconstant{\cC}{\vr} = \sum_{(\vec{r}, \vec{p}, k) \in R} k,}
    i.e., 
    $\totalrateconstant{\cC}{\vr}$ is the sum of the rate constants of all reactions in $R$ sharing the same reactants $\vec{r}$.
\end{defn}


\begin{defn}
    \label{def:uniformly-reactive}
    A uniform CRN $\mathcal C$ is \emph{uniformly reactive} if $\totalrateconstant{\cC}{\vr} = \totalrateconstant{\cC}{\vr'}$ for all $\vec{r},\vr' \in \mathbb N^\Lambda_{\ord(\mathcal C)}$.
\end{defn}

In other words, for a uniformly reactive CRN, the sum of rate constants among all reactions sharing a reactant vector is constant, no matter which reactant vector is chosen.
For example, the reactions 
\opt{full}{
\begin{align*}
    A &\rxn^1 B+C
    & B &\rxn^{2.5} A+C
    & C &\rxn^{1.5} A+A
    \\
    A &\rxn^2 A+B
    & B &\rxn^{0.5} C+C
\end{align*}
}
\opt{sub,final}{
$A \rxn^1 B+C, A \rxn^2 A+B$, 
$B \rxn^{2.5} A+C, B \rxn^{0.5} C+C$,
and
$C \rxn^{1.5} A+A$
}
are not uniformly reactive because the sum of rate constants for reactions with $A$ as the sole reactant is 3,
same as for reactions with sole reactant $B$, yet it is 1.5 for $C$.
But if we add the reaction $C \rxn^{1.5} B+B$, the CRN becomes uniformly reactive.

Uniformly reactive CRNs with order $o$ have the useful property that the probability of choosing a particular 
multiset of reactants $\vec{r}$ for the next reaction under Gillespie kinetics (recall \cref{eq:propensity-defn}) is \emph{exactly} the probability that if we pick $o$ molecules from the population uniformly at random without replacement,
their identities are given by $\vec{r}$.
In other words, picking reactants ``Gillespie-style'' is equivalent to picking reactants ``population-protocol-style''.
\opt{sub}{This is formalized in \Cref{lem:uniformly-reactive}.}

\begin{toappendix}
\begin{lem}
    \label{lem:uniformly-reactive}
    Let $\mathcal C = (\Lambda, R)$ be a uniformly reactive CRN and $\vec{c}$ be some configuration of $\mathcal C$.
    Let $\mathbf{X}$ denote the distribution of reactant vectors in $\mathbb N^\Lambda_{\ord(\mathcal C)}$ sampled by the Gillespie algorithm on $\vc$ in volume $v$.
    Let $\mathbf{Y}$ denote the distribution of multisets of $k$ molecules from $\vec{c}$, drawn by uniformly picking individual molecules without replacement.
    Then $\mathbf{X} = \mathbf{Y}$.
\end{lem}

\begin{proof}
    The probability of the next reaction being some $\alpha \in R$ is the propensity of $\alpha$ divided by the total propensity of all reactions.
    Thus, in $\mathbf{X}$,
    the probability of the next reactant vector being $\vr$ is proportional to the sum of the propensities of all $\alpha$ with that reactant vector, that is, to $\totalrateconstant{\cC}{\vr}$. That is to say,
    \begin{align*}
        \Pr(\rho(\vec c) = \vec{r}) &= \frac{1}{\sum_{\beta \in R} \propensity{\vc}{v}{\beta}}\sum_{\alpha = (\vec{r}, \_, \_) \in R} \propensity{\vc}{v}{\alpha}\\
        &= \frac{\totalrateconstant{\cC}{\vr}}{\sum_{\beta \in R} \propensity{\vc}{v}{\beta}}\cdot \prod_{A \in \vec{r}}\frac {\vec{c}(A)^{\underline{\vec{r}(A)}}}{\vr(A)!}.
    \end{align*}
    By uniform reactivity, the fraction on the left is constant across choices of $\vr$.
    The remainder then combinatorially gives exactly the number of ways to choose an ordered list of molecules from $\vc$ that comprise $\vr$ in some order.
    In $\mathbf{Y}$,
    the probability of $\vr$ being drawn uniformly at random is also proportional to this quantity, so the two distributions must be identical.
\end{proof} 

\end{toappendix}

In transforming a CRN into a uniform one and modifying the order of reactions, any reaction with increased order will have its propensity change as a result of new reactants and the impact of volume on propensity.
We give a way to relate this new propensity to the original propensity.

\begin{defn}
    \label{def:adjusted-parameters}
    Let $\cC = (\Lambda, R)$ be a CRN, $n$ be the molecular count of some configuration, and $v \in \R_{>0}$ be volume. 
    Then for $\alpha = (\vr, \vp, k) \in R$, the \emph{order-adjusted rate constant} for $\alpha$ is given by
    \eq{\adjrateconstant{\cC}{\alpha}{n}{v} = 
        k\cdot {\frac {v^{\delta_o}}{n^{\underline{\delta_o}}}}
    ,}
    Where $\delta_o = \ord(\cC) - \ord(\vr)$. Given a configuration $\vc$ with $\|\vc\| = n$, the \emph{order-adjusted propensity} of $\alpha$ is given by using its order-adjusted rate constant instead of $k$ in \cref{eq:propensity-defn}, that is, by
    \eq{
    \adjpropensity{\vc}{v}{\alpha} 
    = \frac{\adjrateconstant{\cC}{\alpha}{n}{v}}{k}\propensity{\vc}{v}{\alpha} 
    \left(
    = \frac{k\cdot v^{\ord(\cC) - 1}}{n^{\underline{\delta_o}}}
    \cdot 
    \prod_{A \in \vec{r}} \vec{c}(A)^{\underline{\vec{r}(A)}}
    \right).
    }
\end{defn}
This also allows us to precisely describe one potential source of inefficiency.
Because we choose reactants instead of choosing reactions directly, we will sometimes choose reactants that have no corresponding reaction.
We may also choose reactants which have low probability of reacting in the original CRN due to small rate constants, requiring an artificial slowdown to maintain exactness.
We account for both of these effects together.
\begin{defn}
    \label{def:efficiency-condition}
    Let $\cC = (\Lambda, R)$ be a CRN, $\vc$ be a configuration of $\cC$ with $\|\vc\| = n$, and $v \in \R_{>0}$ be volume. 
    Let $\bar{p}_{\vc,v}^\mathrm{tot}$ be the total order-adjusted propensity of $\cC$,
    \eq{
    \bar{p}_{\vc,v}^\mathrm{tot}
    = \sum_{\alpha \in R}\adjpropensity{\vc}{v}{\alpha}.
    }
    Let $\bar{k}^{\mathrm{max}}$ be the maximal order-adjusted rate constant of $\cC$,
    \eq{
    \bar{k}^{\mathrm{max}} 
    = \max_{\alpha \in R} \adjrateconstant{\cC}{\alpha}{n}{v}.
    }
    The \emph{scheduler slowdown factor} of $\cC$ on $\vc$ in volume $v$, denoted
    $\schedulerslowdown{\cC}{v}{\vc} \geq 1$, is given by
    \eq{
    S_{\cC,v}(\vc) = \frac {\bar{k}^{\mathrm{max}} \cdot \binom{2n}{\ord(\cC)}}{\bar{p}_{\vc,v}^\mathrm{tot}}.}
\end{defn}
Note that the term $2n$ is because $n$ is the count of molecules before we add $K$; since we add $n$ copies of $K$,
the total count after this step is $2n$.
Our algorithm can efficiently simulate any CRN execution whose configurations maintain a constant upper bound on this value.
It is sufficient for this that \emph{some} reaction with a large (i.e., not asymptotically smaller than any other) adjusted rate constant has no bottlenecks (i.e., all its reactants have count $\Omega(n)$).
We note that it is possible for this condition to be true at some point in the simulation and become false later;
for instance the reaction $2L \to L+F$ starts (with only $L$ present) with 
$S_{\cC,v}(\vc) \approx 4$ (due to unnecessary added $K$),
but it increases as $L$ is converted to $F,$
eventually reaching $\Theta(n^2)$ when $\# F \gg \#L$,
because most interactions are between two $F$'s and are therefore passive.
\opt{sub}{See \cref{fig:lotka_volterra_plot_with_passive_reactions} for a practical discussion.}



\begin{toappendix}
    
\subsection{Definitions relating to CRN simulation}
\label{sec:defns-related-to-crn-simulation}

\begin{defn}
    \label{def:future-distribution}
    Let $v \in \R_{>0}$ be volume. 
    Let $\mathcal C = (\Lambda, R)$ be a CRN 
    with corresponding continuous time Markov chain $\mathcal M_{\mathcal C,v} = (\mathbb N^\Lambda, \mathbf{R})$. 
    Let $\vec{c}$ be a configuration of $\mathcal C$, and $\ell \in \mathbb N$. 
    Then the \emph{discrete $\ell$-future distribution} of $\mathcal C$, denoted $\ellstep{\mathcal C}{\vec c}{\ell}{v}$, is given by $\overline{\mathbf R}^\ell_{\vec{c}, \cdot}$.
    That is, it is the distribution of configurations of $\mathcal C$ starting from $\vec{c}$ after $\ell$ transitions (allowing terminal configurations to self-transition). 
    
    The \emph{continuous $t$-future distribution} of $\mathcal C$, denoted $\tstep{\mathcal C}{\vec c}{t}{v}$, is a distribution over $\mathbb N^\Lambda \cup \{\emptyset\}$ such that if $\mathbf{X} \sim \tstep{\cC}{\vc}{t}{v}$, then for any configuration $\vec{d}$, $\Pr(\mathbf{X} = \vec{d})$ gives the probability of being in configuration $\vec{d}$ after $t$ time starting from configuration $\vec{c}$, and $\Pr(\mathbf{X} = \emptyset)$ gives the probability that $\mathcal C$ is not in any valid configuration at time $t$ (e.g., for CRNs, that infinitely many reactions happen by time $t$).
\end{defn}

\begin{defn}
    A CRN $\mathcal C$ is said to exhibit \emph{finite-time blowup} starting from some configuration $\vec{c}$ if $\Pr(\tstep{\cC}{\vc}{t}{v} = \emptyset) > 0$ for some $v > 0$ and $t > 0$.
\end{defn}
Finite-time blowup has been discussed for chemical systems governed by ODEs~\cite{johnston2012topics,weber1993finite}. 
It can also occur in stochastic CRNs. 
For example, the CRN $2A \rxn^1 3A$ exhibits finite-time blowup from any configuration with at least 2$A$,
in the sense that, for any time $t > 0$,
there is a positive probability that an infinite number of reactions occur before time $t$.

We now give a definition of simulation between CRNs. 
Many definitions of CRN simulation have been given in the literature.
Different definitions are relevant in different contexts, such as in the continuous ODE model~\cite{cardelli2015forward} or for formal verification of correspondence between different CRN implementations~\cite{johnson2019verifying,shin2019verifying}.
Our definition of simulation is much simpler, because we do not need to worry about checking whether an arbitrary CRN simulates another arbitrary CRN.
We are only worried about showing that any CRN can be simulated by its transformed version as shown in 
\Cref{sec:uniform-transformation};
the transformed CRN is closely related to the original.
To this end, our only concern is that the CRNs have matching dynamics: that is, that the output of the Gillespie algorithm would look the same between them,
ignoring the extra species introduced in the transformed CRN.
\begin{defn}
    Let $\mathcal C = (\Lambda', R)$ be a CRN and let $\Lambda \subseteq \Lambda'$. Given a configuration $\vec{c}$ of $\mathcal C$, define its \emph{restriction to $\Lambda$}, denoted $\vec{c}|_{\Lambda}$, as the configuration obtained by taking $\vec{c}$ and setting the count of all species not in $\Lambda$ to 0.
\end{defn}

\begin{defn}
    \label{def:crn-simulation}
    Let $v \in \R_{>0}$ be volume.
    Let $\mathcal C = (\Lambda, R)$ and $\mathcal C' = (\Lambda', R')$ with $\Lambda \subseteq \Lambda'$.
    We say $\cC'$ \emph{simulates $\cC$ from $\vc$} if there exists a configuration $\vc'$ of $\cC'$ such that for all $t \geq 0$,
    \[\tstep{\mathcal C'}{\vc'}{t}{v}|_{\Lambda} = \tstep{\mathcal C}{\vc}{t}{v}.\]
    We say that $\mathcal C'$ \emph{conditionally simulates $\mathcal C$ from $\vec{c}$} if this statement holds conditioned on these configurations being defined, that is, if $\mathbf{X} = \tstep{\cC'}{\vc'}{t}{v}|_{\Lambda}$ and $\mathbf{Y} = \tstep{\cC}{\vc}{t}{v},$ simulation requires that $\mathbf{X} = \mathbf{Y}$ while conditional simulation requires that
    \[\left(\mathbf{X}\mid \mathbf{X} \neq \emptyset\right) = \left(\mathbf{Y}\mid \mathbf{Y} \neq \emptyset\right).\] 
    If $\mathcal C'$ (conditionally) simulates $\mathcal C$ from every configuration $\vec{c}$ of $\mathcal C$, we say that $\mathcal C'$ (conditionally) simulates $\mathcal C$. 
    We may also say $\cC'$ \emph{simulates $\cC$ in continuous time}; if any of the above statements hold for $\ellstep{\cC}{\vc}{\ell}{v}$ and $\ellstep{\cC'}{\vc'}{\ell}{v}$ (rather than $\tstep{\cC}{\vc}{t}{v}$ and $\tstep{\cC'}{\vc'}{t}{v}$), then we say $\cC'$ \emph{simulates $\cC$ in discrete time}.
\end{defn}


\end{toappendix}

%% file: sections/transformation.tex
\section{CRN transformations}
\label{sec:uniform-transformation}

In \cite{berenbrink2020simulating}, working in the model of population protocols, the authors show how to simulate $\Theta(\sqrt n)$ interactions on a population of $n$ agents in $\Theta(\log n)$ time. 
The key insight is that sequential interactions between non-overlapping pairs of agents can be simulated simultaneously, as a \emph{batch}. 
The results of these interactions only need to be computed once a \emph{collision} occurs, i.e., an agent is chosen to interact a second time. 
On average $\Theta(\sqrt n)$ such interactions will occur before a collision occurs. 
These interactions are referred to as a \emph{collision-free run}.
Rather than simulating reactions individually, the authors first draw from the distribution of the length of a collision-free run.
Then, they determine how many times each particular kind of interaction occurs within that run.
Finally, they simultaneously apply the results of all of these interactions
(e.g., for reaction $A+B \rxn C+D$, rather than decrementing $A,B$ and incrementing $C,D$ repeatedly, instead if the reaction is sampled to occur $103$ times, subtract $103$ from $A$ and $B$, and add $103$ to $C$ and $D$), along with one extra collision interaction sampled to include at least one agent that was part of the collision-free run.

Because general CRNs may contain reactions that unpredictably change total molecular count, this idea is not immediately applicable. 
However, this method can be applied to uniform CRNs, with some modifications to account for changing molecular count and reactions of arbitrary uniform order (as opposed to uniform order 2).
Thus our first step is to transform an arbitrary CRN $\cC$ into a uniform CRN simulating $\cC$ (see \cref{sec:defns-related-to-crn-simulation} for our definition of CRN simulation).
\subsection{Transforming an arbitrary CRN into a uniform CRN}

\begin{algorithm}
\caption{Uniform CRN transformation}
\label{alg:crn-transformation}
\begin{algorithmic}
\REQUIRE{CRN $\cC = (\Lambda, R)$, volume $v \in \mathbb R_{>0}$, integer $k_0 \geq \ord(\cC)$}
\ENSURE{Uniform CRN $\mathcal C' = (\Lambda', R')$ simulating $\cC$ as in \cref{lem:simulation}}
\begin{enumerate}
    \item Add two new species, $K$ and $W$, setting $\Lambda' = \Lambda + \{K, W\}$. 
    \item To define $R'$, for each reaction $\alpha = (\vr, \vp, k)$ in $R$,
    let $\delta_o = \ord(\cC) - \ord(\alpha)$ and $\delta_g = \gen(\cC) - \gen(\alpha)$,
    and add the reaction $\alpha' = (\vr',\vp',k')$ to $R'$, where 
    \begin{enumerate}
        \item 
        $\vr' = \vr + \delta_o\cdot K$,
        \item 
        $\vp' = \vp + \delta_o\cdot K + \delta_g\cdot W$,
        
        \item $k' = \adjrateconstant{\cC}{\alpha}{k_0}{v}$ (see \cref{def:adjusted-parameters}).
    \end{enumerate}
    In other words: add $\ord(\mathcal C) - \ord(\alpha)$ copies of $K$ to both $\vr$ and $\vp$,
    add $\gen(\cC) - \gen(\alpha)$ copies of $W$ to $\vp$,
    and adjust $k$ by a multiplicative factor depending on input constants and how many $K$ will be added to the configuration being simulated.
\end{enumerate}
\end{algorithmic}
\end{algorithm}

Given an arbitrary CRN $\mathcal C$, we transform it into a uniform CRN $\mathcal C'$ via \cref{alg:crn-transformation}.
We demonstrate the method on the example CRN,
\opt{full}{
\begin{align*}
A + B &\rxn^2 C\\
C &\rxn^3 A + B
\end{align*}
}
\opt{sub,final}{
$A + B \rxn^2 C$,
$C \rxn^3 A + B$,
}
with a single reversible hetero-dimerization reaction, a case that existing methods cannot simulate efficiently. 
The second reaction's order is one less than the order of the CRN, so we add one copy of $K$ to catalyze the second reaction, giving the CRN uniform order:
\opt{full}{
\begin{align*}
A + B &\rxn^2 C\\
C + K &\rxn^3 A + B + K
\end{align*}
}
\opt{sub,final}{
$A + B \rxn^2 C$,
$C + K \rxn^3 A + B + K$.
}
Next, the first reaction's generativity is -1, while the CRN has generativity 1. So we add 2 copies of $W$ to the first reactant's products, giving the CRN uniform generativity (thus now uniform):
\opt{full}{
\begin{align*}
A + B &\rxn^2 C + W + W\\
C + K &\rxn^3 A + B + K
\end{align*}
}
\opt{sub,final}{
$A + B \rxn^2 C + W + W$,
$C + K \rxn^3 A + B + K$
}
Finally, since a reaction's rate depends on its reactants,
we must adjust the rate constant of the second reaction where $K$ was added as a reactant,
accounting for volume and our choice of $k_0$.
Here, this will result in multiplying the second reaction's rate constant by $\frac v {k_0}$, so if we take, for example, $v = 6$ and $k_0 = 30$, this would result in the output CRN:
\opt{full}{
\begin{align*}
A + B &\rxn^2 C + W + W\\
C + K &\rxn^{0.6} A + B + K
\end{align*}
}
\opt{sub,final}{
$A + B \rxn^2 C + W + W$,
$C + K \rxn^{0.6} A + B + K.$
}

In summary:
$W$ is an inert (i.e., not a reactant anywhere) ``waste'' species,
so does not affect the dynamics of the CRN.
$K$ has constant count, so its presence has a constant effect on each reaction's propensity over the course of a batch, which is accounted for by modifying rate constants.
$k_0$ represents the count of $K$ that will be added to the original CRN configuration for simulation.
We require $k_0 \geq \ord(\mathcal C)$ for simplicity.
Choosing $k_0 \in \Theta(n)$ yields the simplest asymptotic analysis, and is a reasonable practical choice.



\begin{lemrep}
\label{lem:simulation}
    Let $\cC = (\Lambda, R)$ be a CRN, and let $\cC' = (\Lambda', R')$ be the output of \cref{alg:crn-transformation} on $\cC$ with any choice of $v$ and $k_0$. Then $\cC'$ simulates $\cC$ in continuous and discrete time (see \cref{def:crn-simulation}). 
\end{lemrep}

\opt{sub}{
Our complete definition of CRN simulation is given in \cref{sec:defns-related-to-crn-simulation}. 
In short, $\cC'$ simulating $\cC$ means that samples drawn from the two CRNs have identical distributions of counts of species in $\cC$,
as this is what it means for a simulation algorithm to be exact.
}

\begin{proof}
    Let $\mathbf{k} = k_0\cdot K$.
    Given a configuration $\vec{c}$ of $\mathcal C$, we will use the configuration $\vec{c}' = \vc + \mathbf{k}$ of $\mathcal C'$ to satisfy the definition of simulation given in \cref{def:crn-simulation}. 

    Consider the Markov chain $\mathcal M_{\cC', v} = (S', \mathbf{R}')$ representing $\mathcal C'$. 
    Define an equivalence relation $\sim$ on $S'$ by $\vc_1 \sim \vc_2 \iff \forall A \in \Lambda' \setminus \{W\}, \vc_1(A) = \vc_2(A)$, i.e., they differ only in their count of $W$.
    Let $\widetilde{\mathcal M}_{\mathcal C'}$ be the continuous time Markov chain whose states are the equivalence classes under $\sim$ of the states of $S'$ which contain exactly $k_0$ copies of $K$, with transitions inherited from $\mathbf{R}'$.
    This inheritance is well-defined because $W$ is inert, so equivalent states always have transitions to equivalent states.
    Note also that these are the only equivalence classes we need to consider, as any configuration reachable from $\vec{c}'$ has $k_0$ copies of $K$.

    We observe that $\widetilde{\mathcal M}_{\mathcal C'}$ is isomorphic to $\mathcal M_{\mathcal C}$.
    For any state $\vec{d}$ in $\mathcal M_\cC$, the corresponding state in $\widetilde{\mathcal M}_{\cC'}$ is the equivalence class of $\vec{d} + \mathbf{k}$.
    By construction, corresponding reactions have equal propensities in these two CRNs:
    the factor by which \cref{alg:crn-transformation} multiplies $k$ to obtain $k'$ accounts for the extra $K$ in the reactants and the need to divide by a different power of $v$. 
    Explicitly, for any pair of corresponding reactions $\alpha = (\vr, \vp, k)$ and $\alpha' = (\vr', \vp', k')$,
    \begin{align*}
    \propensity{\vc'}{v}{\alpha'}
    &= \frac{k'}{v^{\|\vec{\vr'}\|}} \cdot \left(\prod_{A \in \vec{r'}} \vec{c'}(A)^{\underline{\vec{r'}(A)}}\right)\\
    &= \frac{k \cdot \frac{v^{\vec{r'}(K)}}{\lfloor v \rfloor^{\underline{\vec{r}(K)}}}}{v^{\vr'(K) + \|\vec{\vr}\|}} \cdot \left(\lfloor v \rfloor^{\underline{\vec{r}(K)}}\prod_{A \in \vec{r}} \vec{c}(A)^{\underline{\vec{r}(A)}}\right)\\\
    &= \frac{k}{v^{\|\vec{\vr}\|}} \cdot \left(\prod_{A \in \vec{r}} \vec{c}(A)^{\underline{\vec{r}(A)}}\right)\\
    &=\propensity{\vc}{v}{\alpha}.    
    \end{align*}
    This shows that corresponding transitions have the same rates in $\mathcal M_{\mathcal C}$ and $\widetilde{\mathcal M}_{\mathcal C'}$.
    It follows that they have identical distributions of configurations sampled at any time or number of steps.
\end{proof}

\subsection{Transforming a uniform CRN into a uniformly reactive CRN}

Our batching algorithm is based on population protocols:
it chooses a random set of molecules, rather than a random reaction. 
These molecules may not correspond to a reaction. 
In this case, for consistency, we still need molecular count to update. 
Additionally, we must account for differing rate constants between different reactions.

We solve these issues by ensuring that the value $\totalrateconstant{\cC}{\vr}$ (see \cref{def:total-rate-constant}) is identical among every possible multiset $\vr$ of reactants the batching algorithm might sample.
This is the notion of a \emph{uniformly reactive} CRN (\cref{def:uniformly-reactive}).
\cref{alg:null-reactions} transforms the output of \cref{alg:crn-transformation} into such a CRN. 
It does so by adding \emph{passive reactions}, which update molecular count consistently but only add $W$, and thus do not affect any reaction's propensity.

For example, the example given for \cref{alg:crn-transformation} would be transformed into the following:
\opt{full}{
\begin{align*}
A + B &\rxn^2 C + W + W\\
C + K &\rxn^{0.6} A + B + K\\
C + K &\rxn^{1.4} C + K + W\\
A + A &\rxn^{2} A + A + W\\
A + C &\rxn^{2} A + C + W\\
&\vdots\\
W + W &\rxn^{2} W + W + W,
\end{align*}
}
\opt{sub,final}{
$A + B \rxn^2 C + W + W$,
$C + K \rxn^{0.6} A + B + K$,
$C + K \rxn^{1.4} C + K + W$,
$A + A \rxn^{2} A + A + W$,
$A + C \rxn^{2} A + C + W$,
$\dots$,
$W + W \rxn^{2} W + W + W,$
}
where there is a passive reaction with rate constant 2 for every possible reactant multiset except for $A + B$ and $C + K$.

\begin{algorithm}
\caption{Total rate constant uniformity transformation}
\label{alg:null-reactions}
\begin{algorithmic}
\REQUIRE{CRN $\cC = (\Lambda, R)$, volume $v \in \mathbb R_{>0}$, integer $k_0 \geq \ord(\cC)$}
\ENSURE{Uniformly reactive CRN $\mathcal C'$ simulating $\cC$ as in \cref{lem:null-reaction-simulation}}
\begin{enumerate}
    \item
    Let $\cC_0 = (\Lambda + \{K, W\}, R_0)$ be the uniform CRN output by \cref{alg:crn-transformation} on $\cC$, $v$ and $k_0$.
    
    \item
    Let $k^{\mathrm{max}} = \max_{\vec r \in \mathbb N^\Lambda_{\ord(\mathcal C_0)}} \totalrateconstant{\cC_0}{\vr}$, the maximum total rate constant of any set of reactants.
    
    \item
    Define a set of \emph{passive reactions} P: for every $\vec{r} \in \mathbb N^\Lambda_{\ord(\mathcal C_0)}$ such that $\totalrateconstant{\cC_0}{\vr} < k^{\mathrm{max}}$, P contains the reaction
    \eq{
    (\vec{r}, \vec{r} + \gen(\cC_0) \cdot W, k^{\mathrm{max}} - \totalrateconstant{\cC_0}{\vr}).
    }
    \item Output $\cC' = (\Lambda + \{K, W\}, R_0 + P)$.
\end{enumerate}
\end{algorithmic}
\end{algorithm}

\begin{obs}
    The output of \cref{alg:null-reactions} is uniformly reactive (see \cref{def:uniformly-reactive}).
\end{obs}

\begin{lemrep}
    \label{lem:null-reaction-simulation}
    Let $\mathcal C = (\Lambda, R)$ be a CRN, and $\vec{c}$ a configuration of $\mathcal C$. Let $\mathcal C' = (\Lambda', R')$ be the result of running \cref{alg:null-reactions} on $\mathcal C$ with any choice of $v$ and $k_0$. Then $\mathcal C'$ conditionally simulates $\mathcal C$ in continuous and discrete time (see \cref{def:crn-simulation}).
\end{lemrep}

Note that we can only guarantee conditional simulation because the reactions we add to $\cC'$ may cause it to exhibit finite-time blowup. 
For example, the reaction $2W \rxn 3W$ causes finite-time blowup on its own.
Removing $W$ between batches prevents this issue.

\begin{proof}
    Given a configuration $\vec c$ of $\mathcal C$, we simulate it from $\vc + k_0 \cdot K$. 
    We must show that these two CRNs have identical distributions over configurations at time $t$, conditioned on both of them having valid configurations at time $t$.
    By construction, there is a natural injection $f: R \to R'$ sending $\alpha \in R$ to its modified version in $R'$ that was output by \cref{alg:crn-transformation}.
    We couple the two CRNs as follows: whenever a reaction $\alpha'$ occurs in $\mathcal C'$, if there is some $\alpha \in R$ such that $f(\alpha) = \alpha'$, then $\alpha$ occurs in $\mathcal C$.
    Otherwise (i.e., if $\alpha'$ is passive), nothing happens in $\mathcal C$.
    This is a valid coupling because any passive $\alpha'$ only affect the count of $W$, so their occurrence does not affect the propensity of any reactions in $R'$ nor the count of any species in $\Lambda$.
    Therefore, from the point of view of $\cC$, these extra reactions do not affect the dynamics of the system so long as $\cC'$ eventually reaches time $t$, which is guaranteed by conditioning. 
\end{proof}

%% file: sections/simulation.tex
\section{Simulating discrete-time CRN executions}
\opt{full,final,sub}{
\label{sec:discrete-time-simulation}
}
\begin{toappendix}
\opt{sub}{\label{apx:discrete-time-simulation}}
\end{toappendix}

We have shown that any CRN can be simulated by a uniformly reactive CRN. 
In this section, we show how to leverage this to efficiently sample a ``discrete-time'' CRN execution,
meaning that we will merely count the number of reactions in each batch,
but make no attempt to sample the amount of continuous time elapsed during the batch according to the Gillespie distribution.
That is, given a CRN $\cC$, a starting configuration $\vc$, volume $v$, and a number of steps to simulate $\ell$, we wish to efficiently draw from the distribution $\ellstep{\cC}{\vc}{\ell}{v}$ (\cref{def:future-distribution}).
\Cref{sec:exact-time} describes how to modify the algorithm to sample accurate timestamps to assign to each sampled configuration.


\subsection{Discrete sampling with a scheduler}
\label{subsec:scheduler}
We now turn our attention to algorithms which simulate CRNs, rather than CRNs which simulate each other.
The first such algorithm, \cref{alg:scheduler}, is analogous to the algorithm $\textsc{SEQ}$ in \cite{berenbrink2020simulating}.
It is not intended to be executed, and is provided for analytical comparison.
The only adjustment we must make is to convert rate constants into probabilities.



\begin{algorithm}
\caption{Scheduler-based uniform CRN simulation}
\label{alg:scheduler}
\begin{algorithmic}
\REQUIRE{CRN $\mathcal C = (\Lambda, R)$, volume $v \in \R_{>0}$, configuration $\vc$ of $\cC$, integer $\ell \in \N$}
\ENSURE{Configuration $\vc'$ of $\cC$ distributed as in \cref{lem:scheduler-correctness}}

    Let $\cC'$ be the uniformly reactive CRN output by \cref{alg:null-reactions} on input $\cC$, $v$, and $k_0 = \|\vc\|$. Let $\vc' = \vc + \|\vc\|\cdot K$.
    Set \texttt{steps} = 0, and repeat until \texttt{steps} = $\ell$:
    \begin{enumerate}
        
        \item 
        Choose $\ord(\cC')$ molecules from $\vc'$, each uniformly at random without replacement, and let $\vr$ be the resulting multiset.
        
        \item
        Choose a reaction $\alpha$ from $\cC'$ with reactant vector $\vr$, with probability proportional to its rate constant (i.e., reaction $\alpha_i = (\vr_i, \vp_i, k_i)$ is chosen with probability $\frac {k_i}{\totalrateconstant{\cC'}{\vr}}$). Update $\vc'$ by executing $\alpha$.

        \item
        If a non-passive reaction was executed, 
        (i.e., anything other than $\vr \rxn \vr + \gen(\cC') \cdot W$), 
        increment \texttt{steps}.
        
        \item
        Set $\vc'(W)$ to 0 (to avoid finite-time blowup).
        
    \end{enumerate}
\end{algorithmic}
\end{algorithm}

From \Cref{lem:uniformly-reactive,lem:null-reaction-simulation}, we can infer that \cref{alg:scheduler} is exact. Formally:

\begin{cor}
    \label{lem:scheduler-correctness}
    On input $\cC$, $v$, $\vc$, and $\ell$, the output of \cref{alg:scheduler} is distributed as $\ellstep{\cC}{\vc}{\ell}{v}$ (see \cref{def:future-distribution}).
\end{cor}
\subsection{Batch size sampling}

In \cite{berenbrink2020simulating}, the authors use a notion of collision-free runs, and sample from the length of such a run. 
Our analysis will be similar to theirs, with some additional details to account for arbitrary order and generativity.
Below,
we use ``red'' and ``green'' to refer to molecules that have interacted (respectively, not interacted) in a batch.

\begin{defn}
    Let $\cC$ be a uniformly reactive CRN and $\vc$ be a configuration of $\cC$.
    Suppose each molecule in $\vc$
    is colored red or green, and that reactions consume specific copies of molecules uniformly at random.
    Whenever a reaction occurs, all products (including catalysts) are colored red.
    Then, given an execution $\cE$ of $\cC$ from $\vc$, we say that the \emph{collision index} of $\cE$ is the index (starting from 0) of the first reaction with any red molecule as a reactant. 
    That is, it is the number of reactions that are executed before any red molecule is a reactant. 

    We define the random variable $\ell$ as the collision index of $\cE$ when $\cE$ is sampled (using normal stochastic sampling) among all (infinite) executions of $\cC$ from $\vc$. 
    We say $\ell \sim \mathbf{coll}(n, r, o, g)$,\footnote{For our purposes, the parameter $r$ will always be 0; however, we show the computation in generality, as general $r$ is necessary to implement \emph{multibatching} as described in \cite{berenbrink2020simulating}. Adapting our result to this regime can be done in much the same way as in \cite{berenbrink2020simulating}.} where $n = |\vc|$, $r$ is the count of red molecules in $\vc$, $o = \ord(\cC)$ and $g = \gen(\cC)$. 
\end{defn}
Note that because $\cC$ is 
uniformly reactive (recall \Cref{def:uniformly-reactive}), this distribution will not depend on the specific reactions in $\cC$ or species present in $\vc$.
\begin{defn}
    \label{def:multifactorial}
    The \emph{multifactorial} $n!^{(g)}$ is given by the product $n(n-g)(n-2g)\ldots$, continuing until (and including) the last positive term.
\end{defn}




\begin{lemrep}
    \label{lem:collision-length-cdf}
    Let $\ell \sim \mathbf{coll}(n, r, o, g)$, where $g \geq 0$.
    Then $\ell$ has (reverse) cumulative distribution
    \opt{full}{
    \begin{align*}
        \Pr(\ell \geq k) = \begin{cases} 
          \frac{(n - r)!}{(n - r - ko)!} \prod_{j=0}^{o-1}\frac{(n - g - j)!^{(g)}}{(n + g(k-1) - j)!^{(g)}} & 0 \leq k \leq \frac {n - r} o, g > 0\\\\
          \frac{(n - r)!}{(n - r - ko)!} \prod_{j=0}^{o-1}\frac{1}{(n - j)^k} & 0 \leq k \leq \frac {n - r} o, g = 0\\\\
          0 & \text{otherwise.}
       \end{cases}
    \end{align*}
    }
    \opt{sub,final}{
    $\Pr(\ell \geq k) = 
    \frac{(n - r)!}{(n - r - ko)!} \prod_{j=0}^{o-1}\frac{(n - g - j)!^{(g)}}{(n + g(k-1) - j)!^{(g)}}$ if $0 \leq k \leq \frac {n - r} o, g > 0$,
    $\Pr(\ell \geq k) = 
    \frac{(n - r)!}{(n - r - ko)!} \prod_{j=0}^{o-1}\frac{1}{(n - j)^k}$ if $0 \leq k \leq \frac {n - r} o, g = 0$,
    and it is $\Pr(\ell \geq k) = 0$ otherwise.
    }
\end{lemrep}

\opt{full}{
\todo{JP: My intuition says that the following will be a useful example to show that actually simulating the collision is important. I.e., ``proposition 1'': if you don't simulate collisions, it does affect things meaningfully. Consider the CRN $A+A \to B+B$, $B+B \to C+C$. Consider the statistic of the time at which you first observe a C in the population. If I never simulate a collision (i.e., do ``tau leaping without replacement'') then this statistic should wind up being later than it otherwise would have been, generally. Maybe this gets even worse if you add things like $C+C \to D+D, D+D \to E+E$, etc.}
}
\begin{proof}
    Because of uniform reactivity, we may ignore what species each molecule is: 
    at each step, every $o$-tuple of molecules is equally likely to comprise the next reaction's reactants.
    Therefore, we can equivalently consider a scheduler that repeatedly picks individual reactants without replacement, and executes a reaction every time it picks $o$ of them (returning the products to the pool of molecules it may pick).
    If, after a reaction, the population has $m$ total molecules with $l$ green and $m - l$ red, then the probability of this scheduler only picking green molecules in the \emph{next} reaction is
    \[\prod_{j=0}^{o-1} \frac{l - j}{m - j}.\]
    Initially, there are $n$ molecules, $r$ red and $n - r$ green. 
    After each reaction, $o$ green molecules are replaced with $o + g$ red molecules. 
    Therefore, the probability of executing at least $k$ reactions before a collision is given by 
    \[\Pr[\ell \geq k] = \prod_{i=0}^{k-1} \prod_{j=0}^{o-1} \frac{\overbrace{n - r - oi}^{\text{\#green after $i$ rxns}} - j}{\underbrace{n + gi}_{\text{pop size after $i$ rxns}} - j}.\]
    Examining the numerator and denominator separately, we observe that the numerator varies over all integers from $n - r - ko + 1$ to $n - r$, i.e., the combined product of all numerators is the ratio
    $\frac{(n - r)!}{(n - r - ko)!}$.
    Factoring this out and commuting the products yields:
    \begin{align*}
        \Pr[\ell \geq k] = &\frac{(n - r)!}{(n - r - ko)!}\cdot \prod_{j=0}^{o-1}\prod_{i=0}^{k-1}\frac{1}{n + gi - j}.
    \end{align*}
    This inner product is the ``Pochhammer $g$-symbol'' $(n + gi - j)_{k,g}$.\footnote{This is normally called the ``Pochhammer $k$-symbol'', but we use $k$ as the first parameter rather than the second.} 
    Our task is now to rewrite it in a way that can be efficiently computed when $k$ is large.
    When $g > 0$, we can rewrite the inner product as a ratio where all but $k$ terms cancel by using multifactorials (\cref{def:multifactorial}):
    \begin{align*}
        \frac{(n - r)!}{(n - r - ko)!}\cdot \prod_{j=0}^{o-1}\frac{(n - g - j)!^{(g)}}{(n + g(k-1) - j)!^{(g)}}.
    \end{align*}
    When $g = 0$, corresponding to a uniformly conservative CRN,
    \footnote{Note that even in this case and when $o = 2$ (i.e., for population protocols), we obtain a different expression from \cite{berenbrink2020simulating}, since we use a different definition of this distribution in terms of reactions ($o$ molecules sampled at a time) instead of reactants (1 molecule sampled at a time).}
    each term in the inner product is identical and the formula simplifies to: 
    \begin{align*}
        \frac{(n - r)!}{(n - r - ko)!}\cdot \prod_{j=0}^{o-1}\frac{1}{(n - j)^k}.
        \qquad\qquad\qquad
        \qedhere
    \end{align*}
\end{proof}

\begin{lemrep}
\label{lem:sample_coll_efficient}
    It is possible to sample from $\mathbf{coll}(n, r, o, g)$ in time $O(o \cdot \log n)$.
\end{lemrep}
\opt{sub,final}{
\begin{proofsketch}
    \opt{sub}{A full proof is given in \Cref{apx:discrete-time-simulation}.}
    Intuitively,
    we use inversion sampling~\cite{devroye2006nonuniform},
    based on the CDF for $\mathbf{coll}(n, r, o, g)$ as shown in \Cref{lem:collision-length-cdf},
    and similarly to~\cite{berenbrink2020simulating},
    working with the log of the values in that formula.
    We rewrite the log of $n!^{(g)}$ in terms of the log-gamma function $\ln \Gamma(x)$, which is efficiently computable~\cite{devroye2006nonuniform}.
\end{proofsketch}
}
\begin{proof}
    To sample $\ell \sim \mathbf{coll}(n, r, o, g)$, we can draw a uniform variate $U \sim \text{Unif}([0,1])$ and use inversion sampling~\cite{devroye2006nonuniform} to draw a sample via binary search in $O(\log n)$ comparisons between $U$ and the CDF (because $0 \leq \ell \leq n$).
    To compare $U$ to the formula given by \cref{lem:collision-length-cdf}, we take the log of both sides and see that the relation we must compute is whether or not
    \begin{align*}
        \log((n-r)!) - \log((n - r - ko)!) + \sum_{j = 0}^{o - 1} \left[\log\left((n - g - j)!^{(g)}\right) - \log\left((n + g(k-1) - j)!^{(g)}\right)\right] < \log U.
    \end{align*}
    To compute this efficiently, we make use of the $\Gamma$ function, the generalization of factorial to non-integers.
    Computing $\log(x!)$ for any $x$ can be done efficiently using standard implementations of the log-gamma function $\log(\Gamma(x))$.
    To compute the log of a multifactorial, we make use of the relation that $x \cdot \Gamma(x) = \Gamma(x+1)$, even for non-integer $x$
    (e.g., $x=17/5$ in the first substitution below).
    We divide each term of each multifactorial product by $g$, and collect these factors of $g$.
    This yields a product where consecutive terms differ by 1, which can be written as a ratio of $\Gamma$ functions.
    For example,
    \begin{align*}
        \log(17!^{(5)}) 
    &=
        \log(17 \cdot 12 \cdot 7 \cdot 2)
    =
        \log\left(5^4 \cdot \frac{17}{5} \cdot \frac{12}{5} \cdot \frac{7}{5} \cdot \frac{2}{5}\right)
    = 
        \log\left(5^4 \cdot \frac{17}{5} \cdot \frac{ \Gamma(17/5)}{\Gamma(17/5)} \cdot \frac{12}{5} \cdot \frac{7}{5} \cdot \frac{2}{5}\right)
    \\&= 
        \log\left(5^4 \cdot \frac{\Gamma(22/5)}{\Gamma(17/5)} \cdot \frac{12}{5} \cdot \frac{7}{5} \cdot \frac{2}{5}\right)
    = 
        \log\left(5^4  \cdot \frac{\Gamma(22/5)}{\Gamma(17/5)} \cdot 
        \frac{12}{5} \cdot \frac{\Gamma(12/5)}{\Gamma(12/5)} \cdot \frac{7}{5} \cdot \frac{2}{5}\right)
    \\&=
        \log\left(5^4  \cdot \frac{\Gamma(22/5)}{\Gamma(17/5)} \cdot \frac{\Gamma(17/5)}{\Gamma(12/5)} \cdot \frac{7}{5} \cdot \frac{2}{5}\right)
    =
        \log\left(5^4 \cdot \frac{\Gamma(22/5)}{\Gamma(12/5)} \cdot \frac{7}{5} \cdot \frac{2}{5}\right)
    \\\dots&=
        \log\left(5^4 \cdot
        \frac{\Gamma(22/5)}{\Gamma(2/5)}
        \right)
    = 
        4\log(5) + \log \Gamma(22/5) - \log \Gamma(2/5).
    \end{align*}
    Using this method, these log-multifactorial terms can be computed efficiently even if they contain many terms.
    It follows that every term on the 
    \todo{DD: left-hand side of what? Perhaps should give equation a number}
    left-hand side can be computed efficiently, and there are $\Theta(o)$ terms.
\end{proof}

\opt{full}{
Since we start each batch with no molecules having interacted, in the next lemma,
we consider only this case ($r=0$ in $\mathbf{coll}(n, 0, o, g)$);
however,
we strongly suspect that a similar result holds that was shown in~\cite{berenbrink2020simulating},
that for $r \geq \sqrt{n}$,
$\E{\ell} = \Theta(n / r)$.
}

\begin{lemrep}
    \label{lem:sqrt-expectation}
    Let $\ell \sim \mathbf{coll}(n, 0, o, g)$. 
    Then $\E{\ell} = \Theta(\sqrt{n})$. 
\end{lemrep}


\begin{proof}
    Recall from the proof of \Cref{lem:collision-length-cdf} that (setting $r=0$)
    \[
    \Pr[\ell \geq k] 
    = 
    \prod_{i=0}^{k-1} 
    \prod_{j=0}^{o-1} 
    \frac{n - oi - j}{n + gi - j}.
    \]

    We first show the lower bound:
    \begingroup
    \allowdisplaybreaks
    \begin{align*}
    \E{\ell}
    &=
    \sum_{k=1}^{\lfloor n/o \rfloor}
    \Pr[\ell \geq k] 
    =
    \sum_{k=1}^{\lfloor n/o \rfloor}
    \prod_{i=0}^{k-1} 
    \prod_{j=0}^{o-1} 
    \frac{n - oi - j}{n + gi - j}
    >
    \sum_{k=1}^{\lfloor n/o \rfloor}
    \prod_{i=0}^{k-1} 
    \prod_{j=0}^{o-1} 
    \frac{n - oi - (o-1)}{n + gi}
    \\&=
    \sum_{k=1}^{\lfloor n/o \rfloor}
    \prod_{i=0}^{k-1} 
    \left(
    \frac{n - oi - o + 1)}{n + gi}
    \right)^o
    >
    \sum_{k=1}^{\lfloor n/o \rfloor}
    \prod_{i=0}^{k-1} 
    \left(
    \frac{n - o(k-1) - o + 1}{n + g(k-1)}
    \right)^o
    =
    \sum_{k=1}^{\lfloor n/o \rfloor}
    \left(
    \frac{n - ok + 1}{n + g(k-1)}
    \right)^{ok}
    \\&>
    \sum_{k=1}^{\lfloor n/o \rfloor}
    \left(
    \frac{n - ok}{n + gk}
    \right)^{ok}
    >
    \sum_{k=1}^{\sqrt{n}}
    \left(
    \frac{n - ok}{n + gk}
    \right)^{ok}
    =
    \sum_{k=1}^{\sqrt{n}}
    \left(
    \frac{n}{n + gk}
    -
    \frac{ok}{n + gk}
    \right)^{ok}
    \\&=
    \sum_{k=1}^{\sqrt{n}}
    \left(
    \frac{n + gk}{n + gk}
    -
    \frac{ok + gk}{n + gk}
    \right)^{ok}
    =
    \sum_{k=1}^{\sqrt{n}}
    \left(
    1
    -
    \frac{ok + gk}{n + gk}
    \right)^{ok}
    >
    \sum_{k=1}^{\sqrt{n}}
    \left(
    1
    -
    \frac{(o+g)\sqrt{n}}{n}
    \right)^{ok}
    \\&=
    \sum_{k=1}^{\sqrt{n}}
    \left(
    1
    -
    \frac{o+g}{\sqrt{n}}
    \right)^{ok}
    >
    \sum_{k=1}^{\sqrt{n}}
    \left(
    1
    -
    \frac{o(o+g)}{\sqrt{n}}
    \right)^{k}
    \qquad
    \begin{array}{l}
    \text{Bernoulli's inequality}
    \\
    (1-x)^o > 1-xo 
    \text{ when } x \leq 1
    \end{array}
    \\&=
    \frac{1 - \left(
    1
    -
    \frac{o(o+g)}{\sqrt{n}}
    \right)^{\sqrt{n}}}
    {1 - \left(
    1
    -
    \frac{o(o+g)}{\sqrt{n}}
    \right)}
    =
    \sqrt{n}
    \frac{1 - \left(
    1
    -
    \frac{o(o+g)}{\sqrt{n}}
    \right)^{\sqrt{n}}}
    {o(o+g)}
    \ge
    \sqrt{n}
    \frac{1 - e^{-o(o+g)}}
    {o(o+g)},
    \text{ since }
    \left( 
    1 + \frac{x}{\sqrt{n}}
    \right)^{\sqrt{n}}
    \le e^x.
    \end{align*}
    \endgroup

    For the upper bound,
    recall in the definition of $\ell \sim \mathbf{coll}(n, r, o, g)$ that $g$ is the number of additional red molecules added each reaction, in addition to the $o$ green molecules that are turned red.
    Clearly $\ell$ has a smaller expected value when $g=0$ than when $g>0$,
    since the extra red molecules in the latter case make it more likely that we pick a red reactant molecule on each interaction.
    Thus for upper-bounding $\E{\ell}$,
    we make the worst-case assumption that $g=0$.
    Similar reasoning lets us assume $o=1$ in the worst case for the upper bound.
    This turns out to be precisely what was analyzed in~\cite[Lemma 4]{berenbrink2020simulating},\footnote{
         This is in fact the distribution of the original birthday problem: how many balls can we throw into $n$ bins before some bin gets two balls?
         Although~\cite{berenbrink2020simulating} studies order-2 CRNs,
         their definition of a run length is based on the number of molecules that interact, as opposed to the number of interactions as in our definition.
    }
    where it is shown
    $\E{\ell} \leq 2 \sqrt{n}$.
\end{proof}

\subsection{Single batch content sampling}
\label{sec:single-batch}

To determine what reactions to simulate in a batch, we use a method very similar to that of \cite{berenbrink2020simulating}.
In population protocols over $q$ states, the authors view this problem as sampling the values of a $q \times q$ transition matrix $D$.
To do this, they first sample the row sums of $D$
(corresponding to the first reactant) using a multivariate hypergeometric distribution,\opt{full}{\footnote{
    A multivariate hypergeometric distribution asks, if we sample $\ell$ molecules from an urn without replacement, prescribed initial counts of molecule species $\{1,\dots,q\}$,
    how many of each species do we get?
}} 
and then sample values within each row
(corresponding, for each possible first reactant species $r_1$, to how many of each species $r_2$ get paired as second reactant to $r_1$), taking $q+q^2$ samples.
We use a similar approach, except that we must populate an $(\ord(\cC))$-dimensional array $D$.
We first sample the $q = |\Lambda|$ codimension-1 sums of $D$: that is, for each species, how many reactions in the batch will have that species as their first reactant.
Once we have done this, we are left with $q$ $(\ord(\cC) - 1)$-dimensional subarrays, and can recursively use multivariate hypergeometric distributions until we have sampled all elements.
This will require $\Theta(q^{\ord(\cC)})$ multivariate hypergeometric samples, each of which each can be obtained in constant time\cite{stadlober1989ratio}.
In distribution, this process is equivalent to sampling all the individual reactants for every reaction one at a time.

\subsection{Full discrete-time simulation algorithm}

With this, we are ready to describe the complete algorithm which efficiently samples from the discrete-time distribution $\ellstep{\cC}{\vc}{\ell}{v}$.
\cref{alg:discrete-time-single-batch} simulates one batch, while \cref{alg:full-discrete-time-algorithm} runs batches until it has simulated $\ell$ reactions.

\begin{algorithm}
\caption{Discrete-time single-batch simulation}
\label{alg:discrete-time-single-batch}
\begin{algorithmic}
\REQUIRE{Uniformly reactive CRN $\cC' = (\Lambda, R)$, configuration $\vc_0$ of $\cC'$, batch size bound $\ell_\mathrm{max} \in \N$}
\ENSURE{Number of steps simulated $\ell_{\mathrm{out}} \leq \ell_{\mathrm{max}}$, Configuration $\vc$ of $\cC'$ distributed as in \cref{lem:partial-discrete-correctness}}

\begin{enumerate}
    \item 
    Let $\vc'$ be an empty configuration, and let $\vc = \vc_0$.
    
    \item
    Sample a collision-free run length $\ell \sim \mathbf{coll}(n, 0, \ord(\cC'), \gen(\cC'))$, as described in \cref{lem:sample_coll_efficient}. If $\ell \geq \ell_{\mathrm{max}}$, set $\ell = \ell_{\mathrm{max}}$, set $\ell_{\mathrm{out}} = \ell$ and \texttt{do\_collision} = False;
    otherwise if $\ell < \ell_{\mathrm{max}}$,
    set $\ell_{\mathrm{out}} = \ell + 1$ and \texttt{do\_collision} = True.
    
    \item
    Let $q = |\Lambda|$.
    Sample the batch by sampling the values of the $\ord(\cC')$-dimensional transition array $D$ by recursively drawing $\Theta(q^{\ord(\cC')})$ multivariate hypergeometric samples summing to $\ell$, as described in 
    \cref{sec:single-batch}.

    \item 
    \label{step:batch}
    Execute the batched reactions. 
    For each entry $r \in \N$ of $D$ corresponding to reactants $\vr$:
    
    Draw a sample from a multinomial distribution on $r$ trials, with probabilities proportional to the rate constants of each reaction in $\cC'$ having reactant vector $\vr$. 
    For each sampled value, apply the result of executing the corresponding reaction that many times, removing the reactants from $\vc$ and adding the products to $\vc'$.
    If the corresponding reaction is not passive, also increase \texttt{steps} by the sampled value.

    \item 
    If \texttt{do\_collision} = True, simulate a collision:
    Draw an ordered list of reactants $\vr$ uniformly from $\vc + \vc'$, conditioned on at least one of them being from $\vc'$ (see proof of \cref{lem:discrete-efficiency-final} for how). Then, execute a single reaction from $\cC'$ with reactant vector $\vr$ as described in step \ref{step:batch}, and if it is not passive, increment \texttt{steps}.

    \item
    Set $\vc := \vc + \vc'$ (and output it).

\end{enumerate}
    
\end{algorithmic}
\end{algorithm}

\begin{algorithm}
\caption{Discrete-time full simulation}
\label{alg:full-discrete-time-algorithm}
\begin{algorithmic}
\REQUIRE{CRN $\mathcal C = (\Lambda, R)$, volume $v \in \R_{>0}$, configuration $\vc_0$ of $\cC$, number of steps $\ell_\mathrm{max} \in \N$}
\ENSURE{Configuration $\vc$ of $\cC$ distributed as in \Cref{lem:full-discrete-correctness}}

Set $\vc = \vc_0$ and \texttt{steps} = 0. Repeat until \texttt{steps} = $\ell_\mathrm{max}$:
\begin{enumerate}
    \item 
    Set $k_0 = \|\vc\|$, and let $\cC'$ be the output of \cref{alg:null-reactions} on inputs $\cC$, $v$, and $k_0$. Let $\vc' = \vc + k_0 \cdot K$.
    
    \item
    Call \cref{alg:discrete-time-single-batch} on input $\cC'$, $\vc'$, and $\ell_{\mathrm{max}} - \mathtt{steps}$. Set $\vc$ to the returned configuration, with all $W$ and $K$ removed. Add the returned number of steps $\ell_{\textrm{out}}$ to \texttt{steps}.

\end{enumerate}
    
\end{algorithmic}
\end{algorithm}

\opt{sub}{
The proof that \Cref{alg:full-discrete-time-algorithm} is both correct (simulates the Gillespie model exactly) and fast is given in \Cref{apx:discrete-time-simulation}. 
We show correctness by coupling the process with \cref{alg:scheduler}, which is correct by \cref{lem:scheduler-correctness}.
We show that each batch simulation takes time $O\left(q^{\ord(\cC)}\log(n)\right)$ time, and simulates $\Theta(\sqrt n)$ reactions.
Some of these reactions may be passive (i.e., not correspond to $\cC)$, and we bound the expected multiplicative slowdown resulting from this by the largest value of $\schedulerslowdown{\cC}{v}{\vc}$ for any $\vc$ during the simulation for a relatively simple bound (see \cref{def:efficiency-condition}). 
}

\begin{toappendix}
    
\begin{lem}
\label{lem:partial-discrete-correctness}
    On input $\cC'$, $\vc_0$, and $\ell_{\mathrm{max}}$, and for any volume $v \in \R_{>0}$, the output of \cref{alg:discrete-time-single-batch} is distributed as $\ellstep{\cC'}{v}{\vc_0}{\ell_{\mathrm{max}}}$.
\end{lem}
\begin{proof}
    Note first that $v$ has no effect on this distribution, as $\cC'$ is uniform, so the value of $v$ affects all propensities identically.
    We couple \cref{alg:discrete-time-single-batch} with \cref{alg:scheduler}, which has the correct distribution by \cref{lem:scheduler-correctness}.

    The only difference between the two processes is that \cref{alg:discrete-time-single-batch} batches reactions. 
    Thus, we can couple the processes by running \cref{alg:scheduler} while keeping track of individual molecules, 
    and whenever a collision occurs (i.e., when a molecule produced since the last collision is chosen as a reactant),
    using that number of reactions to sample the length of a collision-free run in \cref{alg:discrete-time-single-batch}.
    From the perspective of \cref{alg:discrete-time-single-batch}, nothing is changed by this coupling, because it samples from the exact distribution of collision-free run lengths given in \cref{lem:collision-length-cdf}. 
    If it is possible for this collision-free run length to be greater than $\ell_{\mathrm{max}} - \texttt{steps}$,
    then whenever this many reactions occur with no collision in \cref{alg:scheduler},
    we also run \cref{alg:full-discrete-time-algorithm} on a batch of this size (with no collision).
    This maintains the coupling in all cases.
\end{proof}
Together with \cref{lem:null-reaction-simulation}, because $\cC'$ simulating $\cC$ means they have the same distributions of configurations, this implies the following:
\begin{cor}
    \label{lem:full-discrete-correctness}
    On input $\cC$, $v$, $\vc_0$, and $\ell_{\mathrm{max}}$, the output of \cref{alg:full-discrete-time-algorithm} is distributed as $\ellstep{\cC}{v}{\vc_0}{\ell_{\mathrm{max}}}$.
\end{cor}

We now turn to the issue of efficiency. 
At its core, our algorithm's efficiency comes from batching, allowing the simulation of $\Theta(\sqrt n)$ reactions in $O(\log n)$ time.
Most batching steps will be able to run this many reactions, 
so long as there are at least $\Theta(\sqrt n)$ left to simulate.

\begin{lem}
    \label{lem:full-discrete-individual-batch}
    Suppose that \cref{alg:full-discrete-time-algorithm} calls \cref{alg:discrete-time-single-batch} on inputs $\cC'$, $\vc'$, and $\ell_{\mathrm{max}} - \mathtt{steps}$ with $\|\vc'\| = n$ and $(\ell_{\mathrm{max}} - \mathtt{steps}) \in \Omega(\sqrt n)$.
    Then the number of steps returned by \cref{alg:discrete-time-single-batch} is $\Omega\left(\frac{\sqrt{n}}{\schedulerslowdown{\cC}{v}{\vc}}\right)$ in expectation (i.e., the iteration simulates this many reactions from the original CRN).
\end{lem}
\begin{proof}
    \todo{Give some intuition for why this is true throughout the whole batch, not just for the initial configuration}
    By \cref{lem:sqrt-expectation}, after step 2, the value $\ell$,
    the number of (possibly passive) reactions in the batch,
    is $\Theta(\sqrt n)$ in expectation. 
    
    First, we argue using the coupling shown in the proof of \cref{lem:partial-discrete-correctness},
    that we may assume that each of these reactions is chosen with probability proportional to its propensity (i.e., as though chosen by the Gillespie algorithm) \emph{from $\vc'$},
    ignoring any intermediary configurations in the batch. 
    From the Gillespie algorithm's perspective, each individual reaction before the collision (sequentially ordered) can be viewed as having reactants chosen uniformly at random from the configuration it occurs in,
    because this is what the Gillespie algorithm does on uniformly reactive CRNs.
    Thus, the \emph{first} reaction has reactants uniformly chosen from $\vc'$; 
    but from the perspective of \cref{alg:discrete-time-single-batch},
    each reaction within a batch has the same distribution of reactants,
    so all of them must have this distribution.
    We may ignore the collision (which is where the distributions of the two algorithms would be reconciled),
    because here we are only concerned with asymptotic statements, 
    and a single reaction will not change those.

    What remains to show is that each of these steps has probability $\frac{1}{\schedulerslowdown{\cC}{v}{\vc}}$ of being non-passive, which implies the lemma.
    This claim is shown via a dissection of \cref{def:efficiency-condition}.
    We repeat the last line of the definition here:
    \[S_{\cC,v}(\vc) = \frac {\bar{k}^{\mathrm{max}} \cdot \binom{2n}{\ord(\cC)}}{\bar{p}_{\vc,v}^\mathrm{tot}}.\]
    The denominator gives the total propensity of all reactions in $\cC'$ that came from $\cC$.
    The numerator gives the total propensity of all reactions in $\cC'$.
    This latter claim can be seen because we simulate $\cC'$ on a configuration of size $2n$, and each set of $\ord(\cC)$ molecules in $\vc$ will contribute a value of its corresponding total rate constant to the total propensity, and uniform reactivity implies that this value is the same for all $\binom{2n}{\ord(\cC)}$ of these sets.

    Thus, each batch has $\Theta(\sqrt n)$ reactions in expectation, each of which has probability $\frac{1}{\schedulerslowdown{\cC}{v}{\vc}}$ of corresponding to a reaction from $\cC$.
\end{proof}

Our next result requires a (very, very weak) statement about collision-free run lengths not being too small too often.

\begin{lem}
\label{lem:median}
Let $\ell \sim \mathbf{coll}(n, 0, o, g)$ and let $k \in \Theta(\sqrt n)$. 
    Then $\Pr(\ell \geq k) = \Omega(1)$ with respect to $n$.
\end{lem}
\begin{proof}
    In the proof of \cref{lem:sqrt-expectation}, we show that
    \[\Pr(\ell \geq k) > \left(
    \frac{n - ok}{n + gk}
    \right)^{ok}.\]
    Setting $k = c\sqrt n$ and taking $n$ to be large, we can cancel a square root, yielding 
    \[\Pr(\ell \geq k) > \lim_{n \to \infty} \left(
    \frac{\sqrt n - oc}{\sqrt n + gc}
    \right)^{oc\sqrt n}
    = \lim_{m \to \infty} \left(
    \frac{m - oc}{m + gc}
    \right)^{ocm}\]
    Taking a log and applying l'H\^opital's rule, this becomes
    \[= \exp\left(\lim_{m \to \infty} \frac{\ln\left(
    \frac{m - oc}{m + gc}
    \right)}{1/(ocm)}\right)
    = \exp\left(\lim_{m \to \infty} \frac{\left(
    \frac{m + gc}{m - oc}\frac{oc + gc}{(m + gc)^2}
    \right)}{-1/(ocm^2)}\right)\]
    This simplifies to a ratio of two cubic polynomials, and so is clearly constant with respect to $n$.
\end{proof}

The next lemma refers to the other case from \Cref{lem:full-discrete-individual-batch},
which is that we only have $O(\sqrt{n})$ steps left before reaching $\ell_\mathrm{max}$.

\begin{lem}
    \label{lem:full-discrete-end-not-too-long}
    Suppose that \cref{alg:full-discrete-time-algorithm} calls \cref{alg:discrete-time-single-batch} on inputs $\cC'$, $\vc'$, and $\ell_{\mathrm{max}} - \mathtt{steps}$ with $(\ell_{\mathrm{max}} - \mathtt{steps}) \in O(\sqrt n)$, 
    and that every configuration $\vc$ appearing at the start of an iteration from this point satisfies $\schedulerslowdown{\cC}{v}{\vc} \leq s$ for some $s \in \R_{>0}$.
    Then the algorithm will halt in $O\left(s\log n\right)$ more iterations in expectation.
\end{lem}
\begin{proof}
    If $\ell_{\mathrm{max}} - \texttt{steps}$ is in $O(\sqrt n)$, then a constant fraction of batches will simulate $\ell_{\mathrm{max}} - \texttt{steps}$ reactions by \cref{lem:median}.
    On average, the fraction of these reactions that are non-passive will be $\frac{1}{\schedulerslowdown{\cC}{v}{\vc}}$.
    We can view the number of iterations to completion from this point as the maximum of $\ell_{\mathrm{max}} - \mathtt{steps}$ independently distributed geometric variables with success probability at least $1/s$, 
    because each call to \cref{alg:discrete-time-single-batch} has probability at least $1/s$ to simulate a non-passive reaction for each reaction it simulates.
    In other words, one can view each of the $\ell_{\mathrm{max}} - \texttt{steps}$ remaining reactions as a coin with independent probability at least $1/s$ of landing heads,
    each of which is flipped each batch until every coin has landed heads once.
    Thus, we wish to bound the expected maximum number of flips that one of these coins takes to land heads once.
    The given expression of $s\log n$ is a well-known bound on the expectation of this maximum
    (intuitively, each round of flipping eliminates a fraction $s$ of the remaining coins on average).
\end{proof}
\todo{glue; mention that the actual run time is more precisely the sum of terms like in next lemma across each batch; but this is simpler to write and equivalent if $n$ and $S$ stay within constant factors.}
\begin{lem}
    \label{lem:discrete-efficiency-final}
    Suppose \cref{alg:full-discrete-time-algorithm} is run on input $\cC$, $v$, $\vc_0$ and $\ell_{\mathrm{max}}$. 
    Suppose that every configuration $\vc$ appearing at the start of an iteration of this execution satisfies $n_{\mathrm{min}} \leq \|\vc\| \leq n_{\mathrm{max}}$
    and $\schedulerslowdown{\cC}{v}{\vc} \leq s$
    for some $n \in \N$, $s \in \R_{>0}$.
    Then the algorithm runs in time $O\left(\frac{q^{\ord(\cC)}s\ell_{\mathrm{max}}\log(n_{\textrm{max}})}{\sqrt {n_{\textrm{min}}}}\right)$.
\end{lem}
\todo{Make sure to make every max and min subscripts, not superscripts, everywhere}
\begin{proof}
\todo{make ref's for the steps and use ref command here}
    Step 1 consists of simple transformations on the CRN,
    so is not relevant to asymptotic analysis as other steps will take longer.
    Step 2 takes time $O(\ord(\cC)\log(n_{\mathrm{max}}))$ by \cref{lem:sample_coll_efficient}.
    Step 3 can be done in time $\Theta(q^{\ord(\cC)})$, as each multivariate hypergeometric sample can be drawn in constant time \cite{stadlober1989ratio}.
    Step 4 also takes time $\Theta(q^{\ord(\cC)})$, as there are this many entries in $D$ and relevant samples and configuration updates can be done in constant time per entry.
    Steps 5 and 6 are not costly; the conditional sampling in step 5 can be done by simple combinatorial calculations.

    Thus, each iteration takes time $\Theta(q^{\ord(\cC)})$. 
    By \cref{lem:full-discrete-individual-batch}, it will take on average $O\left(\frac {s\ell_{\mathrm{max}}} {\sqrt n}\right)$ iterations until there are $O(\sqrt n)$ remaining reactions to simulate.
    By \cref{lem:full-discrete-end-not-too-long}, the remainder of the algorithm from there does not take too long.
\end{proof}

\end{toappendix}

%% file: sections/making_time_work.tex
\section{Simulating with continuous time}
\opt{full,final,sub}{
\label{sec:exact-time}
}
\begin{toappendix}
\opt{sub}{\label{apx:exact-time}}
\end{toappendix}

In this section, we show our main theorem:
\todo{Does this theorem statement (or perhaps an earlier one? But I think here) need to include an assumption that $\ell \in \Omega(n_{\textrm{min}}^p)$)?}
\begin{thm}
    \label{thm:maintheorem}
    Let $\cC = (\Lambda, R)$ be a CRN with $|\Lambda| = q$, $\vc$ be a configuration of $\cC$, and $v\in \R_{>0}$ be volume. Let $p \in (0,\frac 1 2]$ be a parameter.
    Then it is possible to exactly sample from $\tstep{\cC}{\vc}{t}{v}$ (that is, to sample the configuration of $\cC$ at time $t$ in volume $v$, starting from $\vc$) in expected time
    \[O\left(\frac{(2q)^{\ord(\cC)}s\ell\log(n_{\mathrm{max}})}{n_{\mathrm{min}}^p} + n_{\mathrm{max}}^{2p/3}\ell^{1/3}\log(n_{\mathrm{max}})\right),\]
    so long as all $\ell$ configurations $\vc$ occurring up to time $t$ satisfy $n_{\mathrm{min}} \leq \|\vc\| \le n_{\mathrm{max}}$ and $\schedulerslowdown{\cC}{v}{\vc} \leq s$ (see \cref{def:efficiency-condition} for the definition of $\schedulerslowdown{\cC}{v}{\vc})$.
    
\end{thm}

This theorem is directly implied by \cref{lem:full-continuous-correctness} (correctness) and \cref{lem:full-continuous-efficiency} (efficiency). See \cref{sec:full-continuous-algorithm} for a discussion of the asymptotics.

So far, we have described an algorithm that exactly simulates CRNs in \emph{discrete time}.
That is, it can sample accurately from the distribution of \emph{executions} of a CRN from a given configuration, and then output some subsequence of that execution (outputting the entire execution would take linear time).
We also wish to know the (continuous) \emph{inter-reaction times} - that is, not just what configuration the CRN is in and how many reactions have occurred (i.e., ``discrete time''), but also how long it took to get there: a \emph{timestamp}.
Sampling this information efficiently is surprisingly difficult compared to merely sampling discrete time information. 
Our approach is to exploit the inert $W$ molecule introduced in \cref{alg:crn-transformation}.
By carefully modifying the count of this molecule between batches,
we ensure that the batching algorithm repeatedly cycles through the same small ($\Theta(n^p)$ for some $p \in (0, \frac 1 2])$ set of distinct molecular counts.
This allows us to take advantage of a method called adaptive rejection sampling \cite{gilks1992adaptive} to quickly sample inter-reaction times from the same distribution repeatedly,
rather than sampling the time of each individual reaction.

Throughout this section,
let $\vc \in \N^\Lambda$ be a configuration of the CRN $\mathcal C$ at the start of a batch,
let $n = \| \vc \|$,
let $k \in \N^+$ denote the number of interactions for which we wish to sample the total inter-reaction time,
let $o = \ord(\mathcal C)$,
and let $g = \gen(\mathcal C)$.
Like with individual reactions, sampling individual inter-reaction times would take linear time, so we must sample the sum of many inter-reaction times together.
Each individual inter-reaction time is distributed as an exponential random variable with rate equal to total propensity $\totalpropensity{\vc}{v}$.
For uniformly reactive CRNs this value is always equal to a constant times $\binom{n}{o}$.
Since each reaction increases the total molecular count by $g$, 
the $i$'th reaction has a rate proportional to $\binom{n+ig}{o}$;
it follows that for a batch of size $k$,
we must sample the variable
\eq{
\mathbf{H} = 
\sum_{i=0}^{k-1} 
\mathbf{X} \left(
    \binom{n + ig}{o}
\right),
}
where each $\mathbf{X}(\lambda)$ is independently exponentially distributed with rate $\lambda$. 
The distribution of such a variable $\mathbf{H}$ is known as a \emph{hypoexponential distribution}.\footnote{In the case where the CRN is conservative, that is, each reaction has an equal number of reactants and products, the timestamps follow an Erlang distribution, i.e., a sum of independent and \emph{identically} distributed exponential random variables.}
Our aim here is to sample from it efficiently.

We first show (\Cref{sec:sampling-hypoexponential,sec:end-batch-rejection-sampling}) a theoretical approach allowing us to efficiently sample from $\mathbf{H}$ \emph{exactly} under an assumption of bounded molecular count.
We then discuss (\Cref{sec:approximating-hypo}) more practical approaches for sampling approximately from $\mathbf{H}$, which are faster and appear in theory and practice to be accurate enough that they do not meaningfully affect the output distribution.

\subsection{Sampling exactly from the hypoexponential distribution}
\label{sec:sampling-hypoexponential}
To drastically simplify this discussion, we assume that the execution we simulate has bounded total molecular count.
This assumption is true of any physically realistic CRN execution.
It also does not meaningfully deviate from what the Gillespie algorithm can simulate;
for example, on CRNs which exhibit finite-time blowup, any algorithm attempting to sample a configuration at some given time must have some probability of failure.
Even on CRNs that do not exhibit finite-time blowup, unbounded molecular counts typically occur as a result of exponential growth, a case where the Gillespie algorithm will take exponential time and is not practical.


The PDF $P$ 
of a hypoexponential distribution with distinct rates $\lambda_1, \ldots, \lambda_k$ is~\cite{ross2014introduction}:
\opt{full}{
\begin{equation}
\label{eqn:hypoexponential}
P(t) = \sum_{i=1}^k C_{i,k}\lambda_ie^{-\lambda_it},
\end{equation}
}
\opt{sub,final}{
$
P(t) = \sum_{i=1}^k C_{i,k}\lambda_ie^{-\lambda_it},
$
}
where
\eq{
C_{i,k} = 
\prod_{j \in \{1,\dots,k\} \setminus \{i\}}
\frac{\lambda_j}{\lambda_j - \lambda_i}.
}
\opt{sub}{\Cref{apx:exact-time} shows how to compute the $C_{i,k}$'s efficiently.}

\begin{toappendix}
\opt{sub}{\label{sec:apx-proof-of-c-i-k}
Recall the PDF of a hypoexponential with rate $\lambda_1,\dots,\lambda_k$ is
\begin{equation}
\label{eqn:hypoexponential}
    P(t) = \sum_{i=1}^k C_{i,k}\lambda_ie^{-\lambda_it},
\end{equation}
where each
$C_{i,k}= 
\prod_{j \in \{1,\dots,k\} \setminus \{i\}}
\frac{\lambda_j}{\lambda_j - \lambda_i}.$}
Na\"{i}vely computing all $C_{i,k}$ appears to require $\Theta(k^2)$, 
taking time $O(k)$ for each of the $k$ products $C_{i,k}$.
(Each binomial coefficient $\lambda_i = \binom{n+ig}{o}$ can be computed in $O(1)$ time since we consider $o,g = O(1)$ with respect to $n$.)
However, we can improve this to $\Theta(k \log^2 k)$ time.
First note that we can write
\[
C_{i,k} = 
\prod\limits_{j \ne i}
\frac{\lambda_j}{\lambda_j - \lambda_i}
=
\frac{\prod_{j \ne i} \lambda_j}
{\prod_{j \ne i} (\lambda_j - \lambda_i)}.
\]
The top product can be handled in this way.
We use time $\Theta(k)$ to compute the full product $A = \prod_{j=1}^k \lambda_j$.
We then compute each term $\prod_{j \ne i} \lambda_j$ for $1 \le i \le k$ in time $O(1)$ by a single division $A / \lambda_i$.
Thus the top products can be computed in total time $\Theta(k)$.

Now consider the bottom products
$
\prod_{j \ne i}
(\lambda_j - \lambda_i).
$
Consider the polynomial 
$f(x) = 
\prod_{j=1}^k
(\lambda_j - x).
$
Let $f'(x) = \frac{d}{dx} f(x)$.
Then by the product rule,
\[
    f'(x) 
= 
    \sum_{i=1}^k
    \frac{d}{dx}
    (\lambda_i - x)
    \cdot
    \left(
    \prod_{j \ne i}
    (\lambda_j - x)
    \right)
=
    - \sum_{i=1}^k
    \left(
    \prod_{j \ne i}
    (\lambda_j - x)
    \right).
\]
If we evaluate $f'(\lambda_m)$ for $1 \le m \le k$,
for terms $i \ne m$ in the sum,
one factor $(\lambda_j - x)$ in the product is 0 (specifically for $j=m$,
so those terms of the sum vanish),
and we have
$
C_{i,k}
=
\prod_{j \ne i} (\lambda_j - \lambda_i) 
=
- f'(\lambda_i).
$
Thus it suffices to construct $f$ and evaluate its derivative at $\lambda_i$ to compute $C_{i,k}.$

To evaluate at these points efficiently (time $O(k \log^2 k)$ for all $k$ products),
the first step is to convert the factored polynomial 
$f(x) = 
\prod_{j=1}^k
(\lambda_j - x)$
into its expanded coefficient form.
This can be done with a standard divide-and-conquer recursion, splitting into two polynomials with $k/2$ factors each.
The base case is to FOIL the degree-2 $(\lambda_i - x)(\lambda_{i+1} - x) = \lambda_i \lambda_{i+1} - (\lambda_i + \lambda_{i+1}) x + x^2$.
The recursive case can be handled using standard FFT-based polynomial multiplication routines~\cite{cormen2022introduction},
taking time $\Theta(k \log k)$ to multiply two degree-$k$ polynomials.
Since this will have $\log k$ levels of recursion and spend time $O(k \log k)$ at each level of recursion,
the total time required is $O(k \log^2 k)$.
Now that we have computed the polynomial $f'(x)$,
similar FFT-based methods for \emph{multipoint evaluation}~\cite{von2003modern} 
can be used to evaluate $f'$ at $k$ points in time $O(k \log^2 k)$.
(More generally time $O(k+m) \log^2 (k+m)$ to evaluate a degree-$k$ polynomial at $m$ points; $m=k$ in our case.)
Thus in time $O(k \log^2 k)$,
we can compute the coefficients $C_{i,k}$ used in the definition of the PDF and CDF;
so long as we continue sampling from the same hypoexponential;
these values do not need to be recomputed,
reducing the amortized cost the more reactions are executed in the total simulation.
\end{toappendix}

Nevertheless,
once we have the coefficients $C_{i,k}$,
it still requires time $\Theta(k)$ to evaluate the PDF in \eqref{eqn:hypoexponential}.
However, to achieve an asymptotic speedup over the Gillespie algorithm,
we must sample a hypoexponential defined by $k$ exponentials in time asymptotically smaller than $k$, 
so time $\Theta(k)$ remains too expensive.
However, this PDF is log-concave: that is, the second derivative of $\log(P(t))$ is negative for all $t > 0$.\footnote{This follows because the hypoexponential distribution is a convolution of exponentials, which are log-concave, and convolution preserves log-concavity.}
This allows us to use the black-box method of adaptive rejection sampling, outlined in \cite{gilks1992adaptive}, to sample repeatedly from the distribution without having to evaluate the PDF every time, as with normal rejection sampling.
In short, this method works by establishing upper and lower piecewise linear bounds on $\log(P(t))$, and improving these bounds to match the function closely as more evaluations are made.
The more accurate these bounds are, the less likely that drawing a sample requires one to explicitly compute the PDF.
In \cite{wild1993algorithm}, the authors give empirical evidence and give a proof
\footnote{The authors describe their proof as ``tentative''; while it somewhat lacks detail, we believe it to be sufficiently rigorous to show this asymptotic relationship.}
that this method generally allows one to obtain $O(m)$ samples from such a distribution in only $O(\sqrt[3]{m})$ evaluations of the PDF.
We will assume that this is the case.

In order to take advantage of this, we must guarantee that we sample from the same hypoexponential distribution repeatedly.
The key insight is that $W$ can be freely added or removed without altering any propensities corresponding to the original CRN, 
since $W$ is inert, 
but its presence changes what hypoexponential distribution will be sampled.
If we repeatedly execute some number of (possibly passive) reactions $r$, then remove $r\cdot g$ copies of $W$ from the configuration (i.e., return to the molecular count before the $r$ reactions were executed),
then each such cycle will have total elapsed times that are distributed like independent identical hypoexponentials. 
Because we assume an upper bound $n_{\mathrm{max}}$ on molecular count, we can guarantee there will always be enough $W$ to remove. 
If no such bound exists, our algorithm remains correct, but may be inefficient.
To prevent these extra $W$ from causing too many passive reactions, we can switch to a different hypoexponential whenever the molecular count halves (or doubles up to $n_{\mathrm{max}})$.
This adds an extra $\log n_{\mathrm{max}}$ factor to the analysis.

\subsection{Rejection sampling for exact end times}
\label{sec:end-batch-rejection-sampling}

The typical input to the Gillespie algorithm contains an exact continuous time $t$ at which to sample a configuration.
We will eventually run past time $t$ during a batch, but wish to sample the configuration at time exactly $t$.
To avoid this issue while remaining exact, we use rejection sampling to sample how many reactions from the batch occur before time $t$,
under the assumption that the last reaction occurs after time $t$.

The simplest version of this procedure is as follows:
when the end time of a batch is sampled to be past $t$,
sequentially sample the exponential variables that comprise the batch.
The first such sample that would cause simulated time to go past $t$ indicates the first reaction in the batch that happens \emph{after} time $t$, so we simulate everything before it, without a collision
(as the collision would happen after time $t$).
If all times are sampled and $t$ still has not been reached, we reject the sample and start over.
If the probability of some batch step exceeding time $t$ is $p$,
then the probability of the sample being accepted is also $p$,
leading to on average $\frac 1 p$ rejections.
Therefore, each time a batch step is run, the expected number of rejections is 1.
The expected cost of a rejection is the expected size of the final batch, which is $\Theta(\sqrt n)$.
Therefore, it is generally fine to run this rejection sampling in this way, as sampling a single batch slowly does not affect asymptotics in expectation.

\opt{full}{ 
This can be made more efficient via binary searching by sampling from hypoexponential distributions.
However, on any input $t$ on which we expect $\Omega(n)$ reactions to happen, this is not necessary,
as the last batch contains $\Theta(\sqrt n)$ reactions in expectation (see \cref{lem:sqrt-expectation}), so sampling the last batch slowly only adds an additional $\Theta(\sqrt n)$ time.}

\subsection{Full continuous-time simulation algorithm}
\label{sec:full-continuous-algorithm}
Here we provide our main algorithm, \cref{alg:full-continuous-time-algorithm}.
Like \cref{alg:full-discrete-time-algorithm}, it works by repeatedly calling \cref{alg:discrete-time-single-batch}.
However, it does so in such a way that allows repeated sampling from the same hypoexponential distribution. 
Note in particular that for all of these algorithms,
we describe the output as a single configuration,
but in practice we would sample a sequence of configurations by calling such algorithms repeatedly.

\begin{algorithm}
\caption{Continuous-time exact simulation}
\label{alg:full-continuous-time-algorithm}
\begin{algorithmic}
\REQUIRE{CRN $\mathcal C = (\Lambda, R)$, volume $v \in \R_{>0}$, configuration $\vc_0$ of $\cC$, end time $t_\mathrm{max}$, batching parameter $p \in \left(0, \frac 1 2 \right]$}
\ENSURE{Configuration $\vc$ of $\cC$ distributed as in \cref{lem:full-continuous-correctness}}

Set $\vc = \vc_0$ and $t = 0$. Repeat until step \ref{step:continuous-step-two} exits:
\begin{enumerate}    
    \item
    \label{step:continuous-step-one}
    Let $i$ be the least integer where $\|\vc\| \leq 2^i$,
    $k_0 = \|\vc\|$, 
    $n_0 = 2^{i+1}$,
    $\ell_{\mathrm{max}} = \lfloor n_0^p\rfloor$,
    and $\cC'=$ the output of \Cref{alg:null-reactions} on inputs $\cC$, $v$, and $k_0$. Let $\vc' = \vc + k_0 \cdot K + (n_0 - \|\vc\| - k_0) \cdot W$, so that $\|\vc'\| = n_0$.

    \item 
    \label{step:continuous-step-two}
    Sample a value $t_0$ from the hypoexponential distribution (see \Cref{sec:sampling-hypoexponential}) with rates $\lambda_j = \binom{n_0 + j\cdot\gen(\cC)}{\ord(\cC)}, 0 \leq j < \ell_{\mathrm{max}}$. If $t_0 + t > t_{\mathrm{max}}$, run end-of-simulation rejection sampling starting from $\vc'$ to sample a configuration $\vc$ (see \cref{sec:end-batch-rejection-sampling}), and then output $\vc$ with all $K$ and $W$ removed. Otherwise, add $t_0$ to $t$.

    \item 
    \label{step:step3continuous}
    Repeat until $\ell_{\textrm{max}} = 0$: run \cref{alg:discrete-time-single-batch} on inputs $\cC'$, $\vc'$, and $\ell_{\mathrm{max}}$. Subtract the returned number of steps from $\ell_{\textrm{max}}$. Set $\vc'$ to the returned configuration. 

    \item 
    Set $\vc$ to $\vc'$ with all $K$ and $W$ removed.
    

\end{enumerate}
    
\end{algorithmic}
\end{algorithm}

\begin{lemrep}
    \label{lem:full-continuous-correctness}
    On input $\cC$, $\vc$, $v$, and $t_{\mathrm{max}}$, the output of \cref{alg:full-continuous-time-algorithm} is distributed as $\tstep{\cC}{\vc}{t_{\mathrm{max}}}{v}$.
\end{lemrep}
\begin{proof}
    We couple \cref{alg:full-continuous-time-algorithm} with the Gillespie algorithm run on $\cC'$ \emph{over the course of each individual loop iteration}.
    That is, 
    we ensure that \cref{alg:full-continuous-time-algorithm} and the Gillespie algorithm, 
    operating on the same configuration $\vc'$ described in step \ref{step:continuous-step-one} and both having simulated $t$ time so far,
    have the same probability of executing at least $\ell_{\mathrm{max}}$ reactions before time $t_{\mathrm{max}}$;
    we also ensure that their configuration distributions match in each case at the end of the iteration.
    This coupling implies that each iteration of \cref{alg:full-continuous-time-algorithm} matches the behavior of the Gillespie algorithm on $\vc'$, and by \cref{lem:null-reaction-simulation}, this is equivalent to the Gillespie algorithm run on $\cC$ over distributions of species other than $K$ and $W$, even in continuous time.
    That is, each individual iteration of \cref{alg:full-continuous-time-algorithm} that winds up simulating $t_0$ continuous time on a configuration $\vc$ of $\cC$,
    does exactly the same thing in distribution as the Gillespie algorithm simulating for time $t_0$ on $\vc$.
    Thus, \cref{alg:full-continuous-time-algorithm} as a whole has the same behavior as the Gillespie algorithm.
    
    The coupling, within a single iteration starting at time $t$, dictates the randomness of \cref{alg:full-continuous-time-algorithm} based on the randomness of the Gillespie algorithm.
    In particular, either the Gillespie algorithm simulates $l_{\mathrm{max}}$ total reactions within the remaining time $t_{\mathrm{max}} - t$, or it does not.
    The probability that the Gillespie algorithm does simulate this many reactions within this amount of time corresponds to the probability that the value $t_0$, sampled in step \ref{step:continuous-step-two}, satisfies $t + t_0 < t_{\mathrm{max}}$, 
    because this hypoexponential distribution describes exactly the amount of time that the Gillespie algorithm will take to simulate this many steps (which does not depend on which reactions occur, due to uniform reactivity of $\cC'$).
    If the Gillespie algorithm does simulate this many steps, and it takes time $t_0$ to do so, then step \ref{step:continuous-step-two} samples this value of $t_0$.
    In this case, \cref{lem:partial-discrete-correctness} implies that the coupling remains valid because \cref{alg:full-continuous-time-algorithm} runs reactions by running \cref{alg:discrete-time-single-batch}.
    If the Gillespie algorithm runs past time $t_{\mathrm{max}}$ before simulating this many reactions, 
    the coupling remains valid because rejection sampling (described in \cref{sec:end-batch-rejection-sampling}) correctly conditionally samples how many reactions to simulate.
\end{proof}

\begin{lemrep}
    \label{lem:full-continuous-efficiency}
    Suppose \cref{alg:full-continuous-time-algorithm} is run on input $\cC$, $v$, $\vc_0$, $t_{\mathrm{max}}$, and $p$, simulating an execution of $\cC$ that has $\ell$ reactions. 
    Suppose that every configuration $\vc$ appearing at the start of an iteration of this execution satisfies $n_{\mathrm{min}} \leq \|\vc\| \leq n_{\mathrm{max}}$
    and $\schedulerslowdown{\cC}{v}{\vc} \leq s$
    for some $n_\mathrm{min}, n_\mathrm{max} \in \N$, $s \in \R_{>0}$.
    Then the algorithm runs in time 
    $O\left(\frac{(2q)^{\ord(\cC)}s\ell\log(n_{\mathrm{max}})}{n_{\mathrm{min}}^p} + n_{\mathrm{max}}^{2p/3}\ell^{1/3}\log(n_{\mathrm{max}})\right)$.
\end{lemrep}

To break down this expression: $\ord(\cC)$ and $q$ are constants depending only on $\cC$.
$s$ is a slowdown factor due to effects of the specific configuration being simulated.
\opt{full,sub}{
Empirically, for example,
$s \leq 5$ for the Lotka-Volterra oscillator, as shown in \Cref{fig:lotka_volterra_plot_with_passive_reactions}.
}
On many reasonable inputs $n_{\mathrm{min}}$ and $n_{\mathrm{max}}$ will differ by a multiplicative constant, so we can treat them as being the same, and ignore the logarithmic factors. 
With these simplifications, we can express the runtime as roughly
\eq{\tilde{O}\left(n^{-p}\ell + n^{2p/3}\ell^{1/3}\right).}
The first term is the cost of sampling configurations in discrete time using batching.
It decreases as we increase $p$ toward $\frac 1 2$, because we are able to run more reactions in each batch.
In particular, the Gillespie algorithm fundamentally runs in time $\Omega(\ell)$, 
and the $n^{-p}$ factor represents the savings of batching.
The second term is the cost of adaptive rejection sampling to exactly sample inter-reaction times.
It decreases as we decrease $p$, because this allows us to sample more frequently from a less complex hypoexponential distribution,
allowing adaptive rejection sample to learn about the distribution more quickly.
Because of this, we can consider regimes comparing $\ell$ and $n$, and find the optimal asymptotic runtime of our algorithm in each regime by setting these two exponents equal to each other.
If $\ell \in \Omega(n^{5/4})$, the first term is dominant even when $p = \frac 1 2$, which is the largest value of $p$ that is beneficial (because there are typically $\Theta(n^{1/2})$ reactions between collisions). 
In this regime, the simulated execution is long enough that the adaptive rejection sampling algorithm has enough time to learn the hypoexponential distribution, and there is no asymptotic slowdown compared to our discrete time algorithm.
If $\ell \in \Theta(n)$, the asymptotically optimal value of $p$ is $2/5$, and our algorithm gives a speedup factor over the Gillespie algorithm of $n^{2/5}$.
We generally expect $\ell \in \Omega(n)$, because CRNs do not generally exhibit interesting behavior in a sublinear number of reactions.

\begin{proof}
    The first term comes from the runtime of step \ref{step:step3continuous}, which is shown in \cref{lem:discrete-efficiency-final}.
    There are two relevant differences from that lemma: first, $\sqrt n$ is replaced by $n_{\mathrm{min}}^p$.
    The change in exponent is because this factor comes from the number of reactions that are batched, and \cref{alg:full-continuous-time-algorithm} only calls \cref{alg:discrete-time-single-batch} to simulate $n^p$ reactions at a time.
    $n_{\mathrm{min}}$ is used so that the bound remains valid for every iteration.
    Second, up to half of the molecules in $\vc'$ might be $W$. 
    This might cause passive reactions to be simulated more often. 
    Because at most half of molecules are $W$ at the beginning of each batch, each molecule chosen as a reactant has probability at most $\frac 1 2$ of being $W$. 
    Thus, for each chosen reaction, with probability $\Omega(2^{-\ord(\cC)})$ its reactants contain no $W$. 
    Thus, the presence of $W$ cannot slow the simulation down by more than this factor, which we combine with the $q^{\ord(\cC)}$ term.
    
    What remains to show is the second term. 
    This term comes from step \ref{step:continuous-step-two}, which samples the hypoexponential distribution.
    To simulate $\ell$ reactions from a population size $n_0$, \cref{alg:full-continuous-time-algorithm} must sample a hypoexponential distribution with $n_0^p$ rates (representing the time to run $n_0^p$ reactions) a total of $\frac \ell {n_0^p}$ times.
    We use adaptive rejection sampling\cite{gilks1992adaptive}, which allows us to obtain these $\frac \ell {n_0^p}$ samples in $\sqrt[3]{\frac \ell {n_0^p}}$ evaluations of the hypoexponential PDF.
    To evaluate the hypoexponential PDF for the first time, we must first compute the values $C_{i,n_0^p}$ given in \cref{eqn:hypoexponential}, 
    \opt{full}{which can be done in time $\tilde{O}(n_0^p)$ as shown in \cref{sec:sampling-hypoexponential}.}
    \opt{sub}{which can be done in time $\tilde{O}(n_0^p)$ as shown in \cref{sec:apx-proof-of-c-i-k}.}
    Then, each subsequent evaluation takes time $O(n_0^p)$ to compute and sum $n_0^p$ terms.
    It follows that this process takes time $n_0^p\cdot \sqrt[3]{\frac \ell {n_0^p}} = n_0^{2p/3}\ell^{1/3}$ for a given hypoexponential distribution.
    We use $n_{\mathrm{max}}$ for $n_0$ as a worst-case bound.

    As molecular count changes, step \ref{step:continuous-step-one} may encounter different values of $i$ (and thus different values of $n_0$).
    This is necessary to allow molecular count to grow while ensuring that only a constant fraction of molecules are $W$ at the start of each batch.
    This requires sampling from as many as $\log(n_{\mathrm{max}})$ different hypoexponential distributions during step \ref{step:continuous-step-two}.
    This expression then becomes a multiplicative factor, as $n_{\mathrm{max}}^{2p/3}\ell^{1/3}$ bounds how long it will take to draw all the necessary samples from any one of these distributions.
\end{proof}

\input{sections/approximating-hypoexponential}

%% file: sections/approximating-hypoexponential.tex
\opt{sub}{
\Cref{sec:approximating-hypo}
explains how our implementation of the batching algorithm actually samples from a close approximation of the hypoexponential distribution,
much faster in practice.
It also justifies both theoretically and empirically that this does not change the sampled configuration distribution at all
(thus does not lead to qualitatively different behavior as with other speedup methods such as $\tau$-leaping),
and it changes the timestamp distributions negligibly.

}

\begin{toappendix}

\opt{sub}{
\section{Sampling approximately from the hypoexponential distribution}
}
\opt{full}{
\subsection{Sampling approximately from the hypoexponential distribution}
}

\label{sec:approximating-hypo}

Implementations of the algorithm described in \Cref{sec:sampling-hypoexponential,sec:end-batch-rejection-sampling} spend significant time sampling from the hypoexponential distribution.
Although the asymptotic performance is adequate amortized over many calls to the sampling procedure,
the constants involved make this a performance bottleneck in practice.

In this section we describe an alternative that is faster in practice and produces nearly identical outcomes.
We give up trying to sample from the hypoexponential distribution exactly,
instead sampling from the computationally much simpler 
\emph{gamma} distribution.
This technically means that in practice we are not sampling from precisely the same distribution of times as the Gillespie algorithm.
We emphasize that this is merely a slight imprecision in \emph{timestamps};
the sequence of configurations is still sampled from precisely the same distribution as Gillespie.
Thus this will not lead to inaccuracies in the sampled qualitative behavior of the CRN as with inexact methods such as $\tau$-leaping.
Furthermore, 
we justify that the measured deviations of timestamps from those of Gillespie will be negligible,
i.e., in any reasonably large population,
so small that it would not change a single pixel on a plot of counts over time.
Compared to $\tau$-leaping,
which gives worse approximations with larger $\tau$,
we have no time-accuracy tradeoff.
There are just some places where we observed a prohibitive computational cost in practice, despite the asymptotic performance shown in \Cref{sec:sampling-hypoexponential}.
The optimizations described in this section are therefore largely relevant to practical implementation,
so we focus on whether the approximation leads to inaccurate results in practice.

The approximations described in this section actually get better for larger population sizes $n$.
Thus in practice, for small $n$ one can sample the exact time distribution directly,
using the approximations described in this section only on $n$ sufficiently large that the approximation is so accurate that deviations from the true distribution are undetectable in practice.

Recall that the special case of the hypoexponential, 
where each of the $k \in \N^+$ exponentials being summed has identical rate $\lambda$, is called an \emph{Erlang}  distribution
$\mathbf{Er}(k,\lambda)$.
The gamma distribution $\mathbf{\Gamma}(\alpha,\lambda)$ generalizes $\mathbf{Er}(k,\lambda)$ to allow a real-valued first parameter $\alpha$, but coincides with $\mathbf{Er}(\alpha,\lambda)$ for positive integer $\alpha \in \N^+$.
The PDF of the Erlang has a term $(k-1)!$,
and the gamma function $\Gamma: \mathbb{C} \to \mathbb{C}$, 
defined on all complex numbers, 
has the property that
$\Gamma(k) = (k-1)!$ for all positive integers $k$.
The gamma distribution has the same PDF as the Erlang, but with $\Gamma(k)$ appearing in place of the $(k-1)!$ factorial in Erlang's PDF.

The reason we use the gamma distribution instead of the Erlang to approximate the hypoexponential distribution is that we use the method of 
\emph{moment matching}~\cite{bowman2004estimation},
finding the gamma distribution with the same mean (first moment) and variance (second moment) as the desired hypoexponential.
Furthermore, the rates of the exponentials defining the hypoexponential are very close to each other:
since $k = \Theta(\sqrt n)$ in expectation,
the ratio of the first and last terms of the sum defining the mean is very close to 1.
If the ratio were \emph{equal} to 1, 
then this would be a simple Erlang distribution.
The hypoexponential has many real parameters,
but Erlang only has 2, and the first is an integer.
Thus with the integer restriction, we cannot hope to find an Erlang matching both of these moments precisely.
Yet the \emph{gamma} distribution can match both,
since its two parameters are both real-valued, 
i.e., the gamma has the same number of degrees of freedom as the two moments we want to match.
Empirically in practice,
the real-valued \emph{shape} parameter of the gamma is very close to an integer,
so is ``almost'' an Erlang.

We note that it is possible to have a controllable approximation scheme that trades accuracy for speed.
The hypoexponential is defined as a sum of $k$ exponential random variables.
We could have a parameter $1 \leq c \leq k$ and sample
$c$ gamma random variables, with the $i$'th gamma having expected value to match the $i$'th block of $k/c$ exponentials.
In the case of $c=1$,
this is the approximation we described above.
In the case of $c=k$,
this samples exactly the hypoexponential by individually sampling the $k$ exponentials defining it.
In between,
as $c$ is larger, this is a better and better approximation.
However, in practice, we found that setting $c=1$ leads to a distribution essentially indistinguishable from the hypoexponential being approximating.

The hypoexponential in our case is parameterized by four parameters:

\begin{itemize}
\item 
    $n$: population size
\item
    $k$: number of interactions in a collision-free run
\item
    $o$: order of the CRN (number of reactants)
\item
    $g$: generativity of the CRN
    (number of products minus number of reactants)
\end{itemize}

Recall that the exponential random variables of the inter-reaction times have rates
$\binom{n}{o}$,
$\binom{n+g}{o}$,
$\binom{n+2g}{o}$,
$\dots$,
$\binom{n+(k-1)g}{o}$.
Those are the rates for the hypoexponential distribution giving the time of the entire batch of length $k$.
Such a distribution has
\[
\text{mean}
\qquad
\mu = 
\sum_{i=0}^{k-1} \frac{1}{\binom{n+ig}{o}}
\qquad
\text{and variance}
\qquad
\sigma^2 = 
\sum_{i=0}^{k-1} 
\frac{1}{\binom{n+ig}{o}^2}.
\]

Computing these directly would take time $\Omega(k)$,
defeating the goal of processing a batch of size $k$ in time $o(k).$
(Though in practice we do compute them directly for small values of $n$.)
The rest of \Cref{sec:approximating-hypo}
is devoted to showing that we can more efficiently compute the mean and variance of this hypoexponential distribution.

\opt{full}{
\subsubsection{Technical lemmas}
}
\opt{sub}{
\subsection{Technical lemmas}
}
\label{sec:polygamma-identities}

We first prove several technical lemmas involving identities that will be useful for computing both the mean and variance of a hypoexponential with the rates relevant to our batching algorithm.

\begin{lem}
\label{lem:inverse-binomial}
For all $B,o \in \N^+$ such that $o \le B$,
\[
    \frac{1}
    {\binom{B}{o}}
=
    o
    \cdot
    \sum_{m=0}^{o-1} 
    \binom{o-1}{m}
    \frac{(-1)^{m}}{B-(o-1-m)}.
\]
\end{lem}

\begin{proof}
    First,
    \[
    \frac{1}
    {\binom{B}{o}}
    =
    \frac{o! (B-o)!}{B!}
    =
    o
    \cdot
    \frac{(o-1)!}
    {\prod_{j=0}^{o-1} B-j}
    \]
    so it suffices to show
    \begin{equation}
    \label{eq:lem:inverse-binomial-suffice-to-prove}
    \frac{(o-1)!}
    {\prod_{j=0}^{o-1} B-j}
    =
    \sum_{m=0}^{o-1} 
    \binom{o-1}{m}
    \frac{(-1)^{m}}{(B-(o-1-m))}.
    \end{equation}
    
    Decomposing by partial fractions
    we can write:
    \begin{equation}
    \label{eq:poch-denominator}
    \frac{1}{\prod_{j=0}^{o-1} B-j} = \sum_{j=0}^{o-1} \frac{A_j}{B-j}.
    \end{equation}
    where the coefficients $A_j$ are given by the residue formula:
    \[
    A_j
    = \frac{1}{\prod_{\ell=0, \ell \neq j}^{o-1}(j-\ell)}
    = 
    \frac{1}{\prod_{\ell=0}^{j-1}(j-\ell) \cdot
    \prod_{\ell=j+1}^{o-1}(j-\ell)}
    =
    \frac{1}{j! \prod_{\ell=j+1}^{o-1}(j-\ell)}.
    \]
    
    For the product 
    $\prod_{\ell=j+1}^{o-1}(j-\ell)$,
    when $\ell$ takes values $j+1,j+2,\dots,o-1$, 
    then $(j-\ell)$ takes values
    $-1,-2,\dots,-(o-1-j)$.
    Conventionally factorial is not defined on negative integer values,
    but this product's absolute value is $(o-j-1)!$,
    equal to $(o-1-j)!$ if $o-1-j$ is even or $-((o-1-j)!)$ if $o-1-j$ is odd.
    Written differently,
    $\prod_{\ell=j+1}^{o-1}(j-\ell) = (-1)^{o-1-j} (o-1-j)!$.
    Thus we have
    \[
    A_j 
    = \frac{1}{j! (-1)^{o-1-j} (o-1-j)!}
    = \frac{(-1)^{o-1-j}}{j! (o-1-j)!}.
    \]
    Substituting into \Cref{eq:poch-denominator},
    \[
    \frac{1}{\prod_{j=0}^{o-1}(B-j)}
    = 
    \sum_{j=0}^{o-1} 
    \frac{(-1)^{o-1-j}}{j!(o-1-j)! (B-j)}
    = 
    \sum_{m=0}^{o-1} 
    \frac{(-1)^{m}}{m!(o-1-m)! (B-(o-1-m))},
    \]
    where the second equality follows by letting $m = o-1-j$
    (i.e., add the terms of the sum in the reverse order).
    Now multiply both sides by $(o-1)!$ to show \Cref{eq:lem:inverse-binomial-suffice-to-prove} holds:
    \[
    \frac{(o-1)!}{\prod_{j=0}^{o-1}(B-j)}
    = 
    \sum_{m=0}^{o-1} 
    \frac{(o-1)! (-1)^{m}}
    {m!(o-1-m)! (B-(o-1-m))}
    = 
    \sum_{m=0}^{o-1} 
    \binom{o-1}{m}
    \frac{(-1)^{m}}
    {B-(o-1-m)}.
    \qedhere
    \]
\end{proof}

For each $n \in \N$,
we let $\psi_n: \R \to \R$ denote the \emph{$n$'th polygamma} function with real arguments,
the $(n+1)$'st derivative of the ``log gamma'' function $\ln \Gamma(x)$.
In particular,
$\psi_0(x) = (\ln \Gamma(x))' = \frac{\Gamma(x)'}{\Gamma(x)}$ and 
$\psi_1 = (\ln \Gamma(x))''$ 
are respectively known as the \emph{digamma} and \emph{trigamma} functions,
where $f(x)'$ and $f(x)''$ respectively denote $\frac{df(x)}{dx}$ and $\frac{d^2f(x)}{dx^2}$.

The following is a technical lemma relating sums of a certain form to differences in the digamma function $\psi_0$.

\begin{lem}
\label{lem:diff-digamma-sum}
    For all $A,k,g \in \N^+$, 
    \[
    \sum_{i=0}^{k-1}
    \frac{g}{A+ig}
    =
    \psi_0\left(
        k + \frac{A}{g}
    \right)
    -
    \psi_0\left(
        \frac{A}{g}
    \right).
    \]
\end{lem}

\begin{proof}
    We use the identity~\cite{abramowitz1965handbook,qureshi2018analyticcomputationsdigammafunction}
    for $k \in \N^+$,
    \begin{equation}
    \label{eq:digamma-add-k}
    \psi_0(k+z) = 
    \sum_{i=0}^{k-1}
    \frac{1}{z+i}
    + \psi_0(z).
    \end{equation}
    Then
\begin{align*}
    \sum_{i=0}^{k-1}
    \frac{g}{A+ig}
=
    \sum_{i=0}^{k-1}
    \frac{1}{\frac{A+ig}{g}}
=
    \sum_{i=0}^{k-1}
    \frac{1}{\frac{A}{g}+i}
&=
    \sum_{i=0}^{k-1}
    \frac{1}{\frac{A}{g}+i}
    +
    \psi_0\left(
        \frac{A}{g}
    \right)
    -
    \psi_0\left(
        \frac{A}{g}
    \right)
\\&=
    \psi_0\left(
        k + \frac{A}{g}
    \right)
    -
    \psi_0\left(
        \frac{A}{g}
    \right)
\qquad \text{by \Cref{eq:digamma-add-k}.}
\qquad\qquad
\qedhere
\end{align*}
\end{proof}

The following similar technical lemma involves squaring the terms in the sum of \Cref{lem:diff-digamma-sum},
which turns out to be the difference of two \emph{trigamma} functions $\psi_1$.

\begin{lem}
\label{lem:diff-trigamma-sum}
    For all $A,k,g \in \N^+$, 
    \[
    \sum_{i=0}^{k-1}
    \left(
    \frac{g}{A+ig}
    \right)^2
    =
    \psi_1\left(
    \frac{A}{g}
    \right)
    -
    \psi_1\left(
    k+\frac{A}{g}
    \right)
    \]
\end{lem}
\begin{proof}
The trigamma function $\psi_1$ has infinite series expansion~\cite{lozier2003nist}
$
    \psi_1(z) = 
    \sum_{i=0}^\infty
    \frac{1}{(z+i)^2},
$
so
\begin{align*}
    \psi_1\left(
    \frac{A}{g}
    \right)
&=
    \sum_{i=0}^\infty
    \frac{1}{\left( \frac{A}{g} + i \right)^2}
=
    \sum_{i=0}^\infty
    \frac{1}{\left( \frac{A+ig}{g}\right)^2}
=
    \sum_{i=0}^\infty
    \left(
    \frac{g}{A+ig}
    \right)^2
\end{align*}
and similarly
\begin{align*}
    \psi_1\left(
    k + \frac{A}{g}
    \right)
&=
    \sum_{i=0}^\infty
    \frac{1}{\left( \frac{A}{g} + k + i \right)^2}
=
    \sum_{i=0}^\infty
    \left(
    \frac{g}{(A+(k+i)g)}
    \right)^2
=
    \sum_{m=k}^\infty
    \left(
    \frac{g}{A+mg}
    \right)^2,
\end{align*}
where the last re-indexes the sum letting $m=k+i$.
Now observe
\begin{align*}
    \psi_1\left(
    \frac{A}{g}
    \right)
    -
    \psi_1\left(
    k + \frac{A}{g}
    \right)
&=
    \sum_{i=0}^\infty
    \left(
    \frac{g}{A+ig}
    \right)^2
    -
    \sum_{m=k}^\infty
    \left(
    \frac{g}{A+mg}
    \right)^2
=
    \sum_{i=0}^{k-1}
    \left(
    \frac{g}{A+ig}
    \right)^2.
    \qquad\qquad
    \qedhere
\end{align*}
\end{proof}

\begin{lem}
\label{lem:diff-digamma-sum-two-factors-denominator}
For all $A,B,k,g \in \N^+$ with $A \neq B$,
\[
    \sum_{i=0}^{k-1}
    \frac{1}
    {(A+ig) (B+ig)}
    =
    \frac{1}{g(A-B)} \cdot
    \left[
        \psi_0\left(
            \frac{A}{g}
        \right)
        -
        \psi_0\left(
            k + \frac{A}{g}
        \right)
        -
        \psi_0\left(
            \frac{B}{g}
        \right)
        +
        \psi_0\left(
            k + \frac{B}{g}
        \right)
    \right].
\]
\end{lem}

\begin{proof}
    Note that by partial fraction decomposition,
    \[
    \frac{1}{B+ig}
    -
    \frac{1}{A+ig}
    =
    \frac{A+ig - (B+ig)}
    {(A+ig)(B+ig)}
    =
    \frac{A-B}
    {(A+ig)(B+ig)},
    \]
    so dividing both sides by $A-B$:
    \[
    \frac{1}{(A+ig)(B+ig)}
    =
    \frac{1}{A-B}
    \left(
        \frac{1}{B+ig}
        -
        \frac{1}{A+ig}
    \right).
    \]
    Thus
\begin{align*}
    \sum_{i=0}^{k-1}
    \frac{1}
    {(A+ig) (B+ig)}
&=
    \sum_{i=0}^{k-1}
    \frac{1}{A-B}
    \left(
        \frac{1}{B+ig}
        -
        \frac{1}{A+ig}
    \right)
\\&=
    \frac{1}{A-B} \cdot
    \left[
        \sum_{i=0}^{k-1}
        \frac{1}{B+ig}
        -
        \sum_{i=0}^{k-1}
        \frac{1}{A+ig}
    \right]
\\&=
    \frac{1}{g(A-B)} \cdot
    \left[
        \sum_{i=0}^{k-1}
        \frac{g}{B+ig}
        -
        \sum_{i=0}^{k-1}
        \frac{g}{A+ig}
    \right]
\\&=
    \frac{1}{g(A-B)} \cdot
    \left[
        \psi_0\left(
            k + \frac{B}{g}
        \right)
        -
        \psi_0\left(
            \frac{B}{g}
        \right)
        -
        \psi_0\left(
            k + \frac{A}{g}
        \right)
        +
        \psi_0\left(
            \frac{A}{g}
        \right)
    \right],
\end{align*}
where the final equality follows by applying \Cref{lem:diff-digamma-sum} to each sum.
The lemma follows by rearranging terms in the brackets.
\end{proof}

\opt{full}{
\subsubsection{Computing mean and variance of hypoexponentials}
}
\opt{sub}{
\subsection{Computing mean and variance of hypoexponentials}
}

The mean of a hypoexponential distribution defined as a sum of exponential random variables with rates 
$\binom{n}{o}$,
$\binom{n+g}{o}$,
$\binom{n+2g}{o}$,
$\dots$,
$\binom{n+(k-1)g}{o}$,
i.e., with means equal to the reciprocals of those rates,
by linearity of expectation,
is 
$\sum_{i=0}^{k-1} 
\frac{1}{\binom{n+ig}{o}}$.
Calculating this directly by computing the sum takes time $\Omega(k)$,
but we need to do this when processing a batch of size $k$,
and the entire point of the algorithm is to use much less than time $k$ to process $k$ reactions.

The next lemma allows us to compute the mean of a such a hypoexponential with $k$ terms by computing a sum with only $o$ terms.
This is significant because we'll think of $o$ (the maximum number of reactants in any reaction in the original CRN) as a small constant,
whereas the na\"{i}ve way to calculate the mean 
(the left side of the equation of \Cref{lem:mean-hypo-formula})
would require time $\Theta(k)$,
where $k \gg o$,
to sum all expected values $\frac{1}{\binom{n+ig}{o}}.$
The identities proven in \Cref{sec:polygamma-identities} involving the digamma $\psi_0$ and trigamma $\psi_1$ functions will be used;
fortunately, algorithms exist to compute these functions in time $O(1)$~\cite{spouge1994computation}.\footnote{
    More precisely in time depending only on the desired relative error,
    but independent of the magnitude of the argument.
    See also~\cite{bernardo1976algorithm,schwachheim1969algorithm,mpmath}.
}

We note that in actual implementation, these identities are extremely sensitive to floating-point rounding errors.
In particular, even if each term $t$ in the sum of, e.g., \Cref{lem:mean-hypo-formula}
is ``moderately sized'', i.e., within a few orders of magnitude of 1,
in general there are pairs of terms, e.g., 
$10^{-2} \le t_1,t_2 \le 10^2$,
such that 
$|t_1 - t_2| \ll 10^{-15}$, 
i.e., 
the terms suffer ``catastrophic cancellation'' in which pairwise differences are much smaller than the precision of standard floating-point arithmetic.\footnote{
    Although IEEE 64-bit double-precision floating point numbers can be as small as $10^{-308}$,
    taking the difference of two floats close to 1, 
    since they use about 15 digits of precision, can only represent differences between such numbers as small as $10^{-15},$
    e.g.,
    $1.000000000000000005 - 
     1.000000000000000004$
    is equal to
    $0.000000000000000001,$
    yet the above expression evaluated with double-precision floats evaluates to $0.0$ since both the float literals $1.000000000000000005$
    and
    $1.000000000000000004$
    evaluate to $1.0$.
    So in evaluating the terms of the sums such as in \Cref{lem:mean-hypo-formula},
    we must take care that the individual terms are evaluated with sufficient precision that two different opposite-sign terms with very close absolute values are not rounded to have identical absolute values.
}
Thus, when computing these sums, 
it is necessary to use arbitrary precision libraries such as Python's \texttt{mpmath} package~\cite{mpmath}
to avoid such catastrophic cancellation errors.

\begin{lem}
\label{lem:mean-hypo-formula}
For all $n,k,o,g \in \N^+$,
where $o \le n$,
\[
    \sum_{i=0}^{k-1}
    \frac{1}{\binom{n+ig}{o}}
=
    \frac{o}{g} \cdot
    \sum_{m=0}^{o-1}
    (-1)^m
    \cdot
    \binom{o-1}{m}
    \cdot
\left[ 
    \psi_0\left(
        k + \frac{n - (o-1-m)}{g}
    \right)
-
    \psi_0\left(
        \frac{n - (o-1-m)}{g}
    \right)
\right].
\]
\end{lem}

\begin{proof}
    Letting $B = n+ig$ in \Cref{lem:inverse-binomial},
    we have
    \begin{align*}
    \sum_{i=0}^{k-1}
    \frac{1}
    {\binom{n+ig}{o}}
&=
    \sum_{i=0}^{k-1}
    o
    \cdot
    \sum_{m=0}^{o-1}
    \binom{o-1}{m}
    \frac{(-1)^{m}}{n+ig-(o-1-m)}
\qquad \text{by \Cref{lem:inverse-binomial}}
\\&=
    o
    \cdot
    \sum_{m=0}^{o-1}
    (-1)^m
    \cdot
    \binom{o-1}{m}
    \cdot
    \sum_{i=0}^{k-1}
    \frac{1}{n-(o-1-m)+ig}
\\&=
    \frac{o}{g}
    \cdot
    \sum_{m=0}^{o-1}
    (-1)^m
    \cdot
    \binom{o-1}{m}
    \cdot
    \sum_{i=0}^{k-1}
    \frac{g}{n-(o-1-m)+ig}
\\&=
    \frac{o}{g}
    \cdot
    \sum_{m=0}^{o-1}
    (-1)^m
    \cdot
    \binom{o-1}{m}
    \cdot
\left[ 
    \psi_0\left(
        k + \frac{n - (o-1-m)}{g}
    \right)
-
    \psi_0\left(
        \frac{n - (o-1-m)}{g}
    \right)
\right],
\end{align*}
where the last equality follows from \Cref{lem:diff-digamma-sum} with $A = n - (o-1-m)$.
\end{proof}

To analyze the time complexity of computing the right-hand side of \Cref{lem:mean-hypo-formula},
observe that the binomial coefficients can be computed iteratively while evaluating the sum,
via the identity
$\binom{o-1}{m} = \binom{o-1}{m-1} \cdot \frac{o-m}{m}$,
so that the entire sum requires time $O(o)$ to compute,
since $\psi_0$ can be computed in time $O(1)$.

Similarly to the mean,
by linearity of variance when the random variables are independent,
the variance of a hypoexponential distribution,
defined by rates 
$\binom{n}{o}$,
$\binom{n+g}{o}$,
$\binom{n+2g}{o}$,
$\dots$,
$\binom{n+(k-1)g}{o}$,
is 
$\sum_{i=0}^{k-1} 
\frac{1}{\binom{n+ig}{o}^2}$.
The next lemma, 
similarly to \Cref{lem:mean-hypo-formula},
allows the variance to be computed in time $O(o^2)$.

\begin{lem}
\label{lem:variance-hypo-formula}
    For all $n,k,o,g \in \N^+$,
    where $o \le n$,
\begin{align*}
    \sum_{i=0}^{k-1} \frac{1}{\binom{n+ig}{o}^2}
=&
\notag
    \frac{o^2}{g^2}
    \sum_{m=0}^{o-1}
    \binom{o-1}{m}^2
    \left[
    \psi_1\left(
        \frac{n-(o-1-m)}{g}
    \right)
    -
    \psi_1\left(
        k + \frac{n-(o-1-m)}{g}
    \right)
    \right]
\\&+
    \frac{2 o^2}{g}
    \sum_{m=0}^{o-1}
    \sum_{j=m+1}^{o-1}
    \frac{(-1)^{m+j}}{m-j}
    \binom{o-1}{m}
    \binom{o-1}{j}
\\&
\qquad\qquad\qquad\qquad
\cdot
\notag
    \left[
        \psi_0\left(
            \frac{n-(o-1-m)}{g}
        \right)
        -
        \psi_0\left(
            k + \frac{n-(o-1-m)}{g}
        \right)
    \right.
\\&\qquad\qquad\qquad\qquad \ \ \ -
    \left.
        \psi_0\left(
            \frac{n-(o-1-j)}{g}
        \right)
        +
        \psi_0\left(
            k + \frac{n-(o-1-j)}{g}
        \right)
    \right].
\end{align*}
\end{lem}

\begin{proof}
    Letting $B = n+ig$ in \Cref{lem:inverse-binomial},
    we have
    \begin{equation}
    \label{eq:var-hypo-first-eq}
    \sum_{i=0}^{k-1}
    \frac{1}
    {\binom{n+ig}{o}^2}
    =
    \sum_{i=0}^{k-1}
    o^2
    \cdot
    \left(
    \sum_{m=0}^{o-1}
    \binom{o-1}{m}
    \frac{(-1)^{m}}
    {n+ig-(o-1-m)}
    \right)^2.
    \end{equation}
    Recall the identity\footnote{
        For example,
        $(a+b+c)^2
        =
        (a^2 + ab + ac)
        + (ab + b^2 + bc)
        + (ac + bc + c^2)
        = 
        (a^2 + b^2 + c^2)
        + 2(ab + ac + bc).
        $
    }
    \begin{equation}
    \label{eq:square-of-sums}
        \left(
        \sum_{m=0}^{o-1}
        x_m
        \right)^2
    =
        \sum_{m=0}^{o-1}
        x_m^2
        +
        2 \cdot
        \sum_{m=0}^{o-1}
        \sum_{j=m+1}^{o-1}
        x_m \cdot x_j.
    \end{equation}
    By \Cref{eq:square-of-sums}, 
    we can write the squared sum in 
    \Cref{eq:var-hypo-first-eq} as
\begin{align*}
&
    \left(
    \sum_{m=0}^{o-1}
    \binom{o-1}{m}
    \frac{(-1)^{m}}{n+ig-(o-1-m)}
    \right)^2
\\=&
    \sum_{m=0}^{o-1}
        \left[
        \binom{o-1}{m}
        \frac{(-1)^{m}}{n+ig-(o-1-m)}
        \right]^2
\\&+
    2 \cdot
    \sum_{m=0}^{o-1}
    \sum_{j=m+1}^{o-1}
    \binom{o-1}{m}
    \frac{(-1)^{m}}{n+ig-(o-1-m)}
    \binom{o-1}{j}
    \frac{(-1)^{j}}{n+ig-(o-1-j)}
\\=&
    \sum_{m=0}^{o-1}
    \binom{o-1}{m}^2
    \left(
    \frac{1}{n-(o-1-m)+ig}
    \right)^2
\qquad \text{$((-1)^m)^2 = 1$ for all $m \in \N^+$}
\\&+
    2 
    \sum_{m=0}^{o-1}
    \sum_{j=m+1}^{o-1}
    \binom{o-1}{m}
    \binom{o-1}{j}
    \frac{(-1)^{m+j}}
    {(n-(o-1-m)+ig) (n-(o-1-j)+ig)}.
\end{align*}
    Substituting this back into the right side of
    \Cref{eq:var-hypo-first-eq},
\begin{align*}
&
    \sum_{i=0}^{k-1}
    o^2
    \left(
    \sum_{m=0}^{o-1}
    \binom{o-1}{m}
    \frac{(-1)^{m}}{n+ig-(o-1-m)}
    \right)^2
\\=&
    \sum_{i=0}^{k-1}
    o^2 
    \sum_{m=0}^{o-1}
    \binom{o-1}{m}^2
    \left(
    \frac{1}{n-(o-1-m)+ig}
    \right)^2
\\&+
    \sum_{i=0}^{k-1}
    2 o^2
    \sum_{m=0}^{o-1}
    \sum_{j=m+1}^{o-1}
    \binom{o-1}{m}
    \binom{o-1}{j}
    \frac{(-1)^{m+j}}
    {(n-(o-1-m)+ig) 
     (n-(o-1-j)+ig)}
\\=&
    o^2 
    \sum_{m=0}^{o-1}
    \binom{o-1}{m}^2
    \cdot
    \sum_{i=0}^{k-1}
    \left(
    \frac{1}{n-(o-1-m)+ig}
    \right)^2
\\&+
    2 o^2
    \sum_{m=0}^{o-1}
    \sum_{j=m+1}^{o-1}
    (-1)^{m+j}
    \binom{o-1}{m}
    \binom{o-1}{j}
    \sum_{i=0}^{k-1}
    \frac{1}
    {(n-(o-1-m)+ig) 
     (n-(o-1-j)+ig)}.
\end{align*}
    The first summation can be written
\begin{align}
&
\notag
    o^2 
    \sum_{m=0}^{o-1}
    \binom{o-1}{m}^2
    \cdot
    \sum_{i=0}^{k-1}
    \left(
    \frac{1}{n-(o-1-m)+ig}
    \right)^2
\\=&
\notag
    \frac{o^2}{g^2}
    \sum_{m=0}^{o-1}
    \binom{o-1}{m}^2
    \cdot
    \sum_{i=0}^{k-1}
    \left(
    \frac{g}{n-(o-1-m)+ig}
    \right)^2
\\=&
    \frac{o^2}{g^2}
    \sum_{m=0}^{o-1}
    \binom{o-1}{m}^2
    \left[
    \psi_1\left(
        \frac{n-(o-1-m)}{g}
    \right)
    -
    \psi_1\left(
        k + \frac{n-(o-1-m)}{g}
    \right)
    \right],
\label{eq:lem:variance-hypo-formula-first-sum}
\end{align}
applying \Cref{lem:diff-trigamma-sum} with $A = n-(o-1-m)$ for the last equality.
Applying \Cref{lem:diff-digamma-sum-two-factors-denominator}
with $A=n-(o-1-m)$ and $B=n-(o-1-j)$,
the other terms can be written
\begin{align}
&
\notag
    2 o^2
    \sum_{m=0}^{o-1}
    \sum_{j=m+1}^{o-1}
    (-1)^{m+j}
    \binom{o-1}{m}
    \binom{o-1}{j}
    \sum_{i=0}^{k-1}
    \frac{1}
    {(n-(o-1-m)+ig) 
     (n-(o-1-j)+ig)}
\\=&
\notag
    2 o^2
    \sum_{m=0}^{o-1}
    \sum_{j=m+1}^{o-1}
    (-1)^{m+j}
    \binom{o-1}{m}
    \binom{o-1}{j}
    \frac{1}
    {g [n-(o-1-m) - (n-(o-1-j))]}
\\&\qquad\qquad\qquad\cdot
\notag
    \left[
        \psi_0\left(
            \frac{n-(o-1-m)}{g}
        \right)
        -
        \psi_0\left(
            k + \frac{n-(o-1-m)}{g}
        \right)
    \right.
\\&\qquad\qquad\qquad\ \ \ -
\notag
    \left.
        \psi_0\left(
            \frac{n-(o-1-j)}{g}
        \right)
        +
        \psi_0\left(
            k + \frac{n-(o-1-j)}{g}
        \right)
    \right]
\\=&
\notag
    \frac{2 o^2}{g}
    \sum_{m=0}^{o-1}
    \sum_{j=m+1}^{o-1}
    \frac{(-1)^{m+j}}{m-j}
    \binom{o-1}{m}
    \binom{o-1}{j}
\\&\qquad\qquad\qquad\cdot
\notag
    \left[
        \psi_0\left(
            \frac{n-(o-1-m)}{g}
        \right)
        -
        \psi_0\left(
            k + \frac{n-(o-1-m)}{g}
        \right)
    \right.
\\&\qquad\qquad\qquad\ \ \ -
    \left.
        \psi_0\left(
            \frac{n-(o-1-j)}{g}
        \right)
        +
        \psi_0\left(
            k + \frac{n-(o-1-j)}{g}
        \right)
    \right].
\label{eq:lem:variance-hypo-formula-second-sum}
\end{align}
    The lemma follows by adding
    \eqref{eq:lem:variance-hypo-formula-first-sum}
    and \eqref{eq:lem:variance-hypo-formula-second-sum}.
\end{proof}

\opt{full}{
\subsubsection{Approximating mean and variance of hypoexponentials}
}
\opt{sub}{
\subsection{Approximating mean and variance of hypoexponentials}
}
The identities of \Cref{lem:mean-hypo-formula,lem:variance-hypo-formula}
give a way to compute the mean and variance of our hypoexponential distribution,
defined by a sum with $k$ terms,
in time $o(k)$.
Nevertheless, since empirical testing has shown that these identities require computing the digamma and trigamma functions at higher floating-point precision than standard 64-bit double precision,
there are significant constant factors associated with this approach.
In practice we actually approximate the mean and variance in the following much faster way.
Recall the mean of a hypoexponential defined by rates $\binom{n}{o}, \binom{n+g}{o}, \binom{n+2g}{o}, \dots, \binom{n+(k-1)g}{o}$ is $\sum_{i=0}^{k-1} \frac{1}{\binom{n+ig}{o}}$.

Define the \emph{relative error} between numbers $a$ and $b$ to be $\frac{|a-b|}{\min(|a|,|b|)}$.
Consider the first term $t_1 = \frac{1}{\binom{n}{o}}$ and last term $t_k = \frac{1}{\binom{n+(k-1)g}{o}}$ of the sum.
We compute the relative error between $t_1$ and $t_k$, and if it is is less than $0.1$,\footnote{
    Note that the larger is $n$,
    since $k \approx \sqrt{n}$,
    we expect the relative error between $t_1$ and $t_k$ to be quite small,
    since it converges to 0 as $n \to \infty$.
    For example,
    if $n=10^6, k = \sqrt{n} = 10^3, o=2, g=1$,
    then the relative error between $t_1$ and $t_k$ is $\approx 0.002$.
}
we approximate the sum by using geometric mean of these terms:
\[
\sum_{i=0}^{k-1} \frac{1}{\binom{n+ig}{o}}
\approx
k \cdot \sqrt{\frac{1}{\binom{n}{o}} \cdot \frac{1}{\binom{n+(k-1)g}{o}}}.
\]

The following lemma justifies this approximation, 
establishing a formal relationship between the relative error of the first and last terms $t_1,t_k$ and the relative error between the actual sum and the approximation given using the geometric mean described above.

\begin{lem}
\label{lem:relative-error-geometric-mean-for-mean-of-hypo}
    Let $\delta > 0$.
    If the relative error between 
    $t_1 = 1/{\binom{n}{o}}$ and
    $t_k = 1/{\binom{n+(k-1)g}{o}}$ is $\delta$,
    then the relative error between 
    $\sum_{i=0}^{k-1} \frac{1}{\binom{n+ig}{o}}$ and 
    $k \sqrt{t_1 t_k}$
    is at most 
    $\frac{\delta}{2} + \frac{\delta^2}{8}.$
\end{lem}

\begin{proof}
Let the terms be $t_i = 1 / \binom{n+(i-1)g}{o}$ for $i \in \{1,\dots,k\}$,
so the sum is $\sum_{i=1}^k t_i$.
Note that $t_1 > t_i > t_k$ for all $1 < i < k$;
in the remainder of the proof, we use only this fact.

The relative error between $t_1$ and $t_k$ is then $\delta = \frac{t_1-t_k}{t_k} = \frac{t_1}{t_k} - 1$,
implying $t_1 = t_k(1 + \delta)$.
Let $A = \frac{1}{k} \sum_{i=1}^k t_i$ be the arithmetic mean of the terms $t_1,\dots,t_k$ (note $kA$ is exactly their sum),
and let $G = \left( \prod_{i=1}^k t_i \right)^{1/k}$ be the geometric mean of all $k$ terms.
Let $G_{1,k} = \sqrt{t_1 t_k},$ the geometric mean of the first and last term only.
We first bound $|A-G|$, then $|G - k G_{1,k}|$.

From~\cite{cartwright1978refinement}, with each $p_i = 1/k$ (or $1/n$ as stated in~\cite{cartwright1978refinement}; $n$ there is $k$ in this proof) gives
\[
A - G \leq \frac{1}{2 t_k} \sum_{i=1}^k \frac{1}{k} \left( t_i - \sum_{j=1}^k \frac{1}{k} t_j \right)^2.
\]
Note that the outer sum on the right side is the sample variance $\Var{t_1,t_2,\dots,t_k}$ of the terms,
which we denote as $\Var{t}$, so we can write
\begin{equation}
\label{eq:am-gm-eq1}
A-G \le \frac{\Var{t}}{2 t_k}.
\end{equation}
In the worst case,
subject to only this constraint,
the variance of the terms is maximized when half the terms are $t_k$ and the other half are $t_1$.
In that case, writing $\mu_{1,k} = (t_1 + t_k) / 2$ for the (arithmetic) mean of those two terms, 
the variance would be
\[
\frac{1}{k} \left[ 
    \frac{k}{2} (t_i - \mu_{1,k})^2 + 
    \frac{k}{2} (t_k - \mu_{1,k})^2 
\right]
= \frac{1}{4} (t_1 - t_k)^2,
\]
which follows by substituting $\mu_{1,k} = (t_1 + t_k) / 2$ and algebraic simplification.
Thus $\Var{t} \leq \frac{1}{4} (t_1 - t_k)^2.$
Since $t_1 = t_k(1+\delta)$,
we have
$
    t_1 - t_k
    =
    t_k (1+\delta) - t_k
    =
    t_k \delta,
$
Thus
\[
\Var{t} 
\le \frac{1}{4} (t_1 - t_k)^2
= \frac{t_k^2 \delta^2}{4}.
\]
Substituting in \eqref{eq:am-gm-eq1} gives
$A - G \leq \frac{t_k\delta^2}{8}$.

Now we bound the geometric mean $G$ of all terms using the geometric mean $G_{1,k} = \sqrt{t_1 t_k}$ of the first and last terms.
Recall $t_1 = t_k(1 + \delta)$,
so
\[
G_{1,k} 
= \sqrt{t_1 t_k} 
= \sqrt{t_k^2 (1+\delta)}
= t_k \sqrt{1+\delta}.
\]
Note that $t_k \le G \le t_1 = t_k(1+\delta)$.
If $G=t_k$,
then
\[
|G-G_{1,k}|
= G_{1,k} - G
= t_k \sqrt{1+\delta} - t_k
= t_k (\sqrt{1+\delta} - 1)
\]
and if $G = t_k(1+\delta)$, then
\[
|G-G_{1,k}|
= G - G_{1,k}
= t_k \sqrt{1+\delta} - t_k
= t_k(1+\delta) - t_k \sqrt{1+\delta}
= \frac{t_k (\sqrt{1+\delta} - 1)}{\sqrt{1+\delta}}.
\]
Since $\sqrt{1+\delta} > 1$ for all $\delta>0,$
this means the first case is larger, so in the worst case,
\[
|G - G_{1,k}|
\le t_k (\sqrt{1+\delta} - 1)
\]
Note that for all $\delta > 0$,
we have $\sqrt{1+\delta} < 1 + \frac{\delta}{2}$.
Combining the above bounds on $A-G$ and $|G-G_{1,k}|$ to bound $|A - G_{1,k}|$ by the triangle inequality:
\[
|A - G_{1,k}|
\leq |A - G| + |G - G_{1,k}|
\le \frac{t_k \delta^2}{8} + t_k \left( \sqrt{1+\delta} - 1 \right)
=   \frac{t_k \delta^2}{8} + t_k \left(1 + \frac{\delta}{2} - 1 \right)
=   t_k \left( \frac{\delta}{2} + \frac{\delta^2}{8} \right).
\]
Recall the sum $S = \sum_{i=1}^k t_i = k A$.
This implies our approximation by $k G_{1,k}$ has absolute error
\[
|S - k G_{1,k}| \le k t_k \left( \frac{\delta}{2} + \frac{\delta^2}{8} \right).
\]
Recall $\sqrt{1+\delta} > 1$ for all $\delta > 0$.
So the absolute error above implies relative error
\begin{align*}
\frac{|S - k G_{1,k}|}{k G_{1,k}}
&\le
\frac{k t_k \left( \frac{\delta}{2} + \frac{\delta^2}{8} \right)}{k \sqrt{t_1 t_k}}
=
\frac{t_k \left( \frac{\delta}{2} + \frac{\delta^2}{8} \right)}{\sqrt{t_k^2 (1+\delta)}}
=
\frac{ \frac{\delta}{2} + \frac{\delta^2}{8}}{\sqrt{1+\delta}}
<
\frac{\delta}{2} + \frac{\delta^2}{8}.
\qquad\qquad\qquad
\qedhere
\end{align*}
\end{proof}

We in fact conjecture that the relative error of the approximation of the sum by $k G_{1,k}$ is $O(\delta^2)$ (better than merely $O(\delta)$ as in \Cref{lem:relative-error-geometric-mean-for-mean-of-hypo}), which appears to be the case empirically.
For example, when the relative error of $t_1$ and $t_k$ is $0.1$, in practice the relative error of the approximation and the true mean appears to be $< 0.01$.
However, we have not been able to prove this.
It would likely require using more information about the distribution of the terms $t_i$;
under our worst case assumption in the proof that $G$ could be as small as $t_k$
(i.e, assuming all terms in the sum are $t_k$)
or as large as $t_1$
(i.e, assuming all terms in the sum are $t_1$),
the bound proven seems asymptotically tight without additional constraints on the terms.

We similarly approximate the variance $\sum_{i=1}^k t_i^2$ via $k \sqrt{t_1^2 t_k^2} = kt_1 t_k$,
though we omit a detailed analysis.
A similar proof to that of \Cref{lem:relative-error-geometric-mean-for-mean-of-hypo} shows that it is also bounded by the relative error between the first and last terms.
Since these terms are squared and significantly less than 1,
the relative error between $t_1^2$ and $t_k^2$ is even smaller than in the case of approximating the mean,
making this an even tighter approximation of the variance than that of the mean in \Cref{lem:relative-error-geometric-mean-for-mean-of-hypo}.
In practice, these approximations appear to result in sampling CRN trajectories indistinguishable from the Gillespie algorithm.

\end{toappendix}

%% file: sections/data.tex
\begin{toappendix}
\section{Simulation data}
\label{sec:data}

\subsection{Empirical sampling of data for some CRNs}

This section shows data for a few CRNs that are best suited for using our batching algorithm;
namely they have positive generativity (hence cannot be simulated by the original batching algorithm) but are order 2, hence have the fewest problems with passive reactions as described in \Cref{alg:null-reactions}.
More formally, CRNs with smaller values of $\schedulerslowdown{\cC}{v}{\vc}$ as defined in~\Cref{def:efficiency-condition}.

\begin{figure}
\centering
\includegraphics
[width=0.6\textwidth]{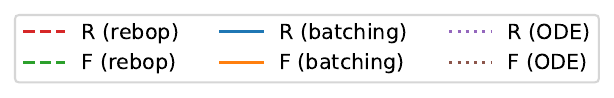}
\\
\begin{subfigure}{0.49\textwidth}
    \centering
    \includegraphics[width=\textwidth]{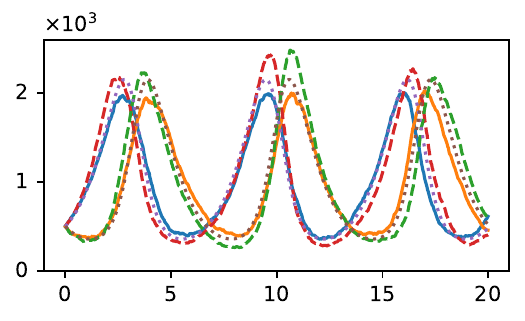}
    \caption{initial $n=10^3$}
    \label{fig:lotka_volterra_counts_time20_n1e3}
\end{subfigure}
\hfill
\begin{subfigure}{0.49\textwidth}
    \centering
    \includegraphics[width=\textwidth]{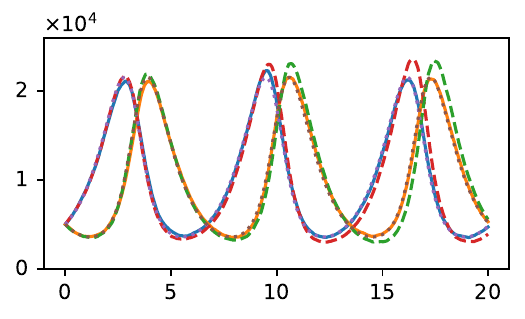}
    \caption{initial $n=10^4$}
    \label{fig:lotka_volterra_counts_time20_n1e4}
\end{subfigure}
\begin{subfigure}{0.49\textwidth}
    \centering
    \includegraphics[width=\textwidth]{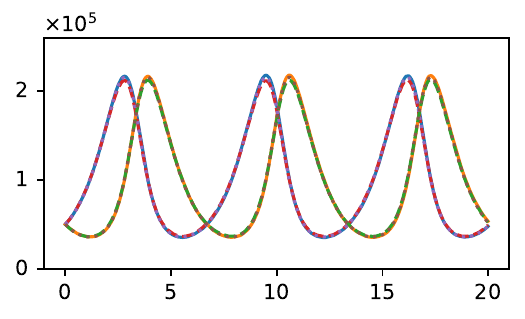}
    \caption{initial $n=10^5$}
    \label{fig:lotka_volterra_counts_time20_n1e5}
\end{subfigure}
\hfill
\begin{subfigure}{0.49\textwidth}
    \centering
    \includegraphics[width=\textwidth]{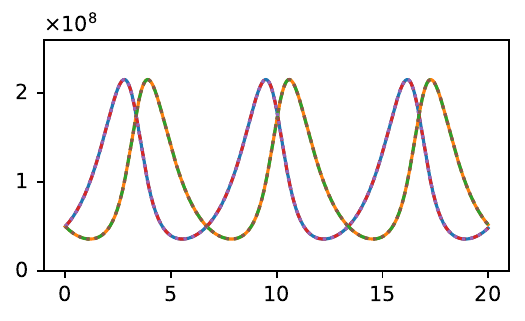}
    \caption{initial $n=10^8$}
    \label{fig:lotka_volterra_counts_time20_n1e8}
\end{subfigure}
\caption{
Plots of counts vs.~time for the Lotka-Volterra oscillator CRN,
for initial population size $n \in \{10^3,10^4,10^5,10^8\}$
with half predator,
half prey,
using both rebop (Gillespie algorithm) and our batching algorithm implementation. Stochastic effects make the plots behave differently for small $n$,
as well as differently from a deterministic ODE approximation,
also plotted.
As $n$ increases, stochastic noise decreases, and both stochastic simulators generate trajectories approaching the deterministic model.
At $n=10^8$ all three plots are nearly indistinguishable.}
\label{fig:lotka_volterra_counts}
\end{figure}

As in \Cref{sec:compare-running-time-batching-gillespie},
when we compare our batching algorithm to Gillespie,
we use rebop~\cite{rebop} as the fastest Gillespie algorithm implementation that we could find.\footnote{
    Rebop is \emph{far} faster than most Gillespie implementations: \url{https://github.com/Armavica/rebop\#performance}
}
We note that rebop does not directly support the concept of volume $v$; it implicitly assumes $v=1$.
Recall in \Cref{sec:prelim} that each reaction with $o$ total reactants has a term $1 / v^{o-1}$ in its rate.
So for proper comparison,
we manually adjust the rebop rate constants in this way
(i.e., divide order-$o$ reaction rate constants by $v^{o-1}$)
so that rebop's reactions have the same total rate as reactions in our batching algorithm.

We note that deviations in simulated trajectories do not imply that the batching algorithm is sampling from the wrong distribution.
Rather, in all cases this is simply stochasticity of the Gillespie model itself;
see \Cref{fig:dimerization_comparison} for further empirical justification of the claim that our algorithm samples from  the Gillespie distribution.

\Cref{fig:lotka_volterra_counts} shows simulations of the Lotka-Volterra chemical oscillator~\cite{lotka1910contribution,volterra1926variazioni}, a.k.a., predator-prey oscillator:
the reactions 
\begin{align*}
R   &\rxn^1 2R
\\
F   &\rxn^1 \emptyset
\\
F+R &\rxn^1 2F
\end{align*}
starting with equal $R$ (rabbits/prey) and $F$ (foxes/predators).
    Although devised originally by Lotka to study chemical reaction networks with autocatalytic reactions~\cite{lotka1910contribution},
    this model was independently devised by Volterra~\cite{volterra1926variazioni} to model actual animal populations that were observed empirically to oscillate over the timescale of years.
    The intuition is that rabbits always have plenty of plants to eat, so constantly reproduce 
    ($R \rxn 2R$),
    foxes die if they are hungry 
    ($F \rxn \emptyset$),
    but foxes reproduce if they eat a rabbit 
    ($F+R \rxn 2F$).

\subsection{Empirically the batching algorithm samples from the Gillespie distribution}

\begin{figure}
    \centering
    \begin{subfigure}{0.31\textwidth}
        \centering
        \includegraphics[width=\textwidth]{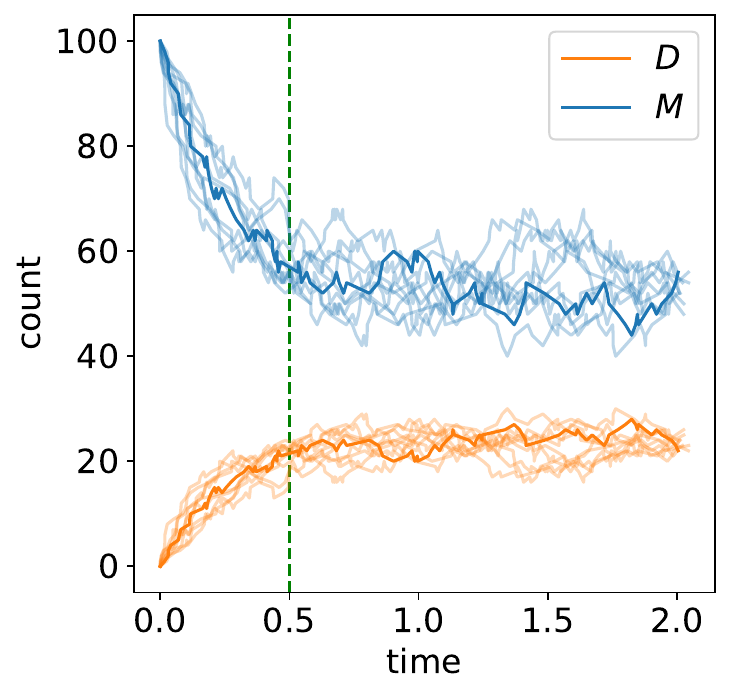}
        \caption{Plot of species counts vs.~time in ten example runs.}
        \label{fig:dimerization_counts_vs_time}
    \end{subfigure}
    \hfill
    \begin{subfigure}{0.68\textwidth}
        \centering
        \includegraphics[width=\textwidth]{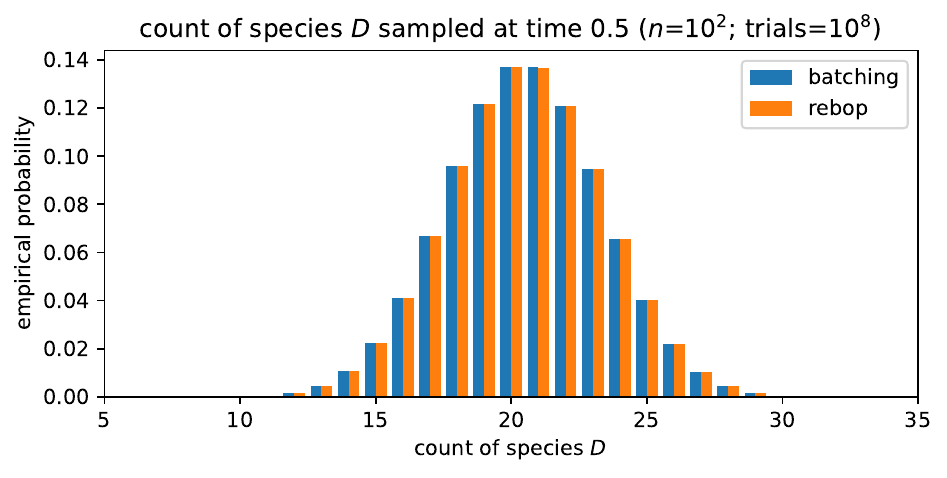}
        \caption{Comparing distribution of batching algorithm to Gillespie (rebop).}
        \label{fig:dimerization_comparison}
    \end{subfigure}
    \caption{
    \Cref{fig:dimerization_counts_vs_time} shows a plot of ten example runs of the reversible dimerization reaction $2M \revrxn_1^1 D$, starting with $\#M = 100$ and $\#D = 0$,
    showing counts vs time for each run.
    The CRN approaches an equilibrium with expected $25$ copies of $D$ and expected $50$ copies of $M$,
    typically reaching there after about 1 unit of time.
    (Though of course the counts bounce around even at equilibrium.)
    \Cref{fig:dimerization_comparison} shows a plot of empirical distributions from rebop and the batching algorithm,
    of the count of $D$ at time 0.5,
    just prior to (likely) convergence,
    where $\E{\#D} \approx 20.5.$
    Here, ``empirical probability'' means the total number of runs in which the given count was the count of $D$ at time $0.5$, divided by the total number of trials.
    }
    \label{fig:dimer_counts_vs_time_and_distribution_comparison}
\end{figure}

\Cref{fig:dimer_counts_vs_time_and_distribution_comparison} demonstrates empirically that the batching algorithm samples from the same distribution as the Gillespie algorithm.
We run both our batching algorithm and rebop on the CRN with the reversible dimerization reaction $2M \revrxn_1^1 D$,
where two monomers $M$ can join to form an unstable dimer $D$,
which can in turn split back into monomers.
For the sake of collecting many samples,
we choose a small initial population size $n = 100$:
start with $\#M = 100$.
Run until time $0.5$,
and measure the count of $D$,
for many trials.
We plot the empirical distribution (number of times $D$ had the given count, divided by the number of trials)
for both our batching algorithm and rebop.

We chose the reversible dimerization CRN $2M \revrxn D$ because,
unlike the Lotka-Volterra CRN
used for other examples,
the dimerization CRN cannot ``go extinct.''
The Lotka-Volterra CRN has the unfortunate property that,
if rabbits die out, then so do foxes eventually, leading to a so-called ``terminal'' configuration in which no reactions are applicable.
Many CRN simulators, such as rebop,
simply hang if asked to simulate to a certain time,
if they go terminal before that time is reached.

Conceptually, this is not difficult to deal with, but in practice it is easier simply to simulate a CRN with no reachable terminal configurations.
Thus we use the CRN $2M \revrxn D$, which cannot go terminal, nor can its molecular count increase without bound.
Nevertheless, 
as with Lotka-Volterra,
it represents an excellent test case for our new algorithm,
since it fundamentally requires a reaction with positive generativity,
the key challenge in adapting the batching algorithm of~\cite{berenbrink2020simulating} to more general CRNs.

\subsection{Batching algorithm has quadratic speedup over Gillespie algorithm}
\label{sec:compare-running-time-batching-gillespie}

\renewcommand{\floatpagefraction}{0.99}%
\begin{figure}
\centering
\includegraphics[width=0.45\linewidth]{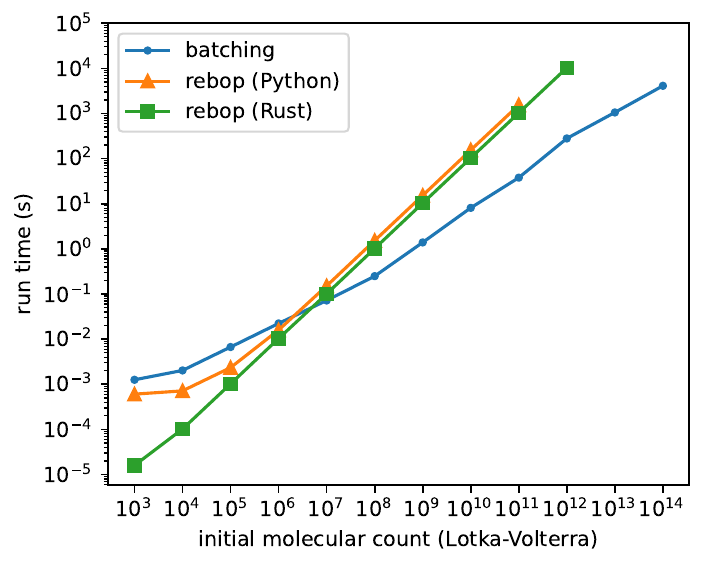}
\includegraphics[width=0.45\linewidth]{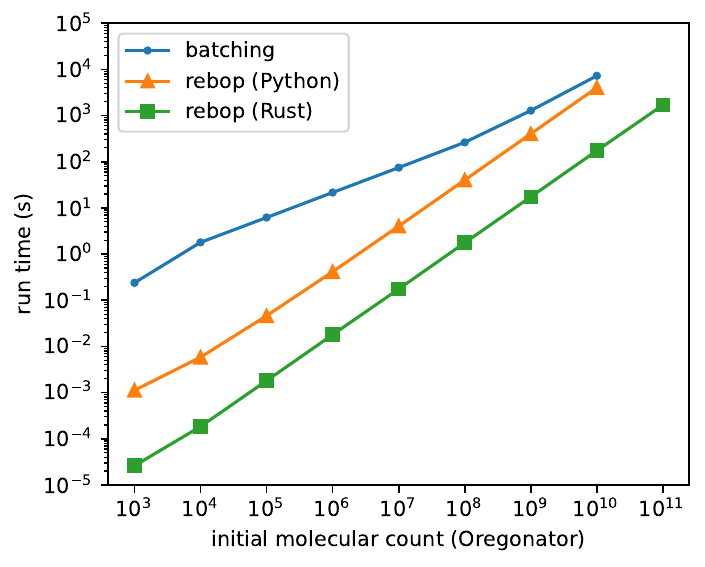}
\caption{
    Runtime scaling of our batching algorithm
    vs.~the Gillespie algorithm, as implemented by the rebop Python package~\cite{rebop},
    and the pure Rust rebop crate~\cite{rebop_rust}. 
    Both Lotka-Volterra reactions (left) and Oregonator reactions (right)
    are simulated on the given initial population size $n$ (with an even split between all species) until time 1.0,
    which corresponds to $\Theta(n)$ total reactions.
    As expected, in both cases rebop shows asymptotic scaling of $\Theta(n)$ time 
    (slope 1 on a log-log plot) 
    compared to scaling of $\Theta(\sqrt{n})$ for batching
    (slope 1/2). 
    However, batching is much more efficient for Lotka-Volterra.
    For exact reactions and rate constants used, see \cref{fig:lotka_volterra_plot_with_passive_reactions}.
}
\label{fig:lotka_volterra_scaling}
\end{figure}

\Cref{fig:lotka_volterra_scaling} shows runtime scaling of our batching algorithm vs.~the Gillespie algorithm, as implemented by the rebop package~\cite{rebop}, which is the fastest Gillespie implementation that we have found.
We tested against both 
the rebop Python package~\cite{rebop},
and the rebop Rust crate~\cite{rebop_rust}
where one can implement the CRN in a pure Rust program.
We show benchmarks for two CRNs, to show how our algorithm's performance relative to the Gillespie algorithm can depend heavily on the CRN being simulated.
For the exact CRNs being simulated, see \cref{fig:lotka_volterra_plot_with_passive_reactions}.

For Lotka-Volterra, batching becomes faster quickly,
and can reasonably simulate much larger population sizes than rebop.
For the Oregonator, 
while it is clear that batching still has an asymptotic advantage,
the additional overhead is sufficiently costly that rebop stays faster until population sizes are too large to cheaply simulate.
We discuss the likely cause of this in \cref{sec:data-passive-reactions}.

The slight nonlinearity at $n = 10^{12}$ for Lotka-Volterra simulation in \Cref{fig:lotka_volterra_scaling} arises due to an implementation issue:
Our way of sampling the length of a collision-free run involves several floating-point operations,
mainly to repeatedly compute the log-gamma function $\ln \Gamma(x)$.
On sufficiently large population sizes,
we found empirically that standard 64-bit double precision floats lacked the precision to compute these numbers.
In particular, our method of inversion sampling samples a float uniformly in the unit interval $u \in [0,1)$.
We compute $\ln u$.
In computing the CDF of the distribution in order to do inversion sampling,
this number $\ln u$ is compared to numbers computed as the difference of log-gamma of much larger numbers,
scaling with population size $n$.
When the log-gamma of those large numbers (recall $\ln \Gamma(n) \approx n \ln n$) is sufficiently large (in this case,
the imprecision causes the inversion sampling to be incorrect, 
because the differences between those large numbers are insufficiently precise to compare meaningfully to $\ln u$.
This issue is fundamental:
the binary search step of our inversion sampling requires us to compare large numbers that may differ by a small amount.

To correct for this,
we use Rust's experimental f128 type that has so-called ``quadruple'' precision,
double that of a standard double (called f64 in Rust).
However, because hardware and OS support for f128 is incomplete,
we had to implement our own slower software implementation of the natural log function and the log-gamma function on f128 values.
This is the source of the observed slowdown when $n \ge 10^{11}.$
In practice, even if we allow these errors through and just use the faster f64 implementation
(also shown in \Cref{fig:lotka_volterra_scaling}),
there does not appear to be any systematic bias that leads the sampled CRN trajectory to appear sufficiently different from the correct Gillespie distribution at many population sizes.
However,
we kept the slow f128 implementation as the standard to ensure correctness.
To avoid using our expensive natural log and log-gamma functions when not needed,
we start each binary search using f64 arithmetic,
and at each step of the search we check whether there is any chance that f128-level precision may change the result.
That is, we check if the potential floating point error (which can be easily calculated) is large enough to change the result of the comparison at that step of the search, 
in which case we switch to using f128.


\todo{DD: This paragraph is somewhat repetitive with early statements.}
More precisely,
all of this happens when we sample the length of a collision-free run.
The algorithm described in the proof of \Cref{lem:sample_coll_efficient} involves summing terms that are the output of the log-gamma function.
These outputs are always represented as f128 values,
which take roughly twice as long to add or subtract as f64 values, accounting for the unconditional slowdown in all population sizes in f128 compared to f64.
However, for sufficiently small values as input to log-gamma,
for efficiency we compute log-gamma using f64 values,
even though we then store that output in an f128 to add the terms.
The additional slowdown for $n \ge 10^{11}$ is because the computation of log-gamma itself begins switching to using f128 values internally,
and this computation is sufficiently complex that the constant factor increase here is more noticeable;
in particular,
computing the natural log function is a significant factor in computing log-gamma, but natural log for f128 is not supported yet in Rust, so we had to write our own software implementation of natural log as well.

Regarding the performance of the rebop Python package versus the rebop pure Rust crate:
the rebop documentation claims 
\begin{displayquote}
\emph{Performance and ergonomics are taken very seriously. For this reason, two independent APIs are provided to describe and simulate reaction networks:}
\begin{itemize}
    \item
    \emph{a macro-based DSL implemented by} [\verb!define_system!], \emph{usually the most efficient, but that requires to compile a rust program;}
    
    \item 
    \emph{a function-based API implemented by the module} [\verb!gillespie!], 
    \emph{also available through Python bindings. This one does not require a rust compilation and allows the system to be defined at run time. It is typically 2 or 3 times slower than the macro DSL, but still faster than all other software tried.}
\end{itemize}
\end{displayquote}

The superior performance of the pure Rust crate (``macro-based DSL'') over the Python API is evident in the difference between the data labeled ``rebop (Python)'' and ``rebop (Rust)'' in \cref{fig:lotka_volterra_scaling}.
Our package, like the rebop Python package, uses a Rust backend with a Python front-end using PyO3 to call Rust from Python.
Since we intend our package as a Python package, which will be much easier to use for most users,
it seems that the fair comparison is to the rebop Python API, rather than the faster pure Rust crate.
Yet,
whether comparing to the rebop Python API or Rust API,
in either case
the linear scaling 
(time $\Theta(n)$ to simulate $n$ reactions)
shown in \Cref{fig:lotka_volterra_scaling} still holds, but the pure Rust implementation of rebop stays superior to our batching algorithm for slightly large population sizes than the rebop Python package.
Nevertheless the batching algorithm is clearly superior to both for Lotka-Volterra when $n \geq 10^7.$

We intend eventually to implement a pure Rust crate for our batching algorithm.
We may see a performance gain as well,
but it is not clear that we can do the same optimizations done by rebop,
which somehow optimizes to get a constant factor speedup based on knowing the reactions at compile time.
In the case of the batching algorithm,
there's no obvious optimization that can be done based on knowing the reactions at compile time.
The constant-factor overhead of rebop's Python API is definitely not merely the overhead of calling from Python,
because the times for large population sizes were measured with a single call to rebop's \texttt{Gillespie.run} method (called from Python, but implemented as a Rust method), which has only a small additive constant overhead to call once, yet the performance of the pure Rust implementation remains a constant factor faster no matter how large $n$.

\subsubsection{Multibatching}
\label{sec:multibatching-discussion}

The batching algorithm of~\cite{berenbrink2020simulating} has an optimized version called ``multibatching'', in which multiple collisions are simulated in a single batch.
This is useful when the time $t_c$ to sample the collision-free run length and simulate a collision is less than the time $t_b$ required to process a batch;
if $t_b \approx t_c \cdot m$ for $m \in \N^+$,
then the ideal tradeoff would be to simulate $m$ collisions for each batch.
This is the main reason it is useful to have a definition of the collision-free run length distribution $\mathbf{coll}(n,r,o,g)$ that allows $r>0$,
even though in this paper we focused on the case $r=0$,
since $r>0$ is how we model the question
``\emph{given that $r$ molecules have already interacted from previous collisions and the batch length up to this point, how many \emph{additional} molecules can be picked before another collision?}'',
and then that number is added to the growing batch length.

In other words, instead of a single step of the algorithm being to process a batch of length $\ell \sim \mathbf{coll}(n,0,o,g)$,
then simulate a single collision,
we instead imagine the following process.
We sample a collision-free run length $\ell_1 \sim \mathbf{coll}(n,0,o,g)$,
then a second collision-free run length 
$\ell_2 \sim \mathbf{coll}(n, (\ell_1 + 1)\cdot (o+g), o, g)$ 
(to model that the first $\ell_1+1$ reactions ($\ell_1$ in the first part of the batch, plus 1 for the collision) turn $(\ell_1 + 1)\cdot (o+g)$ total molecules ``red''.
Then we sample a third collision-free run length 
$\ell_3 \sim \mathbf{coll}(n, (\ell_1 + \ell_2 + 2)\cdot (o+g), o, g)$,
etc.,
up to $\ell_m$.
We then simulate the $m$ interactions involving the $m$ collisions (see~\cite[Section 4]{berenbrink2020simulating} for details),
and a single batch of length $\sum_{i=1}^m \ell_i$.

Although it is possible
to implement multibatching for our generalized batching algorithm,
we have not explored the idea in depth.
Empirically our implementation spends the bulk of its time sampling the collision-free run length,
due mostly to the extra complexity of our definition of a collision, generalized to allow positive generativity.
If this could be optimized to take far less time,
then implementing a similar multibatching approach in our algorithm would become beneficial.

\subsubsection{Passive reactions}
\label{sec:data-passive-reactions}

\begin{figure}
    
    \centering
    \begin{subfigure}{0.99\textwidth}\begin{minipage}{0.2\textwidth}
    \centering
    Lotka Volterra:
    \begin{align*}
        R   &\rxn^1 2R
        \\
        F   &\rxn^1 \emptyset
        \\
        F+R &\rxn^1 2F
    \end{align*}
    \end{minipage}
    \begin{minipage}{0.8\textwidth}
    \includegraphics[width=\linewidth]{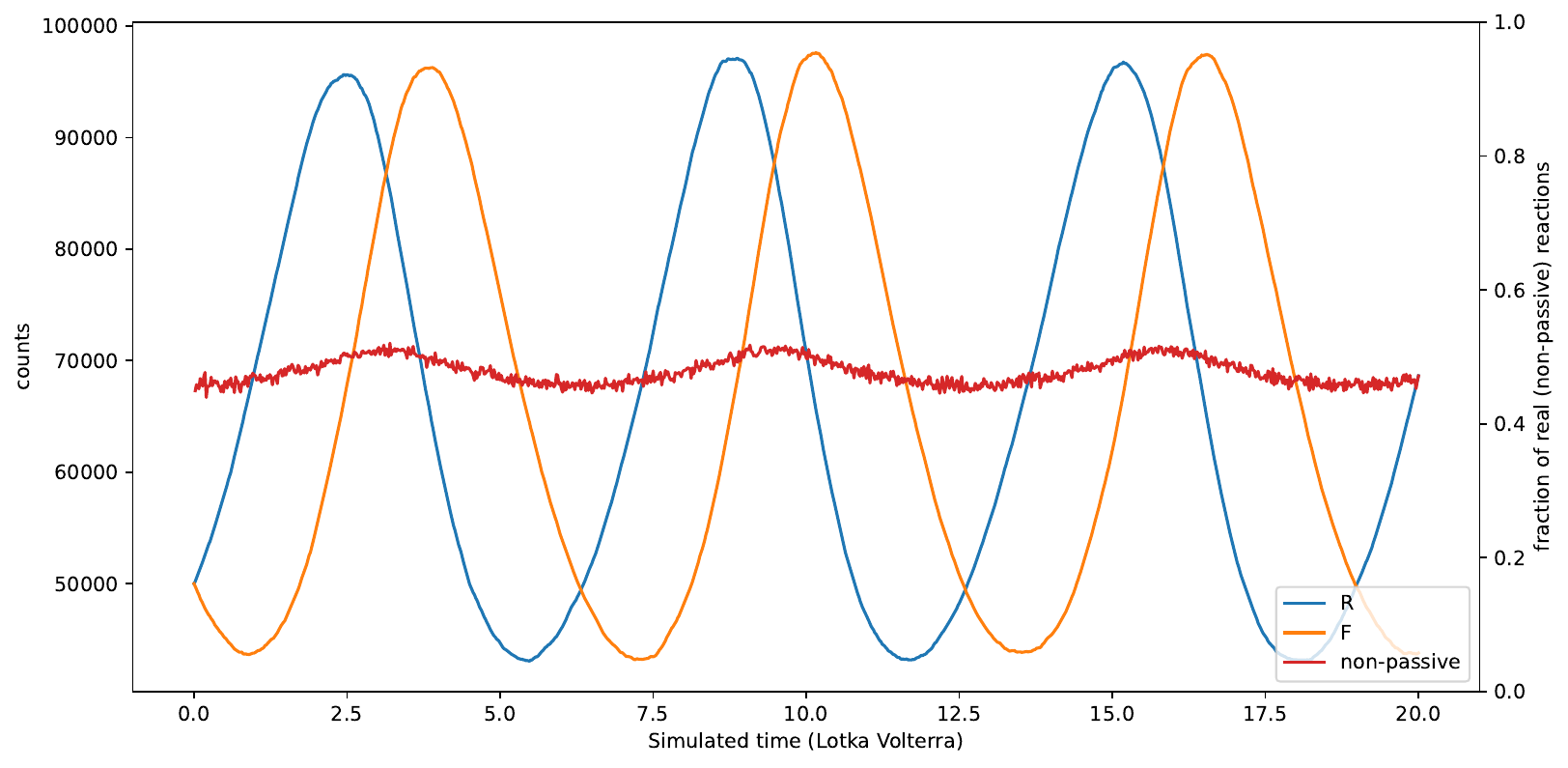}
    \end{minipage}
    \end{subfigure}
    \begin{subfigure}{0.99\textwidth}\begin{minipage}{0.2\textwidth}
    \centering
    R\"ossler:
    \begin{align*}
        X_1   &\rxn^{30} 2X_1
        \\
        2X_1   &\rxn^{0.5} X_1
        \\
        X_1+X_2   &\rxn^{1} 2X_2
        \\
        X_2   &\rxn^{10} \emptyset
        \\
        X_1+X_3   &\rxn^{1} \emptyset
        \\
        X_3   &\rxn^{16.5} 2X_3
        \\
        2X_3   &\rxn^{0.5} X_3        
    \end{align*}
    \end{minipage}
    \begin{minipage}{0.8\textwidth}
    \raggedleft
    \includegraphics[width=0.96\linewidth]{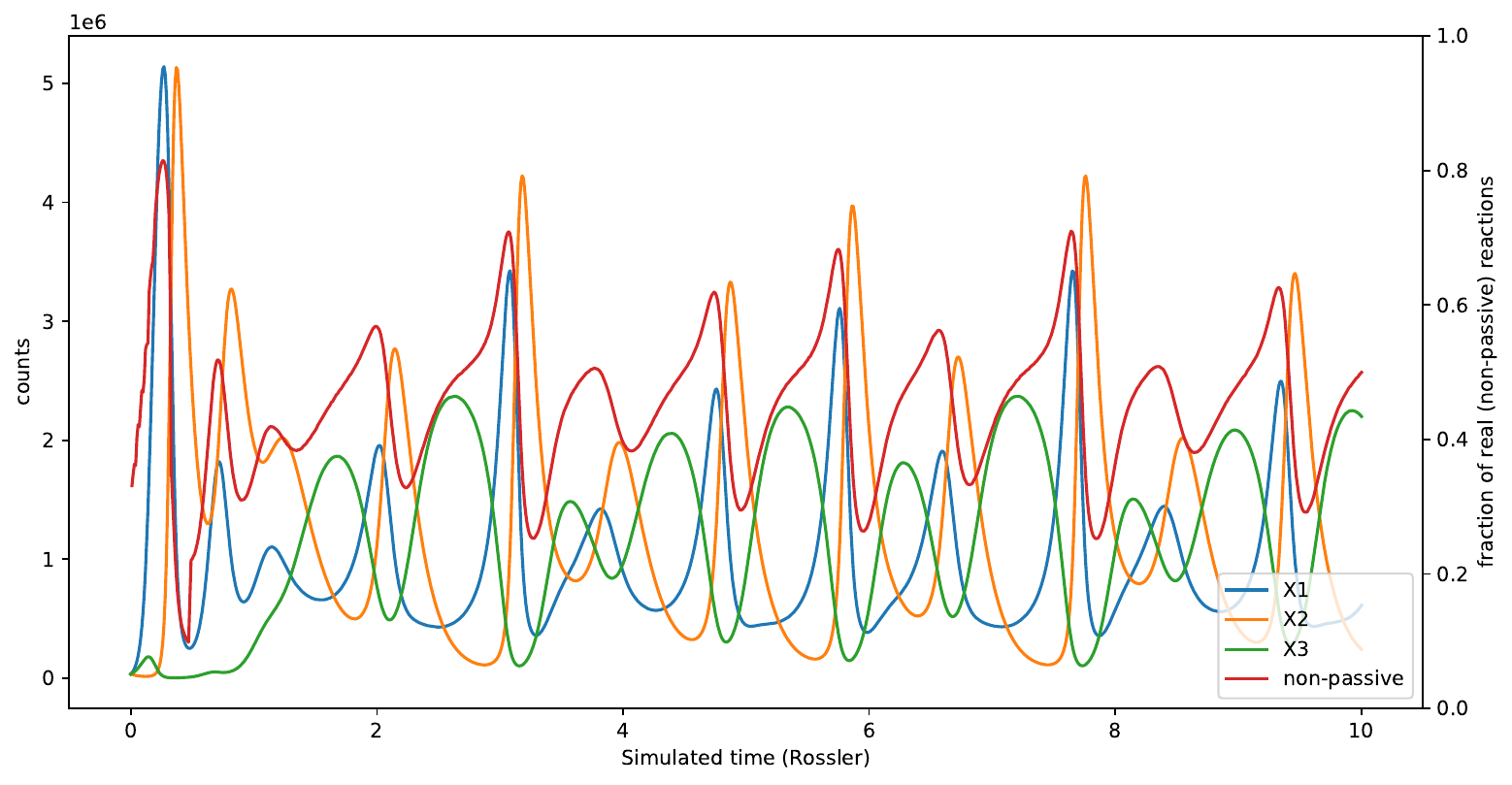}
    \end{minipage}
    \end{subfigure}
    \begin{subfigure}{0.99\textwidth}\begin{minipage}{0.2\textwidth}
    \centering
    Oregonator:
    \begin{align*}
        X_2   &\rxn^{0.1} X_1
        \\
        X_1+X_2   &\rxn^{1000} \emptyset
        \\
        X_1   &\rxn^{520} 2X_1+X_3
        \\
        2X_1   &\rxn^{40} \emptyset
        \\
        X_3   &\rxn^{443.324} X_2
        \\
        X_3   &\rxn^{2.676} \emptyset
    \end{align*}
    \end{minipage}
    \begin{minipage}{0.8\textwidth}
    \includegraphics[width=\linewidth]{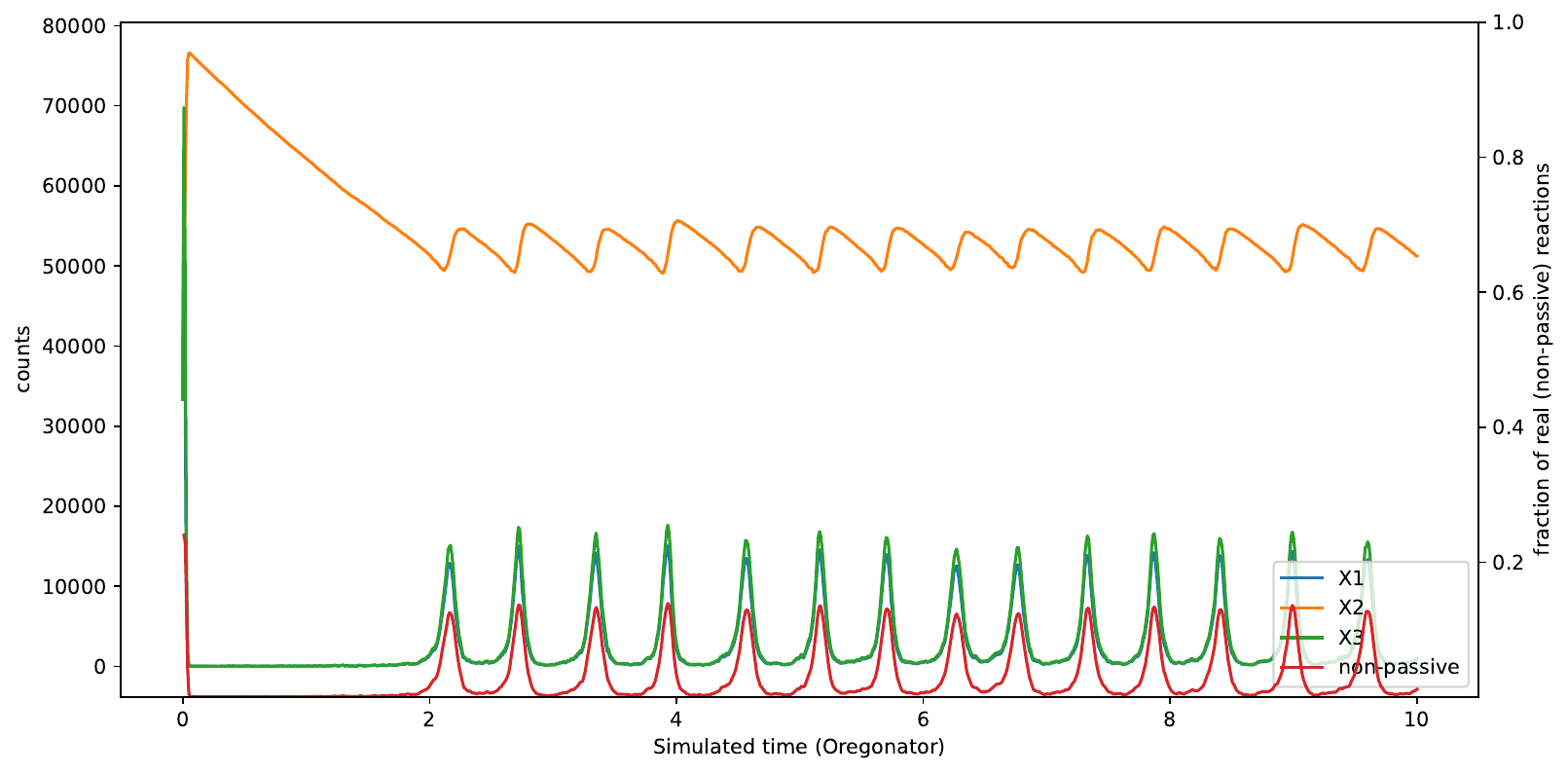}
    \end{minipage}
    \end{subfigure}
    \caption{Plot of counts vs.~time for three CRNs with initial molecular count $n=10^5$, with the fraction of non-passive reactions also shown in red (note separate $y$-axes on right). The first two simulations are not majorly impeded by excessive passive reactions. For the Oregonator, however, the smallest sampled fraction of non-passive reactions is around 0.0003, representing a slowdown factor of $\approx 3000$. During the periodic part, this value oscillates between around 0.002 and 0.1. Rate constants modified from \cite{soloveichik2010dna}. }
    \label{fig:lotka_volterra_plot_with_passive_reactions}
\end{figure}

Recall that a potential source of slowdown for the batching algorithm is simulating many passive reactions
(see \Cref{alg:null-reactions}),
which do not correspond to reactions executed in the original CRN.
The greater the fraction of passive reactions in a batch,
the less progress the batching algorithm makes toward simulating the original CRN.

To measure empirically how many reactions are passive, we counted the fraction of simulated reactions that 
\Cref{fig:lotka_volterra_plot_with_passive_reactions} shows three different CRNs each with initial population size $n=10^5$,
with counts plotted alongside the fraction of reaction sampled at each time point that were not passive,
i.e. that made progress in simulating the CRN.
This comparison shows that, for some CRNs, passive reactions are not a major concern; 
for the Lotka-Volterra oscillators and R\"ossler attractor, typically around half of simulated reactions are passive.
For the Oregonator oscillator, however, we see two issues.
First, during the initial simulation period, more than 99.9\% of simulated reactions are passive,
serving as a major bottleneck.
Second, even during steady-state behavior,
between 90 and 99.5\% of simulated reactions are passive,
making overall simulation much slower.
Indeed, batss simulates the former two CRNs more quickly than rebop even at relatively low population sizes,
whereas rebop is faster for the Oregonator even at large sizes.
We believe these issues occur due to the Oregonator's disparate rate constants,
with the issue being exacerbated during parts of the simulation that have very low counts of some species.

Because $\# K$ influences correction factors in rate constants (see~\Cref{def:adjusted-parameters,alg:crn-transformation}),
changing $\# K$ requires recomputing rate constants.
We avoid changing $K$ on every step because it is computationally prohibitive to recompute the rate constants that frequently.
Currently we set $\#K=n$ (i.e., equal to the number of molecules in the original CRN) when we reset $\#K$,
since this gives the asymptotic behavior that we want of maintaining a $\Omega(1)$ constant fraction of sampled reactions are non-passive.
However, it remains to optimize;
the optimal choice of $\#K$ is $\Theta(n)$,
but choosing the actual value more carefully could potentially reduce the number of null reactions.

\end{toappendix}

%% file: sections/conclusion.tex
\section{Open Questions}

There are two pressing theoretical open questions relating to asymptotic efficiency. The first pertains to slowdown described in \cref{def:efficiency-condition} from probabilistic reactions on certain configurations (e.g., $2L \rxn L + F$ as $\#L \to 0$).\todo{This is only a problem for reversible dimerization when the rate constants are imbalanced, or molecular count is significantly different from volume, so we need to be more specific.}
It is possible that any exact algorithm that chooses reactions by selecting individual reactants will suffer from this slowdown, as it is necessary to ensure correct reaction probabilities in configurations where a high-propensity reaction has one or more low-count reactants. 
There may be some way to utilize the core idea of batching, but choose which reactants comprise a batch in a more clever way.
It may also be possible to choose reactions in some other way than either choosing individual reactions (as in Gillespie) or choosing individual reactants (as in batching).

The second theoretical question pertains to the asymptotic cost of adaptive rejection sampling.
Ideally, in all situations, our algorithm would be able to simulate $\Theta(\sqrt n)$ reactions in $O(\log n)$ time;
however, when sampling from the hypoexponential exactly using adaptive rejection sampling, this only yields optimal efficiency when $\ell \geq n^{5/4}$.
Future work may show how to avoid this by exactly sampling timestamps more efficiently.

Another question is how broadly applicable this algorithm is beyond CRNs. 
It may be possible to adapt the algorithm to other stochastic processes that 
select elements from some set at each iteration to undergo a ``transition'', but are not modeled exactly by the Gillespie algorithm, 
for example surface CRNs~\cite{clamons2020programming} or tile displacement systems~\cite{sarraf2023modular,winfree2023two},
which are both chemical models that account for geometrical arrangement of some chemical species that interact with each other.

Practically, it may also be useful to develop approximate simulation algorithms that derive their logic from our exact simulation algorithm.
For example, by skipping the step of sampling collision-free run length and instead simply choosing the run length to be its expected value, one can avoid the numerical precision issues that our implementation has.
It seems feasible that such practical algorithms may still be amenable to theoretical analysis of bound approximation accuracy as in e.g. \cite{soloveichik2009robust}.


%% file: main.bib
@article{press2007section17,
  title={Section 17.1 {R}unge-{K}utta Method},
  author={Press, William H and Flannery, BP and Teukolsky, SA and Vetterling, WT},
  journal={Numerical Recipes: The Art of Scientific Computing},
  year={2007},
  publisher={Cambridge University Press New York, NY, USA}
}

@article{devroye2006nonuniform,
  title={Nonuniform random variate generation},
  author={Devroye, Luc},
  journal={Handbooks in operations research and management science},
  volume={13},
  pages={83--121},
  year={2006},
  publisher={Elsevier}
}

@book{von2003modern,
  title={Modern Computer Algebra},
  author={Von Zur Gathen, Joachim and Gerhard, J{\"u}rgen},
  year={2003},
  publisher={Cambridge University Press}
}

@book{cormen2022introduction,
  title={Introduction to algorithms},
  author={Cormen, Thomas H and Leiserson, Charles E and Rivest, Ronald L and Stein, Clifford},
  year={2022},
  publisher={MIT press}
}

@article{bowman2004estimation,
  title={Estimation: Method of moments},
  author={Bowman, Kimiko O and Shenton, LR},
  journal={Encyclopedia of Statistical Sciences},
  volume={3},
  year={2004},
  publisher={Wiley Online Library}
}

@book{feinberg2019foundations,
  title={Foundations of Chemical Reaction Network Theory},
  author={Feinberg, Martin},
  volume={202},
  year={2019},
  publisher={Springer}
}

@article{lotka1910contribution,
  title={Contribution to the theory of periodic reactions},
  author={Lotka, Alfred J},
  journal={The Journal of Physical Chemistry},
  volume={14},
  number={3},
  pages={271--274},
  year={1910},
  publisher={ACS Publications}
}

@book{volterra1926variazioni,
  title={Variazioni e fluttuazioni del numero d'individui in specie animali conviventi},
  author={Volterra, Vito},
  year={1926},
  publisher={Societ{\`a} anonima tipografica ``Leonardo da Vinci'''}
}

@article{sarraf2023modular,
  title={Modular reconfiguration of {DNA} origami assemblies using tile displacement},
  author={Sarraf, Namita and Rodriguez, Kellen R and Qian, Lulu},
  journal={Science Robotics},
  volume={8},
  number={77},
  pages={eadf1511},
  year={2023},
  publisher={American Association for the Advancement of Science}
}

@article{winfree2023two,
  title={Two-dimensional tile displacement can simulate cellular automata},
  author={Winfree, Erik and Qian, Lulu},
  journal={arXiv preprint arXiv:2301.01929},
  year={2023}
}

@article{clamons2020programming,
  title={Programming and simulating chemical reaction networks on a surface},
  author={Clamons, Samuel and Qian, Lulu and Winfree, Erik},
  journal={Journal of the Royal Society Interface},
  volume={17},
  number={166},
  pages={20190790},
  year={2020},
  publisher={The Royal Society}
}

@article{cartwright1978refinement,
  title={A refinement of the arithmetic mean-geometric mean inequality},
  author={Cartwright, DI and Field, MJ},
  journal={Proceedings of the American Mathematical Society},
  pages={36--38},
  year={1978},
  publisher={JSTOR}
}

@article{spouge1994computation,
author = {Spouge, John L.},
title = {Computation of the Gamma, Digamma, and Trigamma Functions},
journal = {SIAM Journal on Numerical Analysis},
volume = {31},
number = {3},
pages = {931-944},
year = {1994},
doi = {10.1137/0731050},
URL = {https://doi.org/10.1137/0731050}
}

@misc{rebop,
  title = {rebop {P}ython package},
  howpublished = {\url{https://pypi.org/project/rebop/}\\ \url{https://github.com/Armavica/rebop}}
}

@misc{rebop_rust,
  title = {rebop {R}ust crate},
  howpublished = {\url{https://docs.rs/rebop/latest/rebop/}\\ \url{https://github.com/Armavica/rebop}}
}

@misc{mpmath,
  title="mpmath {P}ython package",
  url={https://mpmath.org/}
}

@inproceedings{ppsim,
  author = "David Doty and Eric Severson",
  title = "ppsim: A software package for efficiently simulating and visualizing population protocols",
  booktitle = {CMSB 2021: Proceedings of the 19th International Conference on Computational Methods in Systems Biology},
  year =	{2021},
  location = {Bordeaux, France},
  pages="245--253",
  url = "https://arxiv.org/abs/2105.04702"
}

@article{kurtz1972relationship,
  title={The relationship between stochastic and deterministic models for chemical reactions},
  author={Thomas G. Kurtz},
  journal={The Journal of Chemical Physics},
  volume={57},
  number={7},
  pages={2976--2978},
  year={1972},
  publisher={The American Institute of Physics}
}

@article{cardelli2016stochastic,
title = {Stochastic analysis of Chemical Reaction Networks using Linear Noise Approximation},
journal = {Biosystems},
volume = {149},
pages = {26-33},
year = {2016},
note = {Selected papers from the Computational Methods in Systems Biology 2015 conference},
issn = {0303-2647},
doi = {https://doi.org/10.1016/j.biosystems.2016.09.004},
url = {https://www.sciencedirect.com/science/article/pii/S0303264716302039},
author = {Luca Cardelli and Marta Kwiatkowska and Luca Laurenti}
}

@article{slepoy2008constant,
  title={A constant-time kinetic {M}onte {C}arlo algorithm for simulation of large biochemical reaction networks},
  author={Slepoy, Alexander and Thompson, Aidan P and Plimpton, Steven J},
  journal={The Journal of Chemical Physics},
  volume={128},
  number={20},
  pages={05B618},
  year={2008},
  publisher={American Institute of Physics}
}

@article{gibson2000efficient,
  title={Efficient exact stochastic simulation of chemical systems with many species and many channels},
  author={Gibson, Michael A and Bruck, Jehoshua},
  journal={The Journal of Physical Chemistry A},
  volume={104},
  number={9},
  pages={1876--1889},
  year={2000},
  publisher={ACS Publications}
}

@article{schwachheim1969algorithm,
author = {Schwachheim, Georges},
title = {Algorithm 349: polygamma functions with arbitrary precision [S14]},
year = {1969},
issue_date = {April 1969},
publisher = {Association for Computing Machinery},
address = {New York, NY, USA},
volume = {12},
number = {4},
issn = {0001-0782},
url = {https://doi.org/10.1145/362912.362928},
doi = {10.1145/362912.362928},
journal = {Communications of the ACM},
month = apr,
pages = {213–214},
numpages = {2}
}

@article{bernardo1976algorithm,
  title={Algorithm AS 103: Psi (digamma) function},
  author={Bernardo, Jose M},
  journal={Journal of the Royal Statistical Society. Series C (Applied Statistics)},
  volume={25},
  number={3},
  pages={315--317},
  year={1976},
  publisher={JSTOR}
}

@article{lozier2003nist,
  title={{NIST} digital library of mathematical functions},
  author={Lozier, Daniel W},
  journal={Annals of Mathematics and Artificial Intelligence},
  volume={38},
  pages={105--119},
  year={2003},
  publisher={Springer},
  url={https://dlmf.nist.gov/5.15}
}

@book{abramowitz1965handbook,
  title={Handbook of mathematical functions: with formulas, graphs, and mathematical tables},
  author={Abramowitz, Milton and Stegun, Irene A},
  volume={55},
  year={1965},
  publisher={Courier Corporation}
}

@misc{qureshi2018analyticcomputationsdigammafunction,
      title={Analytic computations of digamma function using some new identities}, 
      author={M. I. Qureshi and Mohd Shadab},
      year={2018},
      eprint={1806.07948},
      archivePrefix={arXiv},
      primaryClass={math.CA},
      url={https://arxiv.org/abs/1806.07948}, 
}

@article{AngluinADFP2006,
title="Computation in networks of passively mobile finite-state sensors",
author="Dana Angluin and James Aspnes and Zo{\"e} Diamadi and Michael J. Fischer and Ren\'e Peralta",
journal="Distributed Computing",
month=mar,
year=2006,
pages={235--253}
}

@article{guldberg1864studies,
  title = "Studies Concerning Affinity",
  author = "Cato M. Guldberg and Peter Waage",
  journal = "Forhandlinger: Videnskabs-Selskabet i Christinia. Norwegian Academy of Science and Letters",
  year = 1864,
  volume=35,
  note={English translation in~\cite{waage1986studies}}
}

@article{waage1986studies,
  title={Studies concerning affinity},
  author={Waage, Peter and Gulberg, Cato Maximilian},
  journal={Journal of chemical education},
  volume={63},
  number={12},
  pages={1044},
  year={1986},
  publisher={ACS Publications},
  note={English translation; original paper is~\cite{guldberg1864studies}}
}

@inproceedings{berenbrink2020simulating,
  author = {Berenbrink, Petra and Hammer, David and Kaaser, Dominik and Meyer, Ulrich and Penschuck, Manuel and Tran, Hung},
  title =	{{Simulating Population Protocols in Sub-Constant Time per Interaction}},
  booktitle =	{ESA 2020: 28th Annual European Symposium on Algorithms},
  pages =	{16:1--16:22},
  ISBN =	{978-3-95977-162-7},
  ISSN =	{1868-8969},
  year =	{2020},
  volume =	{173},
  URL =		{https://drops.dagstuhl.de/entities/document/10.4230/LIPIcs.ESA.2020.16},
  URN =		{urn:nbn:de:0030-drops-128827},
  doi =		{10.4230/LIPIcs.ESA.2020.16},
}

@article{gillespie1977exact,
  title={Exact stochastic simulation of coupled chemical reactions},
  author={Gillespie, Daniel T},
  journal={The Journal of Physical Chemistry},
  volume={81},
  number={25},
  pages={2340--2361},
  year={1977},
  publisher={ACS Publications}
}

@InProceedings{lathrop2020population,
  author =	{Lathrop, James I. and Lutz, Jack H. and Lutz, Robyn R. and Potter, Hugh D. and Riley, Matthew R.},
  title =	{{Population-Induced Phase Transitions and the Verification of Chemical Reaction Networks}},
  booktitle =	{26th International Conference on DNA Computing and Molecular Programming (DNA 26)},
  pages =	{5:1--5:17},
  series =	{Leibniz International Proceedings in Informatics (LIPIcs)},
  ISBN =	{978-3-95977-163-4},
  ISSN =	{1868-8969},
  year =	{2020},
  volume =	{174},
  editor =	{Geary, Cody and Patitz, Matthew J.},
  publisher =	{Schloss Dagstuhl -- Leibniz-Zentrum f{\"u}r Informatik},
  address =	{Dagstuhl, Germany},
  URL =		{https://drops.dagstuhl.de/entities/document/10.4230/LIPIcs.DNA.2020.5},
  URN =		{urn:nbn:de:0030-drops-129583},
  doi =		{10.4230/LIPIcs.DNA.2020.5}
}

@InProceedings{mayr2014framework,
author="Mayr, Ernst W.
and Weihmann, Jeremias",
editor="Ciardo, Gianfranco
and Kindler, Ekkart",
title="A Framework for Classical Petri Net Problems: Conservative Petri Nets as an Application",
booktitle="Application and Theory of Petri Nets and Concurrency",
year="2014",
publisher="Springer International Publishing",
address="Cham",
pages="314--333",
isbn="978-3-319-07734-5"
}

@article{fanti2020communication,
  title={Communication cost of consensus for nodes with limited memory},
  author={Fanti, Giulia and Holden, Nina and Peres, Yuval and Ranade, Gireeja},
  journal={Proceedings of the National Academy of Sciences},
  volume={117},
  number={11},
  pages={5624--5630},
  year={2020},
  publisher={National Acad Sciences}
}

@article{soloveichik2010dna,
author = {David Soloveichik  and Georg Seelig  and Erik Winfree },
title = {{DNA} as a universal substrate for chemical kinetics},
journal = {Proceedings of the National Academy of Sciences},
volume = {107},
number = {12},
pages = {5393-5398},
year = {2010},
doi = {10.1073/pnas.0909380107},
URL = {https://www.pnas.org/doi/abs/10.1073/pnas.0909380107},
eprint = {https://www.pnas.org/doi/pdf/10.1073/pnas.0909380107}}

@article{johnson2019verifying,
  title={Verifying chemical reaction network implementations: A bisimulation approach},
  author={Johnson, Robert and Dong, Qing and Winfree, Erik},
  journal={Theoretical Computer Science},
  volume={765},
  pages={3--46},
  year={2019},
  publisher={Elsevier}
}

@article{shin2019verifying,
  title={Verifying chemical reaction network implementations: A pathway decomposition approach},
  author={Shin, Seung Woo and Thachuk, Chris and Winfree, Erik},
  journal={Theoretical Computer Science},
  volume={765},
  pages={67--96},
  year={2019},
  publisher={Elsevier}
}

@inproceedings{cardelli2015forward,
  title={Forward and Backward Bisimulations for Chemical Reaction Networks},
  author={Luca Cardelli and Mirco Tribastone and Max Tschaikowski and Andrea Vandin},
  booktitle={CONCUR 2015},
  year={2015}
}

@article{gillespie2001approximate,
  title={Approximate accelerated stochastic simulation of chemically reacting systems},
  author={Gillespie, Daniel T},
  journal={The Journal of Chemical Physics},
  volume={115},
  number={4},
  pages={1716--1733},
  year={2001},
  publisher={American Institute of Physics}
}

@article{cao2006efficient,
  title={Efficient step size selection for the tau-leaping simulation method},
  author={Cao, Yang and Gillespie, Daniel T and Petzold, Linda R},
  journal={The Journal of Chemical Physics},
  volume={124},
  number={4},
  year={2006},
  publisher={AIP Publishing}
}

@article{rathinam2007reversible,
  title={Reversible-equivalent-monomolecular tau: A leaping method for “small number and stiff” stochastic chemical systems},
  author={Rathinam, Muruhan and El Samad, Hana},
  journal={Journal of Computational Physics},
  volume={224},
  number={2},
  pages={897--923},
  year={2007},
  publisher={Elsevier}
}

@article{soloveichik2009robust,
  title={Robust stochastic chemical reaction networks and bounded tau-leaping},
  author={Soloveichik, David},
  journal={Journal of Computational Biology},
  volume={16},
  number={3},
  pages={501--522},
  year={2009},
  publisher={Mary Ann Liebert, Inc. 2 Madison Avenue Larchmont, NY 10538 USA}
}

@article{gillespie2007stochastic,
  title={Stochastic simulation of chemical kinetics},
  author={Gillespie, Daniel T},
  journal={Annu. Rev. Phys. Chem.},
  volume={58},
  number={1},
  pages={35--55},
  year={2007},
  publisher={Annual Reviews}
}

@phdthesis{johnston2012topics,
  title={Topics in chemical reaction network theory},
  author={Johnston, Matthew},
  year={2012},
  school={University of Waterloo}
}

@article{weber1993finite,
  title={Finite-time blow-up in reaction-diffusion equations},
  author={Weber, RO and Barry, SI},
  journal={Mathematical and computer modelling},
  volume={18},
  number={10},
  pages={163--168},
  year={1993},
  publisher={Elsevier}
}

@article{cai2009efficient,
  title={Efficient exact and K-skip methods for stochastic simulation of coupled chemical reactions},
  author={Cai, Xiaodong and Wen, Ji},
  journal={The Journal of Chemical Physics},
  volume={131},
  number={6},
  year={2009},
  publisher={AIP Publishing}
}

@article{mjolsness2009exact,
  title={An exact accelerated stochastic simulation algorithm},
  author={Mjolsness, Eric and Orendorff, David and Chatelain, Philippe and Koumoutsakos, Petros},
  journal={The Journal of Chemical Physics},
  volume={130},
  number={14},
  year={2009},
  publisher={AIP Publishing}
}

@book{ross2014introduction,
  title={Introduction to Probability Models},
  author={Ross, Sheldon M},
  year={2014},
  publisher={Academic press}
}

@article{gilks1992adaptive,
  title={Adaptive rejection sampling for {G}ibbs sampling},
  author={Gilks, Walter R and Wild, Pascal},
  journal={Journal of the Royal Statistical Society: Series C (Applied Statistics)},
  volume={41},
  number={2},
  pages={337--348},
  year={1992},
  publisher={Wiley Online Library}
}

@article{wild1993algorithm,
  title={Algorithm AS 287: Adaptive rejection sampling from log-concave density functions},
  author={Wild, Pascal and Gilks, WR},
  journal={Journal of the Royal Statistical Society. Series C (Applied Statistics)},
  volume={42},
  number={4},
  pages={701--709},
  year={1993},
  publisher={JSTOR}
}

@inproceedings{stadlober1989ratio,
  title={Ratio of uniforms as a convenient method for sampling from classical discrete distributions},
  author={Stadlober, Ernst},
  booktitle={Proceedings of the 21st conference on Winter simulation},
  pages={484--489},
  year={1989}
}

@ARTICLE{thanh2017efficient,
  author={Thanh, Vo Hong and Zunino, Roberto and Priami, Corrado},
  journal={IEEE/ACM Transactions on Computational Biology and Bioinformatics}, 
  title={Efficient Constant-Time Complexity Algorithm for Stochastic Simulation of Large Reaction Networks}, 
  year={2017},
  volume={14},
  number={3},
  pages={657-667},
  doi={10.1109/TCBB.2016.2530066}}
